\pdfoutput=1
\documentclass[sn-mathphys-ay]{sn-jnl}

\usepackage{graphicx}
\usepackage{multirow}
\usepackage{amsmath,amssymb,amsfonts}
\usepackage{amsthm}
\usepackage{mathrsfs}
\usepackage[title]{appendix}
\usepackage{xcolor}
\usepackage{textcomp}
\usepackage{manyfoot}
\usepackage{booktabs}
\usepackage{algorithm}
\usepackage{algorithmicx}
\usepackage{algpseudocode}
\usepackage{listings}
\usepackage{tikz}
\usetikzlibrary{fit}
\usepackage{enumitem}

\theoremstyle{thmstyleone}
\newtheorem{theorem}{Theorem}[section]
\newtheorem{proposition}[theorem]{Proposition}
\newtheorem{lemma}[theorem]{Lemma}
\newtheorem{corollary}[theorem]{Corollary}
\theoremstyle{thmstyletwo}
\newtheorem{definition}[theorem]{Definition}
\newtheorem{assumption}[theorem]{Assumption}
\newtheorem{condition}[theorem]{Condition}
\newtheorem{criterion}[theorem]{Criterion}

\newtheorem{example}[theorem]{Example}
\theoremstyle{thmstylethree}
\newtheorem{remark}[theorem]{Remark}

\newcommand{\X}{\mathcal{X}}
\newcommand{\G}{\mathcal{G}}
\newcommand{\M}{\mathcal{M}}
\newcommand{\Cset}{\mathcal{C}}
\newcommand{\Fset}{\mathcal{F}}
\newcommand{\R}{\mathbb{R}}
\newcommand{\Prob}{\mathbb{P}}
\newcommand{\E}{\mathbb{E}}
\newcommand{\ev}{\ell}
\newcommand{\evp}{\ev^{\scriptscriptstyle+}}

\begin{document}

\title[Does the grand coalition form?]{Does the grand coalition form?
Persistence, arrival, and the role of the sharing rule in a dynamic process of nested binding agreements}

\author[1]{\fnm{Jobst} \sur{Heitzig}}

\affil[1]{\orgname{Potsdam Institute for Climate Impact Research},
\orgaddress{\city{Potsdam}, \country{Germany}}, \texttt{heitzig@pik-potsdam.de}\vspace{-10mm}}

\abstract{We study a dynamic coalition-formation process in the tradition of \citet{KonishiRay2003}:
players repeatedly form and dissolve binding agreements,
evaluate states by discounted long-term expected payoffs,
and hold self-confirming beliefs about the process.
States and payoff sharing follow \citet{HeitzigKornek2018}:
a state is a hierarchy of nested agreements, and the members of a new agreement share the surplus it generates,
measured against the state without that agreement.
All payoff assumptions are structural.

We prove that every grand state ever reached is absorbing, and that every absorbing state is grand,
for every discount factor.
A grand state is actually reached, almost surely, for small discount factors, for three players, and,
at every discount factor, whenever distributional stakes are smaller than each player's share of the efficiency gain.
Otherwise the process can fail only by cycling for ever among non-grand states.
We give exact necessary conditions on such a cycle, decidable for a fixed candidate cycle by linear programming,
and exhibit, under an earlier and weaker notion of profitability, a four-player payoff structure whose only closed class is a cycle of two pairs forming and dissolving alternately.
Under the present definition, and under either termination rule, no equilibrium traverses a cycle on a fixed schedule:
somebody always reaches a state they would rather not leave, and the axioms give them the floor.
Whether arrival can fail by cycling at random is open.
The axioms are not merely postulated: we exhibit a bargaining game---proposal, amendment by substitutes
voted on by their own signatories, final unanimity, and an arbitrarily small delay on failure---whose
equilibria satisfy them as the delay vanishes, and which settles each period in its first round.}

\maketitle

\section{Introduction}\label{sec:intro}

\subsection{The question}\label{sec:question}

Whether a group of players ends up cooperating fully, or fragments into several competing blocs,
is the central question of coalition-formation theory.
Static concepts answer it by asking which structures are stable; dynamic concepts,
following \citet{KonishiRay2003} and the theory of binding agreements of \citet{RayVohra1997,RayVohra1999},
instead specify a stochastic process on coalition structures and ask where it goes.
The dynamic formulation is the more demanding one,
because a structure that no single move improves upon may still fail to be reached,
and a structure that is reached may still be left again.
The two halves of the question are therefore \emph{persistence}---once the grand coalition has formed,
does it survive?---and \emph{arrival}---is it formed at all?
A theory that answers only the first has not shown that full cooperation occurs;
a theory that answers only the second has not shown that it lasts.

We work in the Konishi--Ray framework: in each period some set of players may change the current state;
a player agrees to a change only if it raises her discounted long-term expected payoff given her beliefs about the continuation;
and beliefs are required to be self-confirming,
so that the transition probabilities implied by rational behaviour coincide with the ones players hold.

Within that framework we adopt the state space and the payoff-sharing structure of \citet{HeitzigKornek2018},
which builds on \citet{Heitzig2011}, as a concrete and plausible specification.
Three features matter:
\begin{enumerate}[label=(\roman*)]
\item \emph{States are hierarchies.}
An agreement may be signed by individual players or by coalitions that have signed earlier,
so a state is a family of subsets of the player set,
any two of which are disjoint or nested (a \emph{laminar} family).
The nesting is not a full history:
it records which of the agreements \emph{currently in force} were signed before which,
since an agreement contains precisely those agreements that were already in force when it was signed,
but it records nothing about agreements that have since been terminated,
and nothing about the relative timing of two disjoint agreements.
\item \emph{Moves are merges and chain-rule terminations.}
Any group of top-level coalitions may sign a new overarching agreement; and any agreement may be terminated,
which also terminates every agreement containing it,
since the latter was signed by a party that no longer exists.
\item \emph{Surplus is shared among the signatories of a new agreement.}
When an agreement is signed,
its signatories divide the surplus it generates---measured against the state in which that agreement,
and only that agreement, is missing---in some strictly positive proportions.
That the shares are the fixed weights of a weighted Nash bargaining solution \citep{Kalai1977,KalaiSamet1987},
as in \citet{HeitzigKornek2018}, is a special case that we invoke only occasionally,
and never in the persistence results.
\end{enumerate}
We do \emph{not} adopt the partition function of \citet{HeitzigKornek2018},
i.e.\ the underlying cost--benefit model that produces the numbers;
all payoff assumptions used below are structural,
and they are stated explicitly in Section~\ref{sec:payoffs}.
Nothing from the literature is used in any proof.

\begin{example}[Provenance of the running example]\label{ex:carbon}
Example~\ref{ex:running} is the motivating application of \citet{HeitzigKornek2018}.
Players are countries or regions;
an agreement is a linked emissions-trading market whose members coordinate their caps;
and the grand coalition is a global carbon market with a first-best cap.
The classical literature on international environmental agreements \citep{CarraroSiniscalco1993,Barrett1994} finds small stable coalitions when only one coalition may form and members may leave freely;
allowing several coalitions, nesting and termination changes that picture,
and whether full cooperation is eventually reached is exactly the question studied here.
Cap coordination without a coalition structure is analysed in \citet{Helm2003}.
We use this reading only for intuition; no result below depends on it.
\end{example}

The question we study is whether the grand coalition must form eventually.
We consider throughout the variant in which every agreement coordinates its members' actions immediately,
so that a state is a hierarchy of fully coordinated coalitions.

\subsection{Related models}\label{sec:related}

The dynamic study of coalition formation splits, for our purposes, into three lines, and it is worth
being explicit about which questions each of them answers, because our two questions---persistence
and arrival---are not the questions the first two lines ask.

\emph{Farsighted stability.}
\citet{Chwe1994}, \citet{Xue1998}, \citet{RayVohra2015} and \citet{DuttaVohra2017}, and for networks
\citet{HMV2009}, define stability directly on a dominance relation: a state is stable if no coalition
can begin a chain of moves ending at a state all of its members prefer.
These are static concepts about a dynamic idea.
They carry no discount factor, no probabilities and no notion of a path, and they answer the question
``which states can rest''.
Our persistence results are the dynamic counterpart of that question, and they turn out to be the
easy half.
The arrival question has no static counterpart at all, and the counterexample of
Section~\ref{sec:cycles} shows why one is needed: each of the four states of its cycle is
unremarkable on its own, and what fails is a property of the path.

\emph{Protocol-based bargaining.}
\citet{Bloch1996}, \citet{RayVohra1999}, \citet{Gomes2005}, \citet{GomesJehiel2005} and
\citet{HyndmanRay2007} specify a bargaining protocol: who may propose what, when, and to whom.
This buys sharp predictions and, in several of these papers, results on whether long-run outcomes are
efficient.
It also puts a great deal of structure on the answer, since the protocol is rarely observable and the
predictions are sensitive to it.
Following \citet{KonishiRay2003}, we impose no protocol: in each period any group that could
implement a move may consider doing so, and the only discipline is that a realised move be profitable
and undominated for the players whose approval it needs.
That is weaker, so our positive results are correspondingly more robust, and our negative result---a
process that cycles for ever---is harder to dismiss as an artefact of the order of play.

\emph{Dynamic processes with discounting.}
\citet{KonishiRay2003} is the closest relative, and the framework we adopt.
Three things differ.
First, a state here is a \emph{hierarchy} rather than a partition: agreements may be signed by
coalitions that formed earlier, and payoffs depend on which sub-agreements were already in force when
a given agreement was signed, because that is the reference state against which its surplus is
measured.
Two states with the same partition of the players can therefore pay differently, which is what the
relative shares $\psi$ of Section~\ref{sec:decomp} measure, and it is the whole source of the
distributional conflict that drives the arrival problem.
Second, dissolution obeys the chain rule: terminating an agreement terminates every agreement
containing it, because the latter was signed by a party that no longer exists.
Third, the surplus-sharing rule is a primitive of the model rather than an outcome of a protocol.
Section~\ref{sec:pivot} shows that this third choice is not innocuous: it decides whether the limit
$\delta\to1$ is the tractable case or the intractable one.

\emph{Implementation and the Nash program.}
Axioms \ref{E1}--\ref{E3} restrict a process in the way a cooperative solution concept restricts an
outcome, and the natural objection is that no negotiation need behave that way.
Answering it is the business of the Nash program, and of the literature on non-cooperative
foundations for cooperative concepts: random-proposer bargaining in the manner of
\citet{BaronFerejohn1989} and \citet{Okada1996}, rejector-proposer bargaining as in
\citet{ChatterjeeDuttaRaySengupta1993}, and core implementation as in \citet{PerryReny1994};
\citet{Rogna2019} surveys the protocols and what each of them delivers.
Theorem~\ref{thm:implement} belongs to that tradition but implements a different kind of object.
Those results implement a set of \emph{allocations}---the core, a value---as the equilibrium outcomes
of a game.
What is implemented here is a set of \emph{transition kernels}: the object is the whole dynamic
process, and the axioms constrain which moves may occur at every state, not what anybody ends up
with.
The device that makes this possible is the amendment stage, which turns domination from a
stipulation into the outcome of a vote, and the third structural feature above is what makes it
necessary: since the sharing rule is a primitive, the protocol has nothing to bargain over except
which move happens, so proposals must be moves and objections must be moves as well.

Finally, the application that motivates the state space is the linking of carbon markets, where the
relevant cooperative benchmark is the one studied by \citet{ChanderTulkens1997} and, in the
non-cooperative direction, by \citet{Helm2003}; the classical negative results on the size of stable
coalitions are \citet{CarraroSiniscalco1993} and \citet{Barrett1994}.
Nothing from any of these is used below.

\subsection{On the discount factor}\label{sec:delta}

Players are far-sighted in this model by construction, and irrespective of $\delta$:
each evaluates a state by the payoff stream she expects along the whole future path of the process,
anticipating the moves that she and others will make in response to the current state.
This is a statement about what players take into account, not about how much they care for the future.
How much they care is a separate matter, governed by the discount factor $\delta \in (0,1)$.

The discount factor aggregates three distinct things.
The first is the pure rate of time preference $\rho$ of the players.
The second is the length $\Delta t$ of one period,
i.e.\ of the interval between two opportunities to change the state,
so that a pure discount contribution $e^{-\rho\,\Delta t}$ enters.
The third is trust: if the process breaks down with probability $\beta$ per period---negotiations collapse,
the institutional setting disappears,
the players stop playing this game---then only the surviving branch matters,
and $\delta = (1-\beta)e^{-\rho\,\Delta t}$.
A small $\delta$ therefore means impatient players, or slow negotiations, or fragile ones,
and these are not distinguishable within the model.

The limit $\delta \to 1$ is the case of interest for two reasons.
It is the case in which the model is most demanding, as the last paragraphs of this subsection explain:
a period spent outside the grand coalition costs $(1-\delta)(V^\ast-V(x))$ in aggregate,
where $V^\ast$ and $V(x)$ are the total payoffs in a grand coalition and in some state $x$ outside it,
respectively.
So delay becomes free and holding out for a better division becomes cheap.
And it is the case that the trend in real bargaining points towards,
since $\delta \to 1$ is what one gets by letting $\Delta t \to 0$ at fixed $\rho$ and $\beta$:
the faster the parties can meet, evaluate a proposal and respond to it, the closer $\delta$ is to one.
Where the parties delegate the search for agreements to automated negotiators,
so that the period length is measured in seconds rather than in years of diplomacy,
the relevant regime is precisely $\delta \approx 1$.
That limit turns out to be the demanding one rather than the easy one, for reasons that need the
machinery of Section~\ref{sec:decomp} and are therefore deferred to Section~\ref{sec:discussion}.

\subsection{A standing assumption: transferable utility}\label{sec:tu}

One assumption is so basic to the framework that it is easy to overlook, and it should be stated
before the results rather than after.
Throughout, utility is transferable and transfers are unbounded.
A coalition that signs an agreement divides the surplus it generates in whatever proportions the
sharing rule prescribes, and the model never asks whether the side payments implementing that
division are feasible, or how large they are.
This is what allows the payoff primitive to be a partition function---a number per coalition per
coalition structure---rather than a vector of raw payoffs per player, since any division of a
coalition's total is assumed reachable.

The assumption does real work.
It is what makes the merge-all move Pareto-improving whatever the raw incidence of the merger
(Lemma~\ref{lem:pareto}), and that in turn is feature (F) below, on which most of what follows rests.
Section~\ref{sec:bounded} returns to what happens without it; the short answer is that persistence
survives and much of arrival does not.

\subsection{Three structural features}

Three properties of this particular model, none of which holds in a general coalition-formation model,
carry all of the arguments below.

\begin{description}
\item[(F) Merge-all is available and Pareto-improving.] From every non-grand state $x$ a single move
leads to the grand state $g_x$ obtained by placing one overarching agreement over all existing coalitions.
Because the reference state for that agreement's surplus sharing is $x$ itself,
every player's static payoff strictly increases (Lemma~\ref{lem:pareto}).

\item[(R) Only insiders can dissolve.] Forming a coalition needs the approval of all its members;
dissolving it needs the approval of at least one of them.
Hence no pair of states is traversed in both directions by moves sharing a relevant player (Lemma~\ref{lem:reversal}).

\item[(S) Terminations at a grand state are shadows of terminations at its parent.] Under the model's
chain rule---terminating a coalition also terminates every coalition containing it---every move available at $g_x$ leads to a state that is also reachable in one move from $x$,
by no more players (Lemma~\ref{lem:shadow}).
\end{description}

\subsection{The story in outline}\label{sec:story}

This subsection tells the whole story once, without technical apparatus, and fixes on the way every
term that the summary table of Section~\ref{sec:map} uses.
A reader who stops here should still know what is proved, what is not, and why the distinctions
matter.
We use one example throughout, the one the model was built for.

\begin{example}[Running example: linking carbon markets]\label{ex:running}
The players are countries.
An \emph{agreement} is a linked emissions-trading market whose members coordinate their caps;
the members may be countries, or blocs that had linked earlier, so that agreements nest.
The \emph{grand coalition} is a single global market with a first-best cap.
Total payoff is highest when everybody is inside it, and the question is whether that happens.
We return to this reading after each step below.
It illustrates the concepts, not the pathology: the counterexample of Section~\ref{sec:counterexample}
has a payoff structure quite unlike a climate model, and neither of the two calibrated specifications
in Section~\ref{sec:numerics} produced it.
\end{example}

\emph{States and moves.}
A \emph{state} lists the agreements currently in force, nested inside one another; it is \emph{grand}
if one agreement covers everybody.
Two things can happen.
A \emph{merge} places a new agreement over some of the currently outermost coalitions, and needs the
approval of everyone it covers---agreements are signed unanimously.
A \emph{termination} cancels an agreement, and by the \emph{chain rule} also cancels every agreement
containing it, on the ground that the larger agreement was signed by a party that has ceased to
exist.
The \emph{merge-all} move is the merge that produces the grand coalition in one step, from wherever
the process happens to be.
In the running example, a merge is two markets linking with coordinated caps, and the chain rule says
that if a bloc inside a wider market breaks up, the wider market's cap coordination lapses with it,
because one of its signatories no longer exists.

\emph{Two consent rules.}
Who must approve a termination is the first of the paper's two axes.
Under \emph{unanimous termination} everyone bound by the outermost agreement destroyed must agree;
under \emph{unilateral termination} a single member of the agreement being cancelled suffices.
The second is the realistic one for sovereigns---the Paris Agreement, like most treaties, carries a
withdrawal clause any party may invoke---and Section~\ref{sec:realism} argues the point at length.
Note the asymmetry: forming needs everybody, leaving needs one.
That asymmetry is not decoration; it is what makes the whole analysis work.

\emph{What the players do.}
The model fixes no bargaining procedure.
It constrains the process directly, by three requirements: only rational moves are ever made, the
process does not stall while a rational moving option is available, and a move that some player ranks
first among the rational ones is made with positive probability.
Rational means two things at once.
A move must be worth at least as much, to each player whose consent it needs, as what the negotiation
would otherwise deliver---not merely more than the state they are standing in, which is a weaker
benchmark and the wrong one, since declining a proposal does not commit anybody to spending another
period where they are.
And it must not be undercut by a move that every player \emph{that} one needs prefers, with staying
put counting as such a move, one each player may always call for on their own behalf.
For the running example: no country signs a linkage it expects to do worse under than under continued
negotiation, none is dragged into one that some subset of its own signatories would rather replace,
and a linkage that some country wants most, and that survives both tests, is on the table.

\emph{The requirements are implemented, not stipulated.}
They would be worth little if no negotiation could produce them, and this is where an axiomatic model
usually has to be taken on trust.
Section~\ref{sec:stay} exhibits a bargaining game that produces exactly them.
A proposer is drawn; any player may move a substitute that affects nobody outside the original
motion's consenting set, and the substitute's own consenting set votes it in or out against the
motion; the surviving motion is then put to its consenting set, who must be unanimous.
A failed round costs an arbitrarily small delay, after which the phase begins again with a fresh
proposer.
As the delay vanishes, every equilibrium of that game satisfies the three requirements---and the
delay is never actually paid, because the first proposal is already the best one available, no
substitute defeats it, and it passes at once.
Two features do the work.
Rejection is worth exactly the value of the reopened negotiation, which is what makes the benchmark
above the right one rather than an assumption; and the substitute votes turn ``undercut by a better
move'' from a stipulation into the outcome of a ballot.

\emph{Two questions.}
A state is \emph{absorbing} if the process never leaves it.
\emph{Persistence} asks whether a grand state is absorbing once reached; \emph{arrival} asks whether
one is reached at all.
Both must be answered before one may say that full cooperation occurs: a theory that answers only the
first has not shown cooperation happens, and one that answers only the second has not shown it lasts.

\emph{Persistence: settled, and easily.}
Every grand state that is ever reached is absorbing, and every absorbing state is grand.
This holds for every discount factor, under both consent rules, and with no condition on payoffs.
The reason is structural rather than quantitative.
A grand state can be left only by a termination, and under the chain rule any termination there
either destroys the global agreement---which under unanimity requires everybody's consent, and
everybody was better off inside---or is the shadow of a termination that was already available one
step earlier, in which case the players who wanted it would not have merged in the first place.
No far-sightedness is needed, and no appeal to patience: a global carbon market, once formed, is not
dissolved, whatever the discount factor.

\emph{The dichotomy.}
It follows that an equilibrium can behave in only two ways in the long run: it settles at a grand
state, or it cycles for ever among non-grand states, forming and dissolving agreements without end.
Everything else in the paper is about the second possibility.

\emph{Arrival sometimes holds.}
Three sufficient conditions, of different kinds.
When the discount factor is small, no termination is ever worth making, the process only merges, and
it arrives within $n-1$ steps.
When there are three players, the move structure is too poor to support a cycle: the only way back
from a pair to the non-cooperative state is to dissolve that pair, and the players who dissolve it
are the ones who formed it, which cannot be profitable in both directions.
And for any number of players and any discount factor, arrival holds if every player prefers every
grand state to every non-grand state.
That last condition is about \emph{distribution}, not about efficiency: everyone agrees the grand
coalition is best for the group, and the condition asks in addition that no player prefers some
smaller arrangement for herself.
In the running example it fails when a large permit seller does better in a market that excludes a
competitor than in the global one.

\emph{Arrival can fail---but not on a schedule.}
With four players and a discount factor of one half there is a payoff structure whose only closed
class is a cycle: two pairs alternately form and dissolve, and the global agreement is never signed.
That was found under an earlier and weaker notion of what makes a move worth making, and the reader
should hold it at arm's length until the end of this paragraph.
The mechanism is worth stating, because it is not the one intuition suggests.
Each player would happily be in \emph{some} grand state---each has one they strictly prefer to the
cycle---and yet none forms.
The obstacle is that each grand state has exactly one gateway: the only route into it is the
merge-all move at its own parent state, requiring everyone's signature.
And at that gateway, the players whose consent is needed are precisely the ones that particular
grand state disfavours, because the surplus is divided against the state in which it is signed.
In the running example: whoever currently holds the rent vetoes locking it in at today's baseline,
and who holds it rotates.

That construction, however, predates the requirement that staying be a move like any other, and does
not survive it.
No equilibrium can traverse a cycle deterministically, whatever the payoffs and whichever
termination rule is in force:
go round the cycle and follow any one player's valuation; it must come back to where it started, so
somewhere along the way it falls, and at the top of that fall the player would rather stay than move
on.
Staying is then a rational option for them, and the third axiom puts what they most want on the
table; if some move of their own is better still, that move is on the table instead, and it is not
the one the cycle prescribes.
A cycle would therefore have to be traversed at random, never on a fixed schedule.
Whether such a cycle exists is open, and with it the arrival question.

\emph{Why patience does not help.}
One expects a dynamic model to become well behaved as players grow patient.
Here the opposite is true.
The cost to the group of one more period outside the grand coalition is proportional to $1-\delta$,
so as players become patient, delay becomes free and holding out for a better division becomes cheap.
This is why the persistence results, which rest on the move structure, hold for every discount factor,
while every quantitative condition for arrival degrades as $\delta$ approaches one.
And by the continuous-time reading of Remark~\ref{rem:continuous-time}, faster negotiation is exactly
what pushes $\delta$ towards one---so the hard case is the one that automation moves us towards.

\emph{The sharing rule is the pivot.}
The paper's second axis is how a coalition divides the surplus its agreement creates.
Under \emph{myopic} sharing, which is what the model was originally written with, the surplus divided
is the change in \emph{per-period} payoffs, measured against the state that would obtain if that
agreement alone were missing.
Under \emph{far-sighted} sharing, it is the change in discounted \emph{long-term} payoffs.
The difference is not cosmetic.
Under far-sighted sharing, merging becomes an improvement for everybody in the quantity that actually
governs behaviour, rather than only in this period's payoff; the limit $\delta\to1$ turns from
intractable into tractable; and if every agreement is jointly worth having to its signatories in
long-term terms, then no termination is ever profitable and the grand coalition forms within $n-1$
steps under either consent rule.
For the running example the distinction is whether a country evaluates a linkage by this year's
abatement costs or by the whole path the agreement sets in motion.
That is a design question for a negotiating protocol, and Section~\ref{sec:pivot} shows it is the
question that matters most.

\emph{What is not known.}
Two things.
First, whether arrival can fail at all: the one construction we have is deterministic and is ruled
out under both termination rules, so a failure would have to cycle at random, and we can neither
build one nor exclude one.
Second, arrival under myopic sharing, unilateral termination, and a discount factor close to one---the
combination the motivating application most plausibly satisfies, and the one that fast automated
negotiation moves towards.
There we can neither prove that the grand coalition must eventually form, nor construct an example in
which it does not.

\emph{Two side results.}
The cost of failure can be measured exactly: the shortfall in total payoff equals the payoff gap
multiplied by a discounted measure of how long the grand coalition takes to form, and nothing else.
And that cost cannot be bounded by the efficiency gain at stake---for four players it can be an
arbitrarily large multiple of it---so inefficiency here is paid for in time or in distributional
conflict, never in a quantity one can read off the payoffs.

\subsection{A map of the results}\label{sec:map}

Table~\ref{tab:map} lists the main results with the hypotheses each needs, and
Figure~\ref{fig:spine} shows how they fit together.
The logical spine is short: the axioms are what an explicit bargaining game delivers
(Theorem~\ref{thm:implement}), every equilibrium either settles at a grand state or cycles for ever
among non-grand states (Corollary~\ref{cor:dichotomy}), the first branch is fully understood, and
everything else in the paper is about the second.

\begin{table}[htbp]
\centering
\makebox[\textwidth][c]{%
\begin{tabular}{@{}lllllll@{}}
\toprule
Result & Statement & Term. & Sharing & $\delta$ & Axioms & Conditions \\
\midrule
Thm~\ref{thm:implement} & bargaining implements the axioms & either & any & all & --- & $\varepsilon\to0$ \\
Prop~\ref{prop:proposals} & and does so in the first round & either & any & all & \ref{E1}--\ref{E3} & --- \\
\midrule
Thm~\ref{thm:persist-unan} & all grand states absorbing & unan. & any & all & \ref{E1} & --- \\
Thm~\ref{thm:persist-unil} & reached grand states absorbing & unil. & any & all & \ref{E1} & --- \\
Thm~\ref{thm:absorbing-grand} & absorbing $\Rightarrow$ grand & either & myopic & all & \ref{E1},\ref{E2} & ties$^\ast$ \\
Cor~\ref{cor:dichotomy} & settle at a grand state, or cycle & either & myopic & all & \ref{E1},\ref{E2} & ties$^\ast$ \\
\midrule
Thm~\ref{thm:smalldelta} & arrival & either & myopic & small & \ref{E1},\ref{E2} & gains \\
Thm~\ref{thm:n3} & arrival & either & myopic & all & \ref{E1},\ref{E2} & $n=3$, ties$^\ast$ \\
Thm~\ref{thm:arrival} & arrival & either & myopic & all & \ref{E1}--\ref{E3} & \eqref{eq:prefer} \\
Thm~\ref{thm:onestep} & arrival in one step & either & myopic Nash & large & \ref{E1}--\ref{E3} & no share gain \\
Prop~\ref{prop:flat} & flat grand state absorbing & either & myopic & small & \ref{E1} & --- \\
\midrule
Prop~\ref{prop:rotating} & a $4$-cycle excluded & either & myopic Nash & all & \ref{E1} & sign of $D_i$ \\
Thm~\ref{thm:nodetcycle} & \textbf{no deterministic cycle} & either & any & all & \ref{E1},\ref{E3} & --- \\
Prop~\ref{prop:counterexample}$^{\rm c,\dagger}$ & arrival fails & either & myopic Nash & $\le 3/4$ & \ref{E1}--\ref{E3} & $n=4$ \\
\midrule
Prop~\ref{prop:farsighted} & merge Pareto-improves $\ev$ & either & far-sighted & all & --- & --- \\
Thm~\ref{thm:fs-arrival} & arrival & unan. & far-sighted & large & \ref{E1}--\ref{E3} & spread bound \\
Thm~\ref{thm:combinatorial} & a cycle dies as $\delta\to1$ & unan. & far-sighted & large & \ref{E1}--\ref{E3} & $|H_i|\ge|C|-1$ \\
Thm~\ref{thm:fgains} & \textbf{arrival in $\le n-1$ steps} & \textbf{either} & far-sighted & \textbf{all} & \ref{E1},\ref{E2} & long-run gains \\
Prop~\ref{prop:pdagger} & long-run gains, in primitives & either & far-sighted & all & --- & --- \\
Prop~\ref{prop:price} & price of no-agreement & either & any & all & \ref{E1} & --- \\
Prop~\ref{prop:shapes}$^{\rm c}$ & only one shape of cycle & unan. & myopic Nash & all & \ref{E1}--\ref{E3} & $n=4$ \\
\S\ref{sec:numerics}$^{\rm n}$ & no cycling in two economic models & either & myopic Nash & $\le0.99$ & \ref{E1}--\ref{E3} & $n=4$ \\
\S\ref{sec:counterexample}$^{\rm n}$ & no arriving equilibrium found there & either & myopic Nash & $1/2$ & \ref{E1}--\ref{E3} & $n=4$ \\
Prop~\ref{prop:price-unbounded}$^{\rm c}$ & that price is unbounded in $\gamma$ & unan. & myopic Nash & all & \ref{E1}--\ref{E3} & $n=4$ \\
\bottomrule
\end{tabular}}
\caption{The main results and what each needs.
$^\ast$Criterion~\ref{cond:ties} is needed only in the unilateral case.
Of the last column, `ties' and `long-run gains' are criteria on the process (Criteria~\ref{cond:ties} and~\ref{cond:fgains}) rather than conditions on the primitives.
``Any'' means the proof uses no property of the sharing rule---of the payoffs it uses at most Assumption~\ref{ass:eff}---so those rows cover far-sighted sharing too.
``Myopic'' means only Assumption~\ref{ass:share}, which divides the surplus of
\emph{per-period} payoffs; ``myopic Nash'' adds Condition~\ref{cond:nash}, fixing the shares to be
bargaining weights; ``far-sighted'' divides the surplus of discounted long-term payoffs instead
(Definition~\ref{def:farsighted}).
Assumptions~\ref{ass:eff} and~\ref{ass:share} are in force throughout and are not listed.
A superscript $\rm c$ marks a result established by computation: the computation is finite and
exact, and in the case of Proposition~\ref{prop:counterexample} its outcome is a payoff structure
whose verification is written out in the text and can be checked by hand.
A superscript $\rm n$ marks a numerical finding that is evidence and not proof: a search that found
nothing tells us only that it found nothing.
A superscript $\dagger$ marks a result established under the earlier notion of profitability, before
staying became a move, and not re-established since;
Proposition~\ref{prop:counterexample} is the only one, and Theorem~\ref{thm:nodetcycle} shows its
construction cannot be an equilibrium under the unanimous rule.
Everything unmarked is proved.}\label{tab:map}
\end{table}

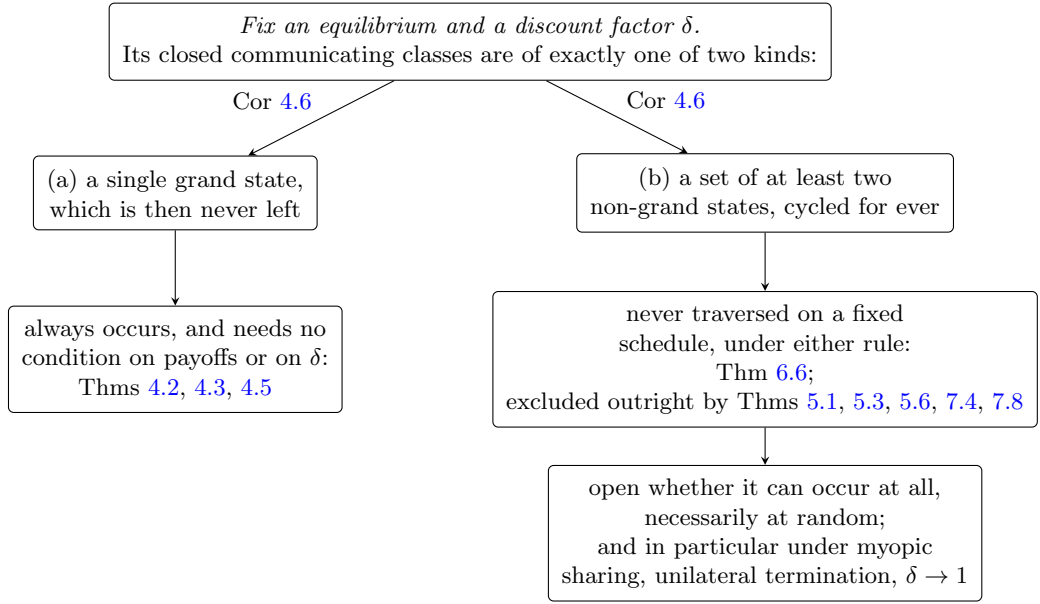
\begin{figure}[htbp]
\centering
\begin{tikzpicture}[>=stealth,font=\small,
  bx/.style={draw,rounded corners=2pt,align=center,inner sep=5pt},
  lb/.style={font=\scriptsize,align=center}]
\node[bx] (top) at (0,3.3) {\emph{Fix an equilibrium and a discount factor $\delta$.}\\
   Its closed communicating classes are of exactly one of two kinds:};
\node[bx] (abs) at (-3.9,1.3) {(a) a single grand state,\\which is then never left};
\node[bx] (cyc) at (3.9,1.3) {(b) a set of at least two\\non-grand states, cycled for ever};
\node[bx] (ok)  at (-3.9,-0.9) {always occurs, and needs no\\condition on payoffs or on $\delta$:\\
   Thms~\ref{thm:persist-unan},~\ref{thm:persist-unil},~\ref{thm:absorbing-grand}};
\node[bx] (no)  at (3.9,-0.9) {never traversed on a fixed\\schedule, under either rule:\\
   Thm~\ref{thm:nodetcycle};\\
   excluded outright by Thms~\ref{thm:smalldelta},~\ref{thm:n3},~\ref{thm:arrival},~\ref{thm:fs-arrival},~\ref{thm:fgains}};
\node[bx] (open) at (3.9,-3.2) {open whether it can occur at all,\\
   necessarily at random;\\
   and in particular under myopic\\sharing, unilateral termination, $\delta\to1$};
\draw[->] (top) -- node[lb,above left]{Cor~\ref{cor:dichotomy}} (abs);
\draw[->] (top) -- node[lb,above right]{Cor~\ref{cor:dichotomy}} (cyc);
\draw[->] (abs) -- (ok);
\draw[->] (cyc) -- (no);
\draw[->] (no) -- (open);
\end{tikzpicture}
\caption{An exhaustive case distinction on the long-run behaviour of a single equilibrium.
The two upper branches are the two kinds of closed communicating class that
Corollary~\ref{cor:dichotomy} permits; the arrows from them lead to what is known about each.
The left case is settled without conditions and for every discount factor; the whole difficulty is
the right one.}\label{fig:spine}
\end{figure}

\section{The model}\label{sec:model}

\subsection{States}

Let $N$ be a finite set of players, $n := |N| \ge 2$.

\begin{definition}[Hierarchy]\label{def:state}
A \emph{state} is a family $x \subseteq 2^{N}$ of subsets of $N$, each of cardinality at least $2$,
that is \emph{laminar}:
for all $K,K' \in x$ either $K \cap K' = \emptyset$ or $K \subseteq K'$ or $K' \subseteq K$.
Elements of $x$ are \emph{agreements} or \emph{nodes}.
Let $\X$ be the (finite) set of states.
\end{definition}

$K \in x$ means the members of $K$ have signed a binding agreement to coordinate their actions,
possibly through agreements signed earlier.
Individual players are not listed; $x = \emptyset$ is the fully non-cooperative state.

\begin{definition}
For $x \in \X$ let $P(x)$ consist of the maximal elements of $x$ together with the singletons $\{i\}$,
$i \notin \bigcup x$; then $P(x)$ is a partition of $N$, the \emph{coalition structure}.
Put $\G := \{x \in \X : N \in x\}$, the set of \emph{grand states}.
As $n \ge 2$, $x \in \G$ iff $P(x) = \{N\}$.
For $x \notin \G$ set $g_x := x \cup \{N\}$, and for $h \in \G$ set $t(h) := h \setminus \{N\}$.
\end{definition}

\begin{lemma}\label{lem:bijection}
$g : \X\setminus\G \to \G$ is a bijection with inverse $t$.
\end{lemma}

\begin{proof}
For $x \notin \G$, $x \cup \{N\}$ is laminar since $K \subseteq N$ for all $K \in x$, and contains $N$,
so $g_x \in \G$ and $t(g_x) = x$.
For $h \in \G$, $h\setminus\{N\}$ does not contain $N$, so $t(h) \notin \G$, and $g_{t(h)} = h$.
\end{proof}

\subsection{Moves}

\begin{definition}[Moves]\label{def:moves}
A \emph{move} at $x$ is a triple $m = (x,\tau(m),R(m))$ with target $\tau(m) \in \X$ and non-empty \emph{relevant set} $R(m) \subseteq N$,
of one of three kinds.
\begin{description}
\item[Merge.] Choose $S \subseteq P(x)$ with $|S|\ge2$ and put $K := \bigcup S$.
Then $\tau(m) := x \cup \{K\}$ and $R(m) := K$.
\item[Termination.] Choose $K \in x$ and put
$\tau(m) := x \setminus \{K' \in x : K \subseteq K'\}$:
the agreement $K$ is deleted together with every agreement containing it (\emph{chain rule}).
Let $S(x,K) \in P(x)$ be the top-level node containing $K$.
The relevant set is
\begin{align*}
R(m) &:= \{i\}\ \text{for some } i \in K &&\text{(\emph{unilateral} termination)},\\
R(m) &:= S(x,K) &&\text{(\emph{unanimous} termination)}.
\end{align*}
\item[Stay.] Choose $i \in N$ and put $\tau(m) := x$ and $R(m) := \{i\}$;
write $s^i_x$ for this move and read it as ``$i$ moves to adjourn''.
\end{description}
$\M(x)$ denotes the set of merges and terminations at $x$, the \emph{moving} moves, and
$\M^{+}(x) := \M(x) \cup \{s^i_x : i \in N\}$ the set of all moves at $x$.
In a termination we call $K$ the \emph{bottom node} and call the players in $R(m)$ the \emph{terminators} of $m$.
Under unilateral termination, terminating a given $K$ yields $|K|$ distinct moves, one per terminator.
\end{definition}

The stays are the only moves with target $x$, since a merge adds an agreement and a termination removes at least one.
That each is relevant to its own mover and to nobody else is the substantive choice, and Definition~\ref{def:rationality} explains what it buys:
with $R=\emptyset$ an empty stay would dominate everything vacuously, while with a collective $R=N$ any profitable move would dominate the adjournment,
so that the status quo would cease to be a player's personal outside option exactly when others want to move.

Merges and terminations map states to states: a union of at least two top-level nodes is disjoint from, or contains,
every element of $x$; and deleting elements preserves laminarity.
A stay maps $x$ to itself.
Under unanimous termination the set whose approval is required is the whole top-level coalition containing $K$,
not merely $K$, because that coalition is destroyed as well.
At a grand state this set is all of $N$; this drives Theorem~\ref{thm:persist-unan}.
Figure~\ref{fig:moves} shows a state and one move of each kind.

\begin{figure}[htbp]
\centering
\begin{tikzpicture}[font=\small,
  pl/.style={inner sep=1.2pt},
  ag/.style={draw,rounded corners=3pt},
  new/.style={draw,rounded corners=3pt,very thick},
  gone/.style={draw=gray,dashed,rounded corners=3pt},
  arr/.style={->,>=stealth,shorten >=4pt,shorten <=4pt}]
\node[pl] (a1) at (0,0.8)    {$A$};
\node[pl] (b1) at (0.6,0.8)  {$B$};
\node[pl] (c1) at (0.3,0.1)  {$C$};
\node[pl] (d1) at (1.5,0.45) {$D$};
\node[pl] (e1) at (2.4,0.8)  {$E$};
\node[pl] (f1) at (2.4,0.1)  {$F$};
\node[ag,inner sep=2pt,fit=(a1)(b1)] (ab1) {};
\node[ag,inner sep=5pt,fit=(ab1)(c1)] (abc1) {};
\node[ag,inner sep=4pt,fit=(e1)(f1)] (ef1) {};
\node[font=\footnotesize] at (1.2,-0.85) {$x=\{AB,ABC,EF\}$};
\begin{scope}[shift={(4.9,0)}]
\node[pl] (a2) at (0,0.8)    {$A$};
\node[pl] (b2) at (0.6,0.8)  {$B$};
\node[pl] (c2) at (0.3,0.1)  {$C$};
\node[pl] (d2) at (1.5,0.45) {$D$};
\node[pl] (e2) at (2.5,0.8)  {$E$};
\node[pl] (f2) at (2.5,0.1)  {$F$};
\node[ag,inner sep=2pt,fit=(a2)(b2)] (ab2) {};
\node[ag,inner sep=5pt,fit=(ab2)(c2)] (abc2) {};
\node[new,inner sep=8pt,fit=(abc2)(d2)] (abcd2) {};
\node[ag,inner sep=4pt,fit=(e2)(f2)] (ef2) {};
\node[font=\footnotesize] at (1.25,-0.85) {$x\cup\{ABCD\}$};
\end{scope}
\begin{scope}[shift={(9.8,0)}]
\node[pl] (a3) at (0,0.8)    {$A$};
\node[pl] (b3) at (0.6,0.8)  {$B$};
\node[pl] (c3) at (0.3,0.1)  {$C$};
\node[pl] (d3) at (1.5,0.45) {$D$};
\node[pl] (e3) at (2.5,0.8)  {$E$};
\node[pl] (f3) at (2.5,0.1)  {$F$};
\node[gone,inner sep=2pt,fit=(a3)(b3)] (ab3) {};
\node[gone,inner sep=5pt,fit=(ab3)(c3)] (abc3) {};
\node[gone,inner sep=8pt,fit=(abc3)(d3)] (abcd3) {};
\node[ag,inner sep=4pt,fit=(e3)(f3)] (ef3) {};
\node[font=\footnotesize] at (1.25,-0.85) {$\{EF\}$};
\end{scope}
\draw[arr] (3.05,0.45) -- (4.65,0.45);
\node[font=\scriptsize,align=center] at (3.85,2.0) {merge\\$S=\{ABC,\{D\}\}$};
\draw[arr] (8.0,0.45) -- (9.55,0.45);
\node[font=\scriptsize,align=center] at (8.78,2.0) {terminate\\bottom node $AB$};
\end{tikzpicture}
\caption{States and moves, with six players and writing $AB$ for $\{A,B\}$ etc.
Left: the state $x=\{AB,ABC,EF\}$; boxes are agreements, nesting is containment, and the coalition structure is $P(x)=\{ABC,\,EF,\,\{D\}\}$.
Middle: the merge with $S=\{ABC,\{D\}\}$ adds the agreement $ABCD$ (thick); its relevant set is $ABCD$.
Right: terminating the bottom node $AB$ deletes it together with every agreement containing it (dashed), by the chain rule, and leaves $\{EF\}$.
The relevant set of that termination is the top-level node $ABCD$ under unanimous termination, and a single member of $AB$ under unilateral termination.}\label{fig:moves}
\end{figure}
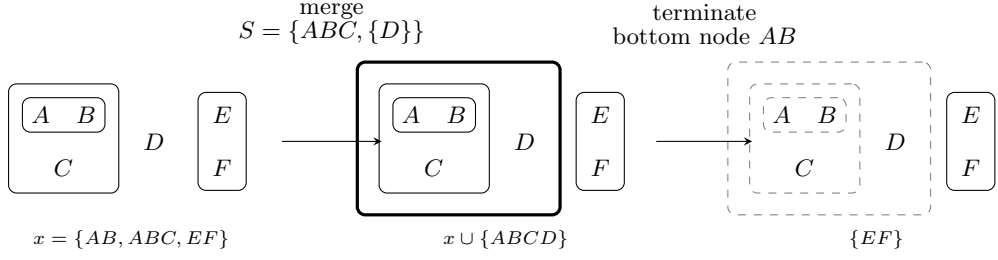

\begin{remark}[Excluded variants]\label{rem:variants}
(i) \emph{Internal termination},
i.e.\ deleting $K$ while keeping the agreements containing it (say $\{\{A,B\},\{A,B,C\}\} \to \{\{A,B,C\}\}$),
is not a move under the chain rule.
Nor should it be:
once $\{A,B\}$ is subsumed in $\{A,B,C\}$ its only remaining effect is on the reference state used to split the larger agreement's surplus,
so allowing $A$ and $B$ to delete it unilaterally would let two members rewrite the burden sharing of an agreement $C$ is still bound by.
That is renegotiation, not termination.
(ii) \emph{Individual withdrawal}, where $i$ leaves and the remainder stays coordinated,
changes the target of a deviation and invalidates Lemma~\ref{lem:pareto} for it; it is not analysed here.
\end{remark}

\begin{lemma}[Unique entry into $\G$]\label{lem:entry}
Let $h \in \G$.
The only element of $\M(y)$ with target $h$, for any $y\ne h$, is the merge at $y=t(h)$ with $S = P(t(h))$, whose relevant set is $N$.
Every element of $\M(h)$ is a termination and its target lies outside $\G$.
In particular no moving move leads from one grand state to another, and the only moves with target $h$ available at $h$ itself are the stays.
\end{lemma}

\begin{proof}
A termination deletes at least one node; if its source is $h$ then $N$ is deleted,
either because $K=N$ or because $K \subsetneq N$ and $N$ is an ancestor of $K$; so its target is not in $\G$.
Hence $h$ can only be the target of a merge at some $y$ with $S\subseteq P(y)$,
$|S|\ge2$ and $y \cup \{\bigcup S\} = h$.
Since $N \in h$ while $N \notin y$ (otherwise $|P(y)|=1<2$),
we need $\bigcup S = N$ and $y = h \setminus \{N\} = t(h)$; then $S = P(t(h))$ and $R = N$.
Finally $|P(h)|=1$, so no merge is available at $h$.
\end{proof}

\begin{lemma}[Shadow lemma]\label{lem:shadow}
Let $h \in \G$ and $x := t(h)$, and let $m$ be a termination at $h$ with bottom node $K$.
If $K=N$ then $\tau(m)=x$.
If $K \subsetneq N$ then $K \in x$,
the termination $\tilde m$ at $x$ with bottom node $K$ and the same terminator (under unilateral termination) exists,
and
\[
\tau(\tilde m) = \tau(m), \qquad R(\tilde m) \subseteq R(m).
\]
\end{lemma}

\begin{proof}
If $K = N$ then $\tau(m) = h\setminus\{N\} = x$.
Let $K \subsetneq N$.
Then $\{K' \in h : K \subseteq K'\} = \{K' \in x : K \subseteq K'\} \cup \{N\}$, so
\begin{align*}
\tau(m) &= h \setminus \{K'\in h: K \subseteq K'\} = (x \cup \{N\}) \setminus \big(\{K'\in x :
K\subseteq K'\}\cup\{N\}\big)\\ &= x \setminus \{K' \in x: K \subseteq K'\} = \tau(\tilde m).
\end{align*}
Under unilateral termination $R(\tilde m) = R(m) = \{i\}$ with $i \in K$;
under unanimous termination $R(m) = S(h,K) = N \supseteq S(x,K) = R(\tilde m)$.
\end{proof}

\subsection{Payoffs}\label{sec:payoffs}

Static payoffs are given by $\pi : \X \to \R^{N}$,
where $\pi_i(x)$ is $i$'s payoff \emph{per period while the state is $x$}.
Write $V(x) := \sum_{i} \pi_i(x)$.

\emph{Convention.} We call \emph{assumption} only what is maintained everywhere in this note,
and \emph{condition} what is invoked in the statement of individual results.
There are exactly two assumptions, both stated in this subsection.
Every other hypothesis is a condition, and every result names the conditions it uses.

\begin{assumption}[Efficiency of the grand partition]\label{ass:eff}
$V$ is constant on $\G$, with value $V^\ast$, and $\gamma := \min_{x \in \X\setminus\G}(V^\ast - V(x)) > 0$.
\end{assumption}

Assumption~\ref{ass:eff} is all we ever need about $V$;
we never assume that $V(x)$ depends only on the coalition structure $P(x)$.
That stronger property is, however, the natural reason to expect Assumption~\ref{ass:eff}:
if the physical outcome is determined by a game between the top-level coalitions,
then $V$ is a function of $P(x)$ alone, it is constant on $\G$ because $P(x)=\{N\}$ there,
and it is strictly maximal there whenever full coordination is strictly efficient.

For $x \in \X$ and $K \in x$ write $x^{-K} := x \setminus \{K\}$ for the state in which the agreement $K$,
and only that agreement, is missing, and
\[
\Delta(x,K) := \sum_{j\in K}\big(\pi_j(x)-\pi_j(x^{-K})\big)
\]
for the \emph{surplus} of $K$ at $x$: what its signatories jointly gain by having signed it.

\begin{assumption}[Surplus sharing at the top level]\label{ass:share}
For every state $x$ and every \emph{top-level} node $K \in x \cap P(x)$ there are \emph{shares} $s_i(x,K)>0$,
$i \in K$, with $\sum_{i\in K}s_i(x,K)=1$, such that
\[
\pi_i(x) = \pi_i(x^{-K}) + s_i(x,K)\,\Delta(x,K) \qquad\text{for } i \in K .
\]
\end{assumption}

In words: the signatories of an agreement divide its surplus among themselves,
each receiving a strictly positive fraction,
so that all of them gain if the surplus is positive and all of them lose if it is negative.
Since there are finitely many states, nodes and players, the smallest share
\[
s_{\min} := \min\{s_i(x,K) : x \in \X,\ K \in x \cap P(x),\ i \in K\}
\]
is strictly positive; under Condition~\ref{cond:nash} it equals $\min_{K,i\in K}w_i/w(K)\ge w_{\min}$.
Nothing is assumed about how the shares are determined, or about whether they vary with the state.
The rule is imposed for top-level $K$ only, because that is the situation in which the agreement is negotiated:
there $x^{-K}$ is the status quo that would prevail if the signatories walked away,
and hence the natural disagreement point.
For a non-top-level $K$ the state $x^{-K}$ is not even reachable by a move,
and no sharing rule is part of the specification there.

\begin{condition}[Weighted Nash bargaining shares]\label{cond:nash}
There are \emph{bargaining weights} $w_i>0$ with $\sum_{i\in N}w_i=1$ such that,
writing $w(K) := \sum_{i\in K}w_i$, the shares of Assumption~\ref{ass:share} are $s_i(x,K)=w_i/w(K)$,
independently of $x$.
\end{condition}

Condition~\ref{cond:nash} is the rule used in \citet{HeitzigKornek2018}:
the surplus is divided in proportion to fixed weights,
which is the asymmetric Nash bargaining solution of the induced bargaining problem,
equivalently a weighted Shapley value of the induced unanimity game \citep{Kalai1977,KalaiSamet1987}.
We write $w_{\min} := \min_i w_i$ whenever the condition is in force.

Under Condition~\ref{cond:nash},
the sharing formula that is imposed only at the top level in fact propagates to every node of the hierarchy.

Nothing so far forces surpluses to be positive; Example~\ref{ex:negative} shows that they need not be.

\begin{condition}[Coalitional gains]\label{cond:gains}
$\Delta(x,K)>0$ for every state $x$ and every top-level node $K \in x \cap P(x)$;
put $\sigma := \min\{\pi_i(x)-\pi_i(x^{-K}) : x \in \X,\ K \in x\cap P(x),\ i \in K\} > 0$.
\end{condition}

Condition~\ref{cond:gains} is used only in Theorem~\ref{thm:smalldelta};
positivity of $\sigma$ follows from $\Delta>0$ and $s_i>0$ via Assumption~\ref{ass:share}.

\begin{example}[A coalition with negative surplus]\label{ex:negative}
Condition~\ref{cond:gains} genuinely restricts.
Let $n=3$ and
\[
\pi_i \;=\; Q - \tfrac12 Q^2 - \tfrac12 q_i^2, \qquad Q := \textstyle\sum_{j} q_j ,
\]
a public-good game with linear-quadratic benefits, so that contributions are strategic substitutes:
the more the others provide, the less each player wants to provide.
A top-level coalition $K$ chooses $(q_i)_{i\in K}$ to maximise $\sum_{i \in K}\pi_i$,
taking outsiders' contributions as given,
and the resulting simultaneous optimisation has the unique solution $q_i = |K_i|\,m$,
where $K_i$ is the top-level coalition of $i$ and $m := 1/(1+\sum_{K \in P}|K|^2)$.
Hence
\[
\begin{array}{lcccc}
\text{state} & q \text{ (member, outsider)} & \text{member's } \pi & \text{joint } \pi \text{ of } \{A,B\} & V\\[2pt]
\emptyset & (\tfrac14,\tfrac14) & \tfrac7{16} & \tfrac78 & \tfrac{21}{16}\\[2pt]
\{\{A,B\}\} & (\tfrac13,\tfrac16) & \tfrac{31}{72} & \tfrac{31}{36} & \tfrac43\\[2pt]
\{N\} & (\tfrac3{10},-) & \tfrac9{20} & - & \tfrac{27}{20}
\end{array}
\]
So $\Delta(\{\{A,B\}\},\{A,B\}) = \tfrac{31}{36}-\tfrac78 = -\tfrac1{72} < 0$:
forming the pair raises total payoff ($\tfrac43 > \tfrac{21}{16}$) but \emph{lowers} the pair's own joint payoff,
because the outsider free-rides by cutting his contribution from $\tfrac14$ to $\tfrac16$.
Under Assumption~\ref{ass:share} both members then share the loss.
Assumptions~\ref{ass:eff} ($V^\ast=\tfrac{27}{20}$, $\gamma=\tfrac1{60}$) and~\ref{ass:share} hold,
so every result of this note except Theorem~\ref{thm:smalldelta} applies to this example.
\end{example}

\subsection{Evaluations and equilibrium}

Fix $\delta \in (0,1)$.
A \emph{process} is a stochastic matrix $p=(p_{x\to y})_{x,y\in\X}$ on $\X$.
We write $(X_t)_{t \ge 0}$ for the canonical Markov chain with transition kernel $p$,
and $\Prob_x$ and $\E_x$ for probability and expectation under that chain started at $X_0=x$.
Given $p$, \emph{evaluations} solve
\begin{equation}\label{eq:bellman}
\ev_i(x) = (1-\delta)\pi_i(x) + \delta \sum_{y}p_{x\to y}\ev_i(y)
 \qquad\text{for all } i \in N \text{ and } x \in \X .
\end{equation}

\begin{lemma}\label{lem:value}
For each $i$, \eqref{eq:bellman} has the unique solution $\ev_i=(1-\delta)(I-\delta p)^{-1}\pi_i$, and
\begin{equation}\label{eq:series}
\ev_i(x) = (1-\delta)\,\E_x\Big[\textstyle\sum_{t\ge0}\delta^t \pi_i(X_t)\Big].
\end{equation}
In particular $\ev_i(x)$ is a convex combination of the numbers $\pi_i(y)$ over states $y$ reachable from $x$;
hence $\ev_i(x) \le \max\{\pi_i(y): y \text{ reachable from } x\}$.
\end{lemma}

\begin{proof}
$\|\delta p\|_\infty = \delta<1$, so $I-\delta p$ is invertible with inverse $\sum_t \delta^t p^t$;
\eqref{eq:series} is that series applied to $\pi_i$,
and $(1-\delta)\sum_t\delta^t=1$ makes the coefficients a probability distribution supported on reachable states.
\end{proof}

\begin{definition}[Moves including staying; rationality]\label{def:rationality}
Fix $p$ and the induced $\ev$, and let $x \in \X$, and recall $\M^{+}(x)$ and the stays $s^i_x$ from Definition~\ref{def:moves}.
Write
\begin{equation}\label{eq:table}
\evp_i(x) := \sum_y p_{x\to y}\,\ev_i(y),
\end{equation}
the value to $i$ of the state the process moves to.
It is $\ev_i$ with the current period removed: \eqref{eq:bellman} reads
$\ev_i(x)=(1-\delta)\pi_i(x)+\delta\,\evp_i(x)$, and the superscript is meant to record that
$\evp_i(x)$ looks one period further ahead than $\ev_i(x)$ does.
It is not to be confused with the total flow $V$ of Subsection~\ref{sec:payoffs}, which carries no
player index.
Define a total preorder $\succeq^{x}_i$ on $\X$ by
\begin{equation*}
y \succeq^{x}_i z \quad:\Longleftrightarrow\quad \ev_i(y) > \ev_i(z), \ \text{ or }\ \ev_i(y)=\ev_i(z) \text{ and } (y = x \text{ or } z \ne x),
\end{equation*}
ranking states by $\ev_i$ with ties broken in favour of the status quo, and write $\succ^{x}_i$ for its strict part.
Then, for $m \in \M^{+}(x)$:
\begin{enumerate}[label=(\alph*)]
\item $m$ is \emph{profitable} if $\ev_i(\tau(m)) \ge \evp_i(x)$ for all $i \in R(m)$;
\item $m$ is \emph{dominated} if there is $m'\in\M^{+}(x)$ with $R(m')\subseteq R(m)$ and $\tau(m') \succ^{x}_i \tau(m)$ for all $i\in R(m')$;
\item $\Cset(x)$ is the set of profitable undominated elements of $\M^{+}(x)$, and
\[
\Fset(x) := \big\{m \in \Cset(x) : \exists\, i \in N\ \forall\,
m'\in\Cset(x):\ \tau(m) \succeq^{x}_i \tau(m')\big\}
\]
is the set of moves that are some player's favourite in $\Cset(x)$.
The player $i$ witnessing membership in $\Fset(x)$ need not be relevant for $m$; they may be any player.
\end{enumerate}
\end{definition}

Two features deserve comment, and Subsection~\ref{sec:stay} obtains both from explicit bargaining.

Staying competes with the other moves.
Since ties are broken towards $x$, a stay $s^i_x$ dominates a move $m$ with $i \in R(m)$ precisely when $\ev_i(x) \ge \ev_i(\tau(m))$;
so an undominated move is one that no relevant player would rather see adjourned, and the clause that used to require profitability against $\ev_i(x)$ is now carried by domination.

Profitability compares the target not with $x$ but with $\evp_i(x)$.
The reason is that a player who declines a move does not thereby spend a period at $x$:
the state for the current period is fixed at the start of it, and declining leaves that decision to be taken again, so what a player forgoes by agreeing to $m$ is the value $\evp_i(x)$ of whatever the process settles on.
The difference between the two benchmarks is one period of flow,
\begin{equation}\label{eq:anchor}
\evp_i(x)-\ev_i(x) \;=\; \tfrac{1-\delta}{\delta}\,\bigl(\ev_i(x)-\pi_i(x)\bigr),
\end{equation}
so the two agree exactly where $\ev_i(x)=\pi_i(x)$, in particular at absorbing states, and $\evp_i(x)$ is the more demanding benchmark exactly where $i$ expects gains.
Where the process moves deterministically, $\evp_i(x)=\ev_i(\tau(m))$ for the realised $m$ and profitability holds with equality;
the discipline on realised moves then rests on domination, which is why staying has to be available as a dominator.

\begin{definition}[Equilibrium]\label{def:eq}
A process $p$ is an \emph{equilibrium} if, with $\ev$ from \eqref{eq:bellman}, for every $x \in \X$:
\begin{enumerate}[label=(E\arabic*)]
\item\label{E1} for all $y \ne x$: $p_{x\to y}>0$ implies $y = \tau(m)$ for some $m \in \Cset(x)$;
\item\label{E2} if $\Cset(x)\cap\M(x)\ne\emptyset$ then $p_{x\to x}<1$;
\item\label{E3} for all $m \in \Fset(x)$: $p_{x \to \tau(m)}>0$.
\end{enumerate}
A state $x$ is \emph{absorbing} if $p_{x\to x}=1$.
\end{definition}

\ref{E1} says that only profitable and undominated moves are realised;
\ref{E2} that the process does not stall when some \emph{moving} profitable undominated move is available;
\ref{E3} that every player's favourite among those moves is realised with positive probability,
as in the model of \citet{HeitzigKornek2018},
where a move's probability is proportional to the total bargaining weight of the players favouring it and all weights are positive.
Note that \ref{E1} does \emph{not} say that only {\em favourite} moves can be realised.
It permits any profitable undominated move, whether or not it is anybody's favourite;
and since every stay has target $x$, \ref{E3} forces $p_{x\to x}>0$ whenever some stay is a favourite;
\ref{E3} only requires that the favourites be \emph{among} the moves realised.
The rule of \citet{HeitzigKornek2018}, under which exactly the favourite moves are realised,
is one process satisfying \ref{E1}--\ref{E3}; a rule realising every profitable undominated move is another;
all results below cover both.

Existence is not needed for anything below---all statements have the form ``every equilibrium has property $P$''---but it is reassuring;
see Appendix~\ref{app:existence}.

\begin{remark}\label{rem:favourite}
Making the status quo an element of the comparison removes a lacuna that arises when it is omitted.
If favourites were taken over moving moves only, a move could be somebody's favourite merely as their least bad option:
its witness would weakly prefer the status quo to everything on offer, and every member of its relevant set, gaining by it, would rank something else higher, so nobody would have an interest in it.
With stays included and ties broken towards $x$, a moving favourite is strictly better than staying for its witness whenever the witness's own stay is itself rational:
if $s^i_x\in\Cset(x)$ and $m\in\Fset(x)$ is moving with witness $i$, then $\tau(m)\succeq^{x}_i x$ and $\tau(m)\ne x$ force $\ev_i(\tau(m))>\ev_i(x)$.
\end{remark}

\section{Preliminaries}\label{sec:prelim}

Three groups of facts, in increasing order of importance.
Section~\ref{sec:averages} separates three payoff quantities that are easy to confuse and whose
confusion produces false theorems; it is the one piece of bookkeeping we ask the reader to take
seriously.
Section~\ref{sec:decomp} draws out the two consequences of the sharing rule that drive everything
that follows.
Section~\ref{sec:routine} collects three lemmas that are used repeatedly and proved once.

\subsection{Three payoff quantities that must not be confused}\label{sec:averages}

Three quantities occur below and the analysis depends on separating them.

\begin{enumerate}[label=(\alph*)]
\item The \emph{static} (or stationary) payoff $\pi_i(x)$: what $i$ receives per period \emph{if $x$
were to prevail}.
A property of $x$ alone.
\item The \emph{discounted long-term expected payoff} $\ev_i(x)$ of \eqref{eq:series}: a property of
$x$ \emph{and} of the entire process $p$.
Profitability, domination and favourites are defined by $\ev$, never by $\pi$.
\item The \emph{long-run average payoff} $\bar a_i$ of a recurrent class: the average of $\pi_i$
against the stationary distribution of that class.
\end{enumerate}

The correct link between (b) and (c) is the following identity.
It is exact at every $\delta$.
It is \emph{not} the statement that $\ev_i$ is constant on a recurrent class, which is false for $\delta<1$.

\begin{lemma}[Occupation identity]\label{lem:occupation}
Let $p$ be a process, $C$ a closed communicating class with stationary distribution $\mu$ ($\mu(C)=1$,
$\mu p = \mu$).
Put $\bar a_i := \sum_{x\in C}\mu(x)\pi_i(x)$ and $\bar V := \sum_{x\in C}\mu(x)V(x) = \sum_i \bar a_i$.
Then for all $i$ and all $\delta \in (0,1)$,
\[
\sum_{x\in C}\mu(x)\,\ev_i(x) \;=\; \bar a_i .
\]
\end{lemma}

\begin{proof}
$C$ is closed, so $\ev_i|_C$ solves \eqref{eq:bellman} for $p|_C$.
Multiplying by $\mu(x)$, summing over $x \in C$ and using $\mu p = \mu$,
\[
\sum_x \mu(x)\ev_i(x) = (1-\delta)\bar a_i + \delta\sum_y\Big(\sum_x\mu(x)p_{x\to y}\Big)\ev_i(y) = (1-\delta)\bar a_i + \delta \sum_y \mu(y)\ev_i(y),
\]
whence $(1-\delta)\sum_x\mu(x)\ev_i(x) = (1-\delta)\bar a_i$; divide by $1-\delta>0$.
\end{proof}

\begin{remark}\label{rem:limit}
With $p$ fixed, $\ev_i \to p^\ast\pi_i$ as $\delta\to1$,
where $p^\ast$ is the Ces\`aro limit of $p^t$ (an Abelian limit theorem; no aperiodicity is needed),
so in the limit $\ev_i$ is constant on $C$ with value $\bar a_i$.
We never use this: the convergence is not uniform in $p$,
and in an equilibrium $p$ itself depends on $\delta$.
Lemma~\ref{lem:occupation} and Lemma~\ref{lem:cycle} below give what is needed, exactly.
\end{remark}

The following is the exact form, at finite $\delta$,
of the statement ``a player moves only where she is worse off than in what follows''.
Note that the comparison is with a \emph{discounted} average of the remaining states of the cycle,
\emph{not} with the stationary average $\bar a_i$.

\begin{lemma}[One-step identity]\label{lem:onestep}
Let $p$ be a process, $x \in \X$,
and suppose the only state other than $x$ that receives positive probability from $x$ is $y \ne x$;
write $q := p_{x\to y} \in [0,1]$.
Then for every $i \in N$ and every $c \in \R$,
\[
\big(1-\delta(1-q)\big)\big(\ev_i(x)-c\big) \;=\; (1-\delta)\big(\pi_i(x)-c\big) \;+\;
\delta q\,\big(\ev_i(y)-c\big).
\]
In particular, since $1-\delta(1-q)>0$,
\begin{enumerate}[label=(\alph*)]
\item taking $c := \ev_i(y)$: \ $\ev_i(x)>\ev_i(y) \iff \pi_i(x)>\ev_i(y)$;
\item taking $c := \pi_i(x)$ and assuming $q>0$: \ $\ev_i(x)>\pi_i(x) \iff \ev_i(y)>\pi_i(x)$.
\end{enumerate}
\end{lemma}

\begin{proof}
By \eqref{eq:bellman}, $\ev_i(x) = (1-\delta)\pi_i(x)+\delta\big((1-q)\ev_i(x)+q\,\ev_i(y)\big)$;
collecting the $\ev_i(x)$ terms gives $\big(1-\delta(1-q)\big)\ev_i(x) = (1-\delta)\pi_i(x)+\delta q\,\ev_i(y)$.
Subtract $\big(1-\delta(1-q)\big)c = (1-\delta)c+\delta q\,c$ from both sides.
For (a) the last term vanishes and $1-\delta>0$; for (b) the first term vanishes and $\delta q>0$.
\end{proof}

\begin{lemma}[Cycle inequality]\label{lem:cycle}
Let $p$ be a process whose closed communicating class $C=\{x_1,\dots,x_k\}$ is a deterministic cycle,
i.e.\ $p_{x_j \to x_{j+1}}=1$ for all $j$ (indices mod $k$), $k \ge 2$.
Then for all $i$ and $j$,
\[
\ev_i(x_{j+1}) > \ev_i(x_j) \quad\Longleftrightarrow\quad \ev_i(x_{j+1}) > \pi_i(x_j) \quad\Longleftrightarrow\quad \sum_{s=1}^{k-1}\delta^{\,s-1}\pi_i(x_{j+s}) \;>\;
\Big(\sum_{s=1}^{k-1}\delta^{\,s-1}\Big)\pi_i(x_j).
\]
\end{lemma}

\begin{proof}
The first equivalence is Lemma~\ref{lem:onestep}(a) applied at $x_j$ with $y=x_{j+1}$ and $q=1$,
read in the form $\ev_i(x_j)>\ev_i(x_{j+1}) \iff \pi_i(x_j)>\ev_i(x_{j+1})$ and negated (a tie on one side is a tie on the other,
since $\ev_i(x_{j+1})-\ev_i(x_j) = (1-\delta)(\ev_i(x_{j+1})-\pi_i(x_j))$ by the same identity).
By \eqref{eq:series} and periodicity,
$\ev_i(x_{j+1}) = \rho\sum_{s=0}^{k-1}\delta^{s}\pi_i(x_{j+1+s})$ with $\rho := (1-\delta)/(1-\delta^{k}) = \big(\sum_{s=0}^{k-1}\delta^s\big)^{-1}$.
Splitting off the term $s=k-1$, which equals $\delta^{k-1}\pi_i(x_j)$,
the inequality $\ev_i(x_{j+1})>\pi_i(x_j)$ becomes $\sum_{s=0}^{k-2}\delta^s\pi_i(x_{j+1+s}) > (\rho^{-1}-\delta^{k-1})\pi_i(x_j) = \big(\sum_{s=0}^{k-2}\delta^s\big)\pi_i(x_j)$,
which is the third form after reindexing.
\end{proof}

\subsection{Pareto improvement and the pie/share decomposition}\label{sec:decomp}

\begin{lemma}[Merge-all strictly Pareto-improves]\label{lem:pareto}
For every $x \in \X\setminus\G$ and every $i \in N$,
\[
\pi_i(g_x) \;=\; \pi_i(x) + s_i(g_x,N)\big(V^\ast - V(x)\big) \;\ge\; \pi_i(x) + s_i(g_x,N)\,\gamma \;>\;
\pi_i(x),
\]
and under Condition~\ref{cond:nash} the shares are $s_i(g_x,N)=w_i$, so that
\[
\pi_i(g_x) \;=\; \pi_i(x) + w_i\big(V^\ast - V(x)\big) \;\ge\; \pi_i(x)+w_i\gamma .
\]
\end{lemma}

\begin{proof}
Apply Assumption~\ref{ass:share} at the state $g_x$ with $K=N$, which is top-level there:
the reference state is $(g_x)^{-N}=x$,
and $\Delta(g_x,N)=\sum_j(\pi_j(g_x)-\pi_j(x)) = V^\ast - V(x) \ge \gamma>0$ by Assumption~\ref{ass:eff}.
The shares are strictly positive and $w(N)=1$.
\end{proof}

This is the single most important consequence of the structure,
and it is worth saying what it does and does not assert.
It says that the players are never blocked, in \emph{static} payoffs, from reaching a grand state in one step:
whatever the current state,
the move that puts one agreement over everything makes everybody strictly better off per period.
It does not say that the move is profitable, because profitability is about evaluations, not static payoffs,
and a player may prefer to wait in the hope of a better division later.
The whole difficulty of the arrival half of the problem lies in that gap.

Throughout the rest of this subsection we impose Condition~\ref{cond:nash};
it is the only place in the note where the exact form of the shares matters,
and Table~\ref{tab:map} records which later results inherit that dependence.

\begin{definition}\label{def:psi}
Assume Condition~\ref{cond:nash}.
For $x \in \X$ and $i\in N$ put
\[
\psi_i(x) := \pi_i(x) - w_i\,V(x), \qquad\text{so}\qquad \sum_{i\in N}\psi_i(x)=0;
\]
and for $x \notin \G$ put $\varphi_i(x) := \pi_i(g_x)$.
\end{definition}

\emph{Interpretation of $\psi$.} Think of $w_i$ as the share of the cake that player $i$ would receive if the cake were divided by bargaining power alone.
Then $\psi_i(x)$ is the amount by which $i$'s actual payoff at $x$ exceeds that benchmark:
$\psi_i(x)>0$ means that the agreements in force at $x$, and the order in which they were signed,
have left $i$ better placed than her raw bargaining weight,
and $\psi_i(x)<0$ that they have left her worse placed.
Since the benchmark shares sum to the whole cake, the $\psi_i$ sum to zero: $\psi$ is a pure redistribution,
measured in payoff units.

Three warnings.
First, $\psi_i(x)$ is not itself a payoff and cannot be compared across players;
only differences of $\psi_i$ across states, for a \emph{fixed} $i$, are used below.
Second, $w_i$ is not a welfare weight chosen by an analyst; it is a primitive of the bargaining rule.
Third---and this is the point of the whole construction---$\psi$ is not a function of the coalition structure $P(x)$:
two states with the same top-level partition but different internal nesting have the same total payoff and generally different $\psi$.
That is because the sharing rule measures each agreement's surplus against the state in which that agreement is absent,
so what a player gets depends on which sub-agreements had already been signed when she signed.
In this sense $\psi(x)$ is what the nesting of the agreements currently in force is worth to each player,
and it is exactly the partial record described in Section~\ref{sec:question}(i):
it says nothing about agreements that were signed and later terminated.

\begin{lemma}\label{lem:psi}
Assume Condition~\ref{cond:nash} and let $x\in\X\setminus\G$ and $i \in N$.
Then $\varphi_i(x) = \psi_i(x)+w_i V^\ast$, and $\psi_i(g_x)=\psi_i(x)$, and $\sum_i \varphi_i(x)=V^\ast$.
Consequently $i$ ranks grand states exactly by $\psi_i$ of their parents,
and no grand state Pareto-dominates another.
\end{lemma}

\begin{proof}
$\varphi_i(x)=\pi_i(x)+w_i(V^\ast-V(x)) = \psi_i(x)+w_i V^\ast$ by Lemma~\ref{lem:pareto},
and $\psi_i(g_x) = \pi_i(g_x)-w_i V^\ast = \varphi_i(x)-w_i V^\ast = \psi_i(x)$.
Summing over $i$ uses $\sum_i\psi_i\equiv0$ and $\sum_iw_i=1$.
Since $\sum_i\varphi_i(x)$ is independent of $x$,
no grand state can be weakly better for all and strictly better for one.
\end{proof}

So payoffs decompose into a \emph{pie} $V(x)$, consumed by each $i$ in the fixed proportion $w_i$,
and a zero-sum \emph{relative share} $\psi(x)$ that reflects how the bargaining positions of the players upon signing the agreements in $x$ depended on the earlier signed sub-agreements in $x$.

The key insight is that merging all coalitions is a pure pie move:
it leaves the relative distribution $\psi$ unchanged and only raises the pie to its maximum.
All conflict between grand states is conflict about $\psi$, and it is exactly constant-sum.

\begin{example}[Three symmetric players]\label{ex:sym}
With benefits $\sum_j q_j$, costs $q_i^2/2$ and equal weights one computes $V(\emptyset)=7.5$,
$V(\{\{A,B\}\})=10.5$, $V^\ast=13.5$, and $\pi_C(\emptyset)=2.5$, $\pi_C(\{\{A,B\}\})=4.5$,
$\pi_A(\{\{A,B\}\})=3$.
Lemma~\ref{lem:pareto} then gives $\pi_C$ at the grand states with parents $\emptyset$, $\{\{A,B\}\}$,
$\{\{B,C\}\}$ equal to $4.5$, $5.5$, $4$, matching the table in \citet{HeitzigKornek2018}.
The relative shares are $\psi_C = 0, +1, -0.5$ at those parents:
the share premium accrues to whoever stayed outside.
\end{example}

\subsection{Welfare bound, domination, non-stalling}\label{sec:routine}

\begin{lemma}[Welfare bound]\label{lem:welfare}
Let $W(x) := \sum_i \ev_i(x)$.
Then $W(x)\le V^\ast$ for all $x \in \X$, with equality if and only if $\Prob_x(\forall t: X_t \in \G)=1$.
\end{lemma}
\begin{proof}
By \eqref{eq:series}, $W(x)=(1-\delta)\E_x[\sum_t\delta^t V(X_t)]$, a convex combination of values of $V$,
all $\le V^\ast$ with equality only on $\G$ (Assumption~\ref{ass:eff}).
\end{proof}

\begin{lemma}[Domination is acyclic; inheritance]\label{lem:domination}
Let $m\in\M^{+}(x)$ be profitable.
There is $m''\in\Cset(x)$ with $R(m'')\subseteq R(m)$ and $\tau(m'')\succeq^{x}_i\tau(m)$ for all $i\in R(m'')$,
the relation being strict if $m$ is dominated.
In particular $\Cset(x)\ne\emptyset$ whenever some element of $\M^{+}(x)$ is profitable.
\end{lemma}

\begin{proof}
Write $m \prec m'$ if $m'$ dominates $m$.
If $m_1\prec m_2\prec\cdots\prec m_r\prec m_1$ then $R(m_2)\subseteq R(m_1),\dots,R(m_1)\subseteq R(m_r)$,
so all these sets equal a common non-empty $R$;
picking $i\in R$ and chaining Definition~\ref{def:rationality}(b) gives $\tau(m_1)\succ^{x}_i\tau(m_1)$,
absurd, since $\succ^{x}_i$, the strict part of a total preorder, is transitive and irreflexive.
Hence $\prec$ is acyclic on the finite set $\M^{+}(x)$,
so a $\prec$-increasing chain from $m$ ends at an undominated $m''$.
Along one step $m\prec m'$ we have $R(m')\subseteq R(m)$ and $\tau(m')\succ^{x}_i\tau(m)$ for $i \in R(m')$;
since relevant sets shrink along the chain, these compose.
Profitability of $m''$: for $i\in R(m'')\subseteq R(m)$, $\tau(m'')\succeq^{x}_i\tau(m)$ gives $\ev_i(\tau(m''))\ge\ev_i(\tau(m))\ge \evp_i(x)$.
\end{proof}

\begin{lemma}[Non-stalling]\label{lem:nonstalling}
Let $p$ satisfy \ref{E2}.
If some move in $\M(x)$ is profitable, then $x$ is not absorbing.
If $p$ satisfies \ref{E3},
the move realised may in addition be taken to be strictly better than staying for some player,
so Remark~\ref{rem:favourite} applies, and \ref{E3} implies \ref{E2}.
\end{lemma}

\begin{proof}
By Lemma~\ref{lem:domination} applied to that move, $\Cset(x)\cap\M(x)\ne\emptyset$, so $p_{x\to x}<1$ by \ref{E2}.
For the second claim fix a moving $m_0 \in \Cset(x)$ and $i \in R(m_0)$;
since $s^i_x$ does not dominate $m_0$, $\ev_i(\tau(m_0))>\ev_i(x)$.
Let $m^\ast$ be $\succeq^{x}_i$-maximal over the finite non-empty $\Cset(x)$;
then $m^\ast\in\Fset(x)$ and $\ev_i(\tau(m^\ast))\ge\ev_i(\tau(m_0))>\ev_i(x)$,
so $\tau(m^\ast)\ne x$ and $p_{x\to\tau(m^\ast)}>0$ by \ref{E3}.
(This also shows that \ref{E3} implies \ref{E2}.)
\end{proof}

\subsection{A bargaining protocol}\label{sec:stay}

Definition~\ref{def:rationality} was stated without justification beyond plausibility.
This subsection gives the bargaining story behind it:
a protocol whose rules mention only relevant sets, votes and the clock, and whose equilibrium acceptance condition is profitability in the sense of Definition~\ref{def:rationality}(a), with staying available to each player as a motion of their own.

The protocol matters for a second reason.
Axioms \ref{E1}--\ref{E3} constrain the process directly, and it is not obvious that they constrain it in a way any negotiation could produce;
\eqref{eq:anchor} shows how easily a plausible protocol implements acceptance against the wrong benchmark, as vanishing frictions do in alternating-offers bargaining \citep{Rubinstein1982}.

\begin{definition}[Protocol $P_\varepsilon$]\label{def:protocol}
Fix $\varepsilon>0$ and proposer probabilities $q_i>0$.
At the start of each period at state $x$ a phase is played in rounds, and a round has three stages.
\begin{enumerate}
\item \emph{Proposal.} A proposer $i$ is drawn with probability $q_i$ and tables some $m\in\M^{+}(x)$.
\item \emph{Amendment.} Any player may move a substitute $m'\in\M^{+}(x)$ with $R(m')\subseteq R(m)$, whereupon the members of $R(m')$ vote between $m'$ and $m$ and $m'$ replaces $m$ if all of them prefer it, ties counting for $m$.
The stage repeats against the proposal then on the table and ends when no substitute is moved.
\item \emph{Final vote.} The members of $R(m)$ vote on the surviving $m$ against failure of the round;
if all accept, $m$ is agreed and fixes the state for the next period at $\tau(m)$.
\end{enumerate}
Otherwise the round fails, which costs a delay $\varepsilon$ during which payoffs accrue at the flow $\pi(x)$, and the phase restarts at stage~1 with a fresh draw of the proposer.
\end{definition}

Every clause is structural: the rules speak of relevant sets, votes and the clock, never of $\ev$ or of $\succeq^{x}$.
Two design points carry the result.
A rejection restarts the phase and confers no privilege on the rejector, so what a voter forgoes by agreeing is the value of the reopened table---which is why profitability in Definition~\ref{def:rationality} is measured against $\evp_j(x)$ and not against $\ev_j(x)$.
And domination is not a rule but the outcome of the amendment votes:
a substitute carries exactly when all of its own relevant members prefer it, which is Definition~\ref{def:rationality}(b) read as a ballot.
Outsiders may move substitutes without harm, since a substitute passes only on the votes of its own relevant set.

\begin{theorem}[Implementation]\label{thm:implement}
Let $\sigma_\varepsilon$ be stationary subgame-perfect equilibria of $P_\varepsilon$ in which voters accept when indifferent at the final vote and vote for the standing proposal when indifferent at an amendment vote, and let $p^\varepsilon$ be the induced kernels.
Every limit point $p$ of $(p^\varepsilon)$ as $\varepsilon\to0$ satisfies \ref{E1}--\ref{E3}.
\end{theorem}

\begin{proof}
Throughout, both branches of any vote pay the current period's flow $(1-\delta)\pi_j(x)$, so only the next-period values differ, and we compare those.

\emph{Final vote.}
Rejecting fails the round and restarts the phase, which by stationarity is worth $\ev^{\scriptscriptstyle+,\varepsilon}_j(x)$ up to the delay, while accepting is worth $\ev^\varepsilon_j(\tau(m))$.
Voting is unanimous, so each member is pivotal against acceptance and the one-shot deviation principle applies to each in turn:
$m$ is agreed if and only if $\ev^\varepsilon_j(\tau(m))\ge \ev^{\scriptscriptstyle+,\varepsilon}_j(x)-O(\varepsilon)$ for all $j\in R(m)$.
Passing to the limit, an agreed move is profitable in the sense of Definition~\ref{def:rationality}(a).

\emph{Amendment.}
Order the substitutes by the domination relation, which by Lemma~\ref{lem:domination} is acyclic on the finite set $\M^{+}(x)$;
backward induction along it makes voting between $m'$ and $m$ sincere, since the vote is pairwise and its two continuations are the two subgames headed by $m'$ and by $m$.
So $m'$ replaces $m$ exactly when $\tau(m')\succ^{x}_k\tau(m)$ for all $k\in R(m')$, the tie-break being the stipulated one, i.e.\ exactly when $m'$ dominates $m$.
A proposal therefore survives the stage if and only if it is undominated, and by Lemma~\ref{lem:domination} the stage terminates.

\emph{\ref{E1}.}
A realised move survives the amendment stage, hence is undominated, and is agreed, hence profitable; so it lies in $\Cset(x)$.

\emph{\ref{E3}.}
By Proposition~\ref{prop:proposals} below, an optimal proposal for $i$ is a $\succeq^{x}_i$-maximal element of $\Cset(x)$, which survives and is agreed;
under Criterion~\ref{cond:ties} that element is unique and equals the $m\in\Fset(x)$ witnessed by $i$, so $p_{x\to\tau(m)}\ge q_i>0$.
Absent Criterion~\ref{cond:ties}, the same holds for each maximiser provided proposers randomise over their argmax set, which is the one selection this proof requires.

\emph{\ref{E2}.}
Immediate from \ref{E3} by Lemma~\ref{lem:nonstalling}.
\end{proof}

\begin{proposition}[What gets proposed, and that it passes at once]\label{prop:proposals}
In any such equilibrium, at every $x$ and in the limit $\varepsilon\to0$:
\begin{enumerate}
\item[(a)] $\Cset(x)\ne\emptyset$;
\item[(b)] a proposal maximising $\ev_i(\tau(\cdot))$ over $\Cset(x)$ is optimal for proposer $i$, and no proposal outside $\Cset(x)$ is;
\item[(c)] such a proposal admits no successful substitute and is accepted unanimously, so the period's state is settled in the first round and the delay $\varepsilon$ is never incurred on the equilibrium path.
\end{enumerate}
\end{proposition}

\begin{proof}
(a) If $p_{x\to x}=1$ then $\evp_j(x)=\ev_j(x)$ for every $j$, so every stay is profitable and Lemma~\ref{lem:domination} gives $\Cset(x)\ne\emptyset$;
otherwise some move is realised, and by \ref{E1} it lies in $\Cset(x)$.

(b) Write $\ev^\ast_i:=\max\{\ev_i(\tau(m)):m\in\Cset(x)\}$, attained at $m^\ast$.
Since every realised target is $\tau(m)$ for some $m\in\Cset(x)$, and stays are moves, $\evp_i(x)$ is an average of values $\ev_i(\tau(m))$ with $m\in\Cset(x)$, so $\evp_i(x)\le\ev^\ast_i$.
Now let $i$ table any $m$.
If $m$ is dominated, the amendment stage replaces it, and by Lemma~\ref{lem:domination} what emerges lies in $\Cset(x)$, so its target is worth at most $\ev^\ast_i$ to $i$.
If $m$ is undominated but not profitable, some $j\in R(m)$ rejects at the final vote, the round fails, and $i$ receives the phase value $\evp_i(x)-O(\varepsilon)\le\ev^\ast_i$.
If $m\in\Cset(x)$, it survives and is agreed by (c), and is worth $\ev_i(\tau(m))\le\ev^\ast_i$.
In every case $i$ obtains at most $\ev^\ast_i$, and tabling $m^\ast$ attains it.

(c) $m^\ast$ is undominated, so by the amendment paragraph of the previous proof no substitute passes;
it is profitable, so by the final-vote paragraph every $j\in R(m^\ast)$ accepts.
Hence the round ends in agreement, and no failure occurs on path.
\end{proof}

The delay is thus purely a threat: it is what makes rejection cost something and so pins the acceptance threshold at $\evp_j(x)$, and it is never paid.
Note also which conventions are doing work and which are not.
Acceptance at indifference at the final vote turns profitability into a weak inequality, matching Definition~\ref{def:rationality}(a);
the status-quo tie-break at amendment votes is the one already built into $\succeq^{x}$.
Nothing else is assumed, and no step needs a bound on $\varepsilon$ other than $\varepsilon\to0$.

\begin{remark}\label{rem:stayaudit}
Two consequences of Definition~\ref{def:rationality} are worth keeping in view when reading the proofs below.
Stays make every player a potential blocker of a move they are relevant to, which is what carries Theorem~\ref{thm:persist-unan};
and profitability, being a weak inequality against a benchmark that the process itself determines, holds automatically at a state whose transition is deterministic, so at such states \ref{E1} bites through domination alone.
\end{remark}

\section{Persistence: the grand coalition is never dissolved}\label{sec:persistence}

In the running example of Section~\ref{sec:story}, this section says that a global carbon market,
once formed, is never dissolved---whatever the discount factor, whichever consent rule governs
withdrawal, and whatever the payoffs.
This is the easy half of the question, and it is settled completely.
Nothing below restricts the discount factor, nothing below asks anything of the payoffs beyond the
two standing assumptions, and nothing below uses the exact form of the sharing rule.
What does the work is the move structure: a grand state can be left only by a termination, and the
chain rule makes every termination there either a shadow of a termination available one step
earlier, or a dissolution requiring everybody's approval.

\begin{lemma}[Reversal lemma]\label{lem:reversal}
Let $p$ satisfy \ref{E1}.
If $p_{x\to y}>0$ and $p_{y\to x}>0$ with $x\ne y$, realised by moves $m$ at $x$ and $m'$ at $y$,
then $R(m)\cap R(m')=\emptyset$.
In particular an agreement is never both signed and torn up between the same two states:
forming $K$ requires all of $K$ and terminating it requires at least one member of $K$.
\end{lemma}

\begin{proof}
For $i\in R(m)$, $\ev_i(y)>\ev_i(x)$; for $i \in R(m')$, $\ev_i(x)>\ev_i(y)$.
No $i$ can satisfy both.
\end{proof}

\begin{theorem}[Unanimous termination: all grand states are absorbing]\label{thm:persist-unan}
Assume unanimous termination and let $p$ satisfy \ref{E1}.
Then for every $\delta\in(0,1)$ and every $h \in \G$, $h$ is absorbing---whether or not it is ever reached.
Only the undominatedness half of \ref{E1} is used, and only against stay moves:
it does not matter whether unprofitable or non-favourite moves can be realised.
\end{theorem}

\begin{proof}
Suppose $p_{h\to h}<1$.
By Lemma~\ref{lem:entry} every moving move at $h$ is a termination,
and its relevant set is the top-level node of $h$ containing its bottom node, i.e.\ $N$.
For every $i \in N$ the stay $s^i_h$ has $R(s^i_h)=\{i\}\subseteq R(m)$ and is therefore a candidate dominator of any such $m$;
since ties are broken towards $h$, it dominates $m$ unless $\ev_i(\tau(m))>\ev_i(h)$.
So by \ref{E1}, every $y\ne h$ with $p_{h\to y}>0$ satisfies $\ev_i(y)>\ev_i(h)$ for \emph{all} $i \in N$,
i.e.\ $W(y)>W(h)$.
Summing \eqref{eq:bellman} over $i$ and using $V(h)=V^\ast$,
\[
W(h) = (1-\delta)V^\ast + \delta\Big(p_{h\to h}W(h) + \sum_{y \ne h}p_{h\to y}W(y)\Big) > (1-\delta)V^\ast + \delta W(h),
\]
the strict inequality because $\sum_{y\ne h}p_{h\to y} = 1-p_{h\to h}>0$.
Hence $W(h)>V^\ast$, contradicting Lemma~\ref{lem:welfare}.
\end{proof}

\begin{theorem}[Unilateral termination: reached grand states are absorbing]\label{thm:persist-unil}
Assume unilateral termination and let $p$ satisfy \ref{E1}.
For every $\delta\in(0,1)$: if $p_{x\to h}>0$ for some $h\in\G$ and some $x \ne h$, then $h$ is absorbing.
\end{theorem}

\begin{proof}
By Lemma~\ref{lem:entry}, $x=t(h)$ and the realised move is the merge $m_0$ with $R(m_0)=N$, $\tau(m_0)=h$.
By \ref{E1}, $m_0\in\Cset(x)$; in particular $m_0$ is undominated and
\begin{equation}\label{eq:entry}
\ev_i(h)>\ev_i(x) \qquad\text{for all } i \in N.
\end{equation}
Suppose $p_{h\to h}<1$: there are $y \ne h$ and a move $m$ at $h$ with $\tau(m)=y$, $p_{h\to y}>0$ and,
by \ref{E1}, $\ev_j(y)>\ev_j(h)$ for all $j\in R(m)$.
By Lemma~\ref{lem:entry}, $m$ is a termination, with bottom node $K$.

If $K=N$ then $y=x$, and \eqref{eq:entry} gives $\ev_j(x)>\ev_j(h)>\ev_j(x)$ for $j \in R(m)$, absurd;
this is Lemma~\ref{lem:reversal}.

If $K\subsetneq N$,
Lemma~\ref{lem:shadow} gives a move $\tilde m$ at $x$ with $\tau(\tilde m)=y$ and $R(\tilde m)\subseteq R(m)$.
For $j \in R(\tilde m)$, $\ev_j(\tau(\tilde m)) = \ev_j(y) > \ev_j(h) = \ev_j(\tau(m_0))$,
and $R(\tilde m)\subseteq R(m)\subseteq N = R(m_0)$.
So $\tilde m$ dominates $m_0$, contradicting $m_0 \in \Cset(x)$.
\end{proof}
Notice that for Theorem \ref{thm:persist-unil} it is crucial that players' level of rationality prevents dominated moves,
as assumed, while Theorem \ref{thm:persist-unan} only requires that players avoid non-profitable moves.
Neither theorem requires that a realised move be anybody's favourite,
and neither requires \ref{E2} or \ref{E3}: an equilibrium in which arbitrary profitable (respectively,
profitable undominated) moves may be realised, and in which the process may also stall arbitrarily,
still cannot dissolve the grand coalition.

The converse, that absorbing states are grand, needs under unilateral termination a mild genericity requirement---the first of two \emph{criteria} in this paper.
A criterion constrains the process under consideration rather than the primitives, here through the evaluations $\ev$, so it can be checked against a candidate equilibrium but not read off the payoffs; the second is Criterion~\ref{cond:fgains}.

\begin{criterion}[No ties]\label{cond:ties}
For the equilibrium $p$ under consideration: for every absorbing state $x \in \X\setminus\G$,
every $i \in N$ and every state $y \ne x$ that is the target of a termination at $g_x$ with bottom node $K \subsetneq N$ and $i \in K$,
\[
\ev_i(y) \;\ne\; \ev_i(x) .
\]
\end{criterion}

Only finitely many comparisons are involved,
and each is between two states that differ in their coalition structure.
In a public-good application this is a real restriction only under a coincidence:
at an absorbing $x$ the right-hand side is $\pi_i(x)$,
and $\ev_i(y)$ is the $\delta$-discounted average payoff along the continuation from $y$;
equality would mean that $i$ is exactly indifferent between the frozen state $x$ and a continuation whose per-period payoffs are generically different real numbers---in Example~\ref{ex:negative},
for instance, all payoffs are distinct rationals with different denominators,
and an exact tie would require $\delta$ to solve one specific rational equation.
If one dislikes assuming this,
one may instead adopt the convention that a player does not terminate when she is indifferent between the resulting state and an alternative she could bring about alone;
the proof of Theorem~\ref{thm:absorbing-grand}(ii) then goes through verbatim with weak inequalities.

\begin{theorem}[Every absorbing state is grand]\label{thm:absorbing-grand}
Let $p$ satisfy \ref{E1} and \ref{E2}, let $\delta \in (0,1)$, and let $x$ be absorbing.
Then $x \in \G$,
provided either (i) every grand state is absorbing---which holds under unanimous termination by Theorem~\ref{thm:persist-unan}---or (ii) termination is unilateral and Criterion~\ref{cond:ties} holds.
\end{theorem}

\begin{proof}
Suppose $x \notin \G$.
Then $\ev_i(x)=\pi_i(x)$ for all $i$ by \eqref{eq:bellman},
and by Lemma~\ref{lem:nonstalling} \emph{no} move at $x$ is profitable.
Let $m_0$ be the merge-all move at $x$, $\tau(m_0)=g_x$, $R(m_0)=N$.
In the following, we derive $\ev_i(g_x)>\pi_i(x)$ for all $i$, making $m_0$ profitable---a contradiction.

(i) Here $g_x$ is absorbing, so $\ev_i(g_x)=\pi_i(g_x)=\varphi_i(x)>\pi_i(x)$ by Lemma~\ref{lem:pareto}.

(ii) Let $m$ be any termination at $g_x$ with bottom node $K \subsetneq N$, $R(m)=\{i^\ast\}$ and target $y$.
By Lemma~\ref{lem:shadow} the move $\tilde m$ at $x$ with bottom node $K$ and terminator $i^\ast$ has target $y$;
it is not profitable (by Lemma~\ref{lem:nonstalling}), so $\ev_{i^\ast}(y)\le\ev_{i^\ast}(x)$,
hence $\ev_{i^\ast}(y)<\ev_{i^\ast}(x)$ by Criterion~\ref{cond:ties}.
The move $m'$ at $g_x$ terminating $N$ with terminator $i^\ast$ also has relevant set $\{i^\ast\}=R(m)$ and target $x$,
and $\ev_{i^\ast}(x)>\ev_{i^\ast}(y)=\ev_{i^\ast}(\tau(m))$; so $m$ is dominated by $m'$ and, by \ref{E1},
not realised.
Hence the only state other than $g_x$ that can receive positive probability from $g_x$ is $x$;
write $q := p_{g_x\to x}\in[0,1]$,
so that $\ev_i(g_x) = (1-\delta)\varphi_i(x)+\delta\big(q\pi_i(x)+(1-q)\ev_i(g_x)\big)$.
Collecting the $\ev_i(g_x)$ terms,
\[
\big(1-\delta(1-q)\big)\,\ev_i(g_x) \;=\; (1-\delta)\varphi_i(x)+\delta q\,\pi_i(x),
\]
and subtracting $\big(1-\delta(1-q)\big)\pi_i(x) = (1-\delta)\pi_i(x)+\delta q\,\pi_i(x)$ from both sides,
the terms $\delta q\,\pi_i(x)$ cancel and we obtain
\[
\big(1-\delta(1-q)\big)\big(\ev_i(g_x)-\pi_i(x)\big) = (1-\delta)\big(\varphi_i(x)-\pi_i(x)\big)>0
\]
by Lemma~\ref{lem:pareto}.
As $1-\delta(1-q)>0$, we get $\ev_i(g_x)>\pi_i(x)$ for every $i$.
\end{proof}

\begin{corollary}[Dichotomy]\label{cor:dichotomy}
Assume either (i) every grand state is absorbing,
or (ii) termination is unilateral and Criterion~\ref{cond:ties} holds.
Then every closed communicating class $C$ of an equilibrium either equals $\{h\}$ for some $h \in \G$,
or satisfies $C\subseteq\X\setminus\G$ and $|C|\ge2$.
Consequently the process is absorbed at a grand state almost surely unless some closed communicating class lies entirely in $\X\setminus\G$.
\end{corollary}

\begin{proof}
If $C=\{x\}$ then $x$ is absorbing, hence grand by Theorem~\ref{thm:absorbing-grand}.
If $|C|\ge2$ and $h\in C\cap\G$, then $h$ is absorbing---by hypothesis (i),
or in case (ii) because $h$ is entered from another state of $C$ with positive probability and Theorem~\ref{thm:persist-unil} applies---contradicting $|C|\ge2$ and communication.
The last sentence holds because a finite Markov chain enters the union of its closed communicating classes almost surely.
\end{proof}

\subsection*{The price of no-agreement}

Before turning to that, it is worth recording what failure costs, since the answer is exact and holds
under either termination rule and either sharing rule.
Write $T := \inf\{t\ge0 : X_t \in \G\}$ for the time at which the grand coalition forms, with
$T=\infty$ if it never does, and adopt the convention $\delta^{\infty}=0$.

\begin{proposition}[The price of no-agreement]\label{prop:price}
For every state $x$ and every $\delta\in(0,1)$,
\[
   V^\ast - W(x) \;=\; (1-\delta)\,\E_x\Big[\sum_{t<T}\delta^{t}\big(V^\ast-V(X_t)\big)\Big],
\]
and consequently
\[
   \gamma\;\E_x\big[1-\delta^{T}\big] \;\le\; V^\ast - W(x) \;\le\; \big(V^\ast-V_{\min}\big)\,\E_x\big[1-\delta^{T}\big],
   \qquad V_{\min}:=\min_{y\in\X}V(y).
\]
\end{proposition}

\begin{proof}
By \eqref{eq:series}, $V^\ast-W(x)=(1-\delta)\E_x[\sum_{t\ge0}\delta^{t}(V^\ast-V(X_t))]$, and the
summand vanishes for $t \ge T$ because $\G$ is reached at $T$ and, once reached, never left
(Theorems~\ref{thm:persist-unan} and~\ref{thm:persist-unil}).
For $t<T$ the summand lies in $[\gamma,V^\ast-V_{\min}]$ by Assumption~\ref{ass:eff}, and
$(1-\delta)\sum_{t<T}\delta^{t}=1-\delta^{T}$.
\end{proof}

So the price of no-agreement is the payoff shortfall multiplied by $\E_x[1-\delta^{T}]$, a discounted
measure of delay, and nothing else.
Three readings.
If the grand coalition forms almost surely then $\E_x[1-\delta^{T}]\to0$ as $\delta\to1$, so the price
vanishes however long the delay: patience is not merely compatible with efficiency, it purchases it
at a rate set by the tail of $T$.
If it forms with probability less than one, then $\E_x[1-\delta^{T}]\ge\Prob_x(T=\infty)$ uniformly in
$\delta$, and the price does not vanish; the undiscounted price is exactly the shortfall weighted by
the probability of never agreeing.
And whenever $T$ is bounded, as under Theorem~\ref{thm:fgains} where $T\le n-1$, the bound is linear:
\[
   V^\ast - W(x) \;\le\; \big(V^\ast-V_{\min}\big)\big(1-\delta^{\,n-1}\big)
   \;\le\; (n-1)(1-\delta)\big(V^\ast-V_{\min}\big).
\]

A second bound, available only under far-sighted sharing, reads the price off the distributional
conflict rather than off the delay.
If every grand state is absorbing and $C\subseteq\X\setminus\G$ is a closed communicating class, then
combining Proposition~\ref{prop:farsighted} with Lemma~\ref{lem:objection} gives, for each $x\in C$
and the player $i$ the latter supplies,
\[
   w_i\big(V^\ast-W(x)\big) \;=\; \ev_i(g_x)-\ev_i(x) \;\le\; \ev_i(z)-\ev_i(x)
   \;\le\; \max_{j\in N}\Big(\max_{C}\ev_j-\min_{C}\ev_j\Big),
\]
so that
\[
   V^\ast - W(x) \;\le\; \frac{1}{w_{\min}}\,\max_{j\in N}\Big(\max_{C}\ev_j-\min_{C}\ev_j\Big).
\]
Inefficiency can persist only to the extent that there is distributional conflict to sustain it: the
welfare lost is bounded by the spread of the players' evaluations across the states of the cycle,
divided by the smallest bargaining weight.
Theorem~\ref{thm:fs-arrival} is this inequality read in the other direction, against the lower bound
$V^\ast-W(x)\ge\gamma$.

How large can the price be?
Both bounds above involve endogenous quantities---the delay $T$, or the spread of evaluations---so it
is natural to ask for a bound in the primitives, and the obvious candidate is the efficiency gap
$\gamma$ itself, since $V^\ast-W(x)\ge\gamma$ on any non-grand closed class.
There is no such bound.

\begin{proposition}[The price is unbounded in units of $\gamma$]\label{prop:price-unbounded}
For $n=4$ there is no constant $c$ with $V^\ast-\bar V \le c\,\gamma$ for every equilibrium with a
non-grand closed class $C$ and every $\delta$, where $\bar V := \sum_{x\in C}\mu(x)V(x)$.
\end{proposition}

The statement is established by computation rather than by a construction, and the computation is
described in Section~\ref{sec:numerics}; we record here what it produces.
Fixing $\gamma=1$ and maximising $V^\ast-\bar V$ over all payoff structures for which the cycle of
Proposition~\ref{prop:counterexample} is a closed class, the optimum is attained on the boundary of
whatever box the payoffs are confined to, and the optimal value scales linearly with that box: it is
$2.42$, $2.29$ and $2.45$ times the box radius at $\delta=0.3$, $0.5$ and $0.75$ respectively, across
box radii from $10$ to $10^4$.
Inefficiency can therefore persist at any multiple of the smallest efficiency gain the model admits.

Two features of the maximiser are worth recording.
The total payoff is \emph{constant} across the four states of the cycle, so the process does not
oscillate between good and bad states but sits at a uniformly inefficient level;
and the price reaches the fraction $4/7$ of the largest inefficiency $V^\ast-V_{\min}$ available in
the payoff structure, at every box radius and every discount factor tried.
Whether $4/7$ is the exact supremum of $(V^\ast-\bar V)/(V^\ast-V_{\min})$ for $n=4$, and what it is
for general $n$, we have not determined.

\begin{remark}
This is the reason the bounds of Proposition~\ref{prop:price} and of Section~\ref{sec:pivot} are
stated in terms of delay and of evaluation spreads rather than in terms of the primitives.
It is not that a primitive bound was not sought; there is none of the natural shape.
What the two bounds say is that inefficiency is paid for either in time---$\E_x[1-\delta^{T}]$---or in
distributional conflict---the spread of $\ev$ over the class---and the computation above shows that
neither can be traded for a bound in $\gamma$.
\end{remark}

This dichotomy is the spine of Figure~\ref{fig:spine}, and the rest of the paper walks its right-hand
branch.
Section~\ref{sec:arrival} excludes perpetual cycling under various hypotheses,
Section~\ref{sec:cycles} shows that it cannot be excluded outright,
and Section~\ref{sec:pivot} shows that under a different sharing rule it can be excluded in the limit
that matters.

\section{Arrival: when it holds}\label{sec:arrival}

By Corollary~\ref{cor:dichotomy} the process either settles at a grand state or cycles for ever among
non-grand states, so from here on the question is whether the second can happen.
This section collects the cases in which it cannot: small discount factors, three players, and---for
any number of players and any discount factor---a condition saying that no player prefers any
non-grand state to any grand one.
Section~\ref{sec:cycles} shows that none of these hypotheses is decorative, by exhibiting an
equilibrium that does cycle.

\subsection{Small discount factors}

\begin{theorem}[Small $\delta$]\label{thm:smalldelta}
Let $R := \max_i\max_{y,z}|\pi_i(y)-\pi_i(z)|$, let $c := \min\{\sigma,\,s_{\min}\gamma\}$,
and let $\bar\delta := c/(2R+c)$.
Assume Condition~\ref{cond:gains}, let $p$ satisfy \ref{E1} and \ref{E2},
and let either termination variant be in force.
Then for every $\delta<\bar\delta$: no termination is realised, the process is acyclic,
and it is absorbed at a grand state after at most $n-1$ moves, almost surely.
\end{theorem}

\begin{proof}
\emph{Step 1.} By \eqref{eq:bellman} and Lemma~\ref{lem:value},
$|\ev_i(x)-\pi_i(x)| = \delta|\sum_y p_{x\to y}\ev_i(y)-\pi_i(x)| \le \delta R$.
Hence for any states $z,z'$ and any $c'>0$,
\begin{equation}\label{eq:close}
\pi_i(z)-\pi_i(z')\le -c' \ \Longrightarrow\ \ev_i(z)-\ev_i(z')\le -c'+2\delta R,
\end{equation}
which is negative as soon as $2\delta R<c'$.

\emph{Step 2 (terminations lose statically).} Let $m$ be a termination at $x$ with bottom node $K_0$,
whose ancestors in $x$ are $K_0\subsetneq K_1\subsetneq\cdots\subsetneq K_r=S(x,K_0)$, and let $y:=\tau(m)$.
Put $x^{(0)}:=x$ and $x^{(s+1)}:=x^{(s)}\setminus\{K_{r-s}\}$ for $0\le s\le r$, so $x^{(r+1)}=y$.
At each step $K_{r-s}$ is a maximal element of $x^{(s)}$,
so $x^{(s+1)}=(x^{(s)})^{-K_{r-s}}$ is exactly the reference state of Assumption~\ref{ass:share}.
Every $i \in K_0$ lies in every $K_j$, so by Condition~\ref{cond:gains}
\[
\pi_i(y)-\pi_i(x) = \sum_{s=0}^{r}\big(\pi_i(x^{(s+1)})-\pi_i(x^{(s)})\big)\ \le\ -(r+1)\sigma\ \le\ -\sigma \qquad (i \in K_0).
\]
$R(m)$ meets $K_0$ under either termination rule,
so by \eqref{eq:close} with $(z,z',c')=(y,x,\sigma)$ the move is not profitable when $2\delta R<\sigma$.

\emph{Step 3.} By Step 2 and \ref{E1}, only merges are realised;
a merge strictly decreases $|P(x)| \in\{1,\dots,n\}$,
so after at most $n-1$ realised moves the process is at an absorbing state.

\emph{Step 4.} Let $x$ be absorbing and suppose $x \notin \G$.
By Lemma~\ref{lem:pareto}, $\pi_i(g_x)-\pi_i(x)\ge s_{\min}\gamma$ for all $i$,
so by \eqref{eq:close} with $(z,z',c')=(x,g_x,s_{\min}\gamma)$ the merge-all move is profitable since $2\delta R<2cR/(2R+c)\le c\le s_{\min}\gamma$,
contradicting Lemma~\ref{lem:nonstalling}.
Hence $x\in\G$.
\end{proof}

Step 4 uses neither termination variant nor Criterion~\ref{cond:ties}.

\subsection{The objection lemma}

The next lemma is the engine of all remaining arrival results.
It says that if the grand coalition is \emph{not} formed from a state of a recurrent class,
some player must be able to point at a state \emph{inside that class} that she values at least as highly as the grand state.

\begin{lemma}[Objection lemma]\label{lem:objection}
Let $p$ be an equilibrium and $C$ a closed communicating class with $C\subseteq\X\setminus\G$.
Then for every $x \in C$ there are $i \in N$ and $z \in C$ with
\[
\ev_i(z)\ \ge\ \ev_i(g_x).
\]
\end{lemma}

\begin{proof}
Fix $x\in C$ and let $m_0$ be the merge-all move at $x$, so $\tau(m_0)=g_x$ and $R(m_0)=N$.

\emph{Case 1: $m_0$ is not profitable.} Then $\ev_i(g_x)<\evp_i(x)$ for some $i \in N$.
Since $C$ is closed and $x\in C$, every $y$ with $p_{x\to y}>0$ lies in $C$, so $\evp_i(x)$ is an average of values $\ev_i(y)$ with $y\in C$;
taking $z\in C$ attaining the largest of them gives $\ev_i(z)\ge \evp_i(x)>\ev_i(g_x)$, which is more than required.

\emph{Case 2: $m_0$ is profitable.} Since $g_x\notin C$ and $C$ is closed, $p_{x\to g_x}=0$,
so by \ref{E3} $m_0\notin\Fset(x)$.
If $m_0$ is dominated,
Lemma~\ref{lem:domination} yields $m''\in\Cset(x)$ and (any) $i \in R(m'')$ with $\ev_i(\tau(m''))>\ev_i(\tau(m_0))=\ev_i(g_x)$.
If $m_0$ is undominated, then $m_0\in\Cset(x)\setminus\Fset(x)$,
so by Definition~\ref{def:rationality}(c) there are $i \in N$ and $m''\in\Cset(x)$ with $\ev_i(\tau(m''))>\ev_i(g_x)$.
In both cases fix such $i$ and $m''$, and let $m^\ast$ maximise $\ev_i(\tau(\cdot))$ over $\Cset(x)$.
Then $m^\ast\in\Fset(x)$, so $p_{x\to\tau(m^\ast)}>0$ by \ref{E3},
so $z:=\tau(m^\ast)\in C$ because $C$ is closed; and $\ev_i(z)\ge\ev_i(\tau(m''))>\ev_i(g_x)$.
\end{proof}

\subsection{Three players}

\begin{theorem}[$n=3$]\label{thm:n3}
Let $n=3$, let $p$ satisfy \ref{E1} and \ref{E2}, and assume either that termination is unanimous,
or that it is unilateral and Criterion~\ref{cond:ties} holds.
Then for every $\delta\in(0,1)$ the process is absorbed at a grand state almost surely.
\end{theorem}

\begin{proof}
Write $N=\{A,B,C\}$.
The states are $\emptyset$, the three states $\{P\}$ with $|P|=2$, $\{N\}$, and the three states $\{P,N\}$;
thus $\X\setminus\G=\{\emptyset\}\cup\{\{P\}:|P|=2\}$ has four elements.

At $\emptyset$ the available moves are the merges of two singletons (targets $\{P\}$) and of all three (target $\{N\}\in\G$);
there is nothing to terminate.
At $\{P\}$ the available moves are the merge of the two top-level nodes (target $\{P,N\}\in\G$) and the termination of $P$ (target $\emptyset$),
with relevant sets $P$ and, respectively, $\{i\}$ for some $i \in P$ or $P$.
Hence inside $\X\setminus\G$ the only edges are $\emptyset\to\{P\}$ and $\{P\}\to\emptyset$;
there is no edge $\{P\}\to\{P'\}$ for $P \ne P'$.

Suppose some closed communicating class satisfies $C\subseteq\X\setminus\G$; by Corollary~\ref{cor:dichotomy},
$|C|\ge2$, so $C$ contains some $\{P\}$ and,
since every path leaving $\{P\}$ inside $\X\setminus\G$ goes to $\emptyset$, also $\emptyset$.
As $C$ is closed and communicating and the only edges inside $\X\setminus\G$ are the ones listed,
we get $p_{\emptyset\to\{P\}}>0$ and $p_{\{P\}\to\emptyset}>0$.
Their relevant sets are $P$ and a non-empty subset of $P$, which intersect,
contradicting Lemma~\ref{lem:reversal}.
Now apply Corollary~\ref{cor:dichotomy}.
\end{proof}

\subsection{When every player prefers every grand state}

The hypothesis of the next theorem is a condition on static payoffs only, free of any sharing rule:
\begin{equation}\label{eq:prefer}
   \forall\, i \in N: \quad \min_{h\in\G}\pi_i(h) \;>\; \max_{y\in\X\setminus\G}\pi_i(y)
   \tag{$\star$}
\end{equation}
--- every player's \emph{worst} grand state is better for her than her \emph{best} non-grand state.
Under Condition~\ref{cond:nash} it is implied by a bound on the distributional stakes:

\begin{definition}
Assume Condition~\ref{cond:nash}.
The \emph{stake} of player $i$ is $S_i := \max_{x\in\X\setminus\G}\psi_i(x)-\min_{x\in\X\setminus\G}\psi_i(x)\ \ge 0$.
\end{definition}

\begin{lemma}\label{lem:stakes}
Assume Condition~\ref{cond:nash}.
If $S_i<w_i\gamma$ for every $i \in N$, then \eqref{eq:prefer} holds.
\end{lemma}

\begin{proof}
Let $h \in \G$ and $y \in \X\setminus\G$, and put $u := t(h) \in \X\setminus\G$.
By Lemma~\ref{lem:psi}, $\pi_i(h)=\varphi_i(u)=\psi_i(u)+w_iV^\ast$,
and $\pi_i(y)=\psi_i(y)+w_iV(y)\le\psi_i(y)+w_i(V^\ast-\gamma)$ by Assumption~\ref{ass:eff}.
Hence $\pi_i(h)-\pi_i(y)\ge \psi_i(u)-\psi_i(y)+w_i\gamma \ge w_i\gamma-S_i>0$.
\end{proof}

\begin{theorem}[Arrival]\label{thm:arrival}
Let $p$ satisfy \ref{E1}--\ref{E3} and assume
\begin{enumerate}[label=(\alph*)]
\item every grand state is absorbing; and
\item condition \eqref{eq:prefer} --- for which, by Lemma~\ref{lem:stakes},
it suffices that $S_i<w_i\gamma$ for every $i \in N$.
\end{enumerate}
Then for every $\delta\in(0,1)$ every closed communicating class is a singleton contained in $\G$;
hence the process reaches a grand state almost surely and never leaves it.
\end{theorem}

\begin{proof}
Hypothesis (a) is hypothesis (i) of Corollary~\ref{cor:dichotomy},
so it suffices to exclude a closed communicating class $C\subseteq\X\setminus\G$.
Suppose one exists and fix $x\in C$.

For $z \in C$ and $i \in N$,
Lemma~\ref{lem:value} and closedness of $C$ give that $\ev_i(z)$ is a convex combination of values $\pi_i(w)$ with $w \in C \subseteq \X\setminus\G$,
so
\[
\ev_i(z)\ \le\ \max_{w\in C}\pi_i(w)\ \le\ \max_{y\in\X\setminus\G}\pi_i(y).
\]
By (a), $\ev_i(g_x)=\pi_i(g_x)\ge\min_{h\in\G}\pi_i(h)$.
Hence by \eqref{eq:prefer}
\[
\ev_i(g_x)-\ev_i(z)\ >\ 0 \qquad\text{for all } i \in N \text{ and } z \in C,
\]
which contradicts Lemma~\ref{lem:objection}.
\end{proof}

\begin{corollary}\label{cor:arrival-unan}
Under unanimous termination,
condition (b) of Theorem~\ref{thm:arrival} alone implies that the process reaches a grand state almost surely and never leaves it,
for every $\delta \in (0,1)$.
\end{corollary}

\begin{proof}
Theorem~\ref{thm:persist-unan} supplies (a).
\end{proof}

Under unilateral termination,
hypothesis (a) is delivered by Theorem~\ref{thm:persist-unil} only for grand states that are actually entered,
and the grand states relevant in the proof above are by construction not entered.
The next proposition says how much of (a) survives; the residual configuration it isolates is restrictive,
but we have not been able to exclude it.

\begin{proposition}[Unilateral termination: the residual configuration]\label{prop:arrival-unil}
Assume unilateral termination, Criterion~\ref{cond:ties}, and condition \eqref{eq:prefer}.
Let $C\subseteq\X\setminus\G$ be a closed communicating class of an equilibrium, let $x \in C$,
and put $m_i := \max_{w\in C}\pi_i(w)$,
so that $\ev_i(z)\le m_i<\varphi_i(x)$ for all $i \in N$ and all $z \in C$,
exactly as in the proof of Theorem~\ref{thm:arrival}.
Then $g_x$ is not absorbing,
and among the moves realised at $g_x$ there is a termination with bottom node $K\subsetneq N$,
terminator $j$ and target $y$ such that
\[
\ev_j(g_x)\;<\;\ev_j(y), \qquad \ev_j(x)\;<\;\ev_j(y)\;\le\;m_j\;<\;\varphi_j(x).
\]
In particular $\ev_j(g_x)<m_j<\varphi_j(x)$:
the merge-all move can be blocked only by pessimism about the continuation \emph{after} the merge,
never by the static payoffs the merge delivers.
\end{proposition}

\begin{proof}
If $g_x$ were absorbing then $\ev_i(g_x)=\varphi_i(x)>m_i\ge\ev_i(z)$ for all $i$ and all $z \in C$,
contradicting Lemma~\ref{lem:objection}.
So some move is realised at $g_x$; by Lemma~\ref{lem:entry} each is a termination with target outside $\G$ and,
under unilateral termination, with a single terminator.

Suppose every realised move at $g_x$ had bottom node $N$, i.e.\ target $x$.
Then $x$ is the only state other than $g_x$ receiving positive probability from $g_x$,
so with $q:=p_{g_x\to x}>0$,
\eqref{eq:bellman} gives $\big(1-\delta(1-q)\big)\big(\ev_i(g_x)-\ev_i(x)\big)=(1-\delta)\big(\varphi_i(x)-\ev_i(x)\big)$ for every $i$;
as $\ev_i(x)\le m_i<\varphi_i(x)$ and $1-\delta(1-q)>0$, this yields $\ev_i(g_x)>\ev_i(x)$ for every $i$,
contradicting profitability of a move from $g_x$ to $x$.

Hence some realised move $m$ at $g_x$ has bottom node $K\subsetneq N$;
let $j$ be its terminator and $y:=\tau(m)\ne x$, so $\ev_j(y)>\ev_j(g_x)$ by \ref{E1}.
By Lemma~\ref{lem:shadow} the move $\tilde m$ at $x$ with bottom node $K$ and terminator $j$ has target $y$.
The move at $g_x$ terminating $N$ with terminator $j$ has relevant set $\{j\}=R(m)$ and target $x$;
if $\ev_j(y)<\ev_j(x)$ it would dominate $m$, so $m$ would not be realised.
Ties being excluded by Criterion~\ref{cond:ties}, $\ev_j(y)>\ev_j(x)$, i.e.\ $\tilde m$ is profitable at $x$.
By Lemma~\ref{lem:domination} there is $m''\in\Cset(x)$ with $\ev_j(\tau(m''))\ge\ev_j(y)$;
letting $m^\ast$ maximise $\ev_j(\tau(\cdot))$ over $\Cset(x)$ gives $m^\ast\in\Fset(x)$,
hence $p_{x\to\tau(m^\ast)}>0$ by \ref{E3},
hence $\tau(m^\ast)\in C$ and $\ev_j(y)\le\ev_j(\tau(m^\ast))\le m_j$.
\end{proof}

\begin{remark}[What is open]\label{rem:open-unil}
Proposition~\ref{prop:arrival-unil} leaves exactly one scenario alive under unilateral termination:
an equilibrium may specify off-path behaviour at the unentered grand states $g_x$, $x\in C$,
that depresses $\ev_j(g_x)$ below the payoffs available inside $C$ and thereby makes merge-all dominated at every $x \in C$.
That is the mechanism identified in the proof of Theorem~\ref{thm:persist-unil}:
non-absorption of $g_x$ produces a shadow move that dominates merge-all at $x$.
Under unanimous termination the scenario cannot arise (Theorem~\ref{thm:persist-unan}),
which is why Corollary~\ref{cor:arrival-unan} is unconditional.
Deciding it under unilateral termination and the $\delta$-free hypothesis \eqref{eq:prefer} is open.
The scenario is not vacuous:
the counterexample of Proposition~\ref{prop:counterexample} realises it under unilateral termination,
with the four grand states above the cycle dissolved again as soon as they are formed.
There \eqref{eq:prefer} fails, so it does not settle the question,
but it shows that the mechanism this remark describes does occur in equilibrium.
\end{remark}

\subsection{When does the grand coalition form in one step?}\label{sec:onestep}

The simplest conceivable equilibrium sends every non-grand state straight to the grand state above it.
Write $p^\dagger$ for the process with $p^\dagger_{x\to g_x}=1$ for $x\in\X\setminus\G$ and $p^\dagger_{h\to h}=1$ for $h\in\G$,
so that $\ev_i(h)=\pi_i(h)$ for grand $h$ and
\[
   \ev_i(x) \;=\; (1-\delta)\pi_i(x)+\delta\varphi_i(x), \qquad x \in \X\setminus\G .
\]
It is not an equilibrium in general, and the reason is exactly the free-riding effect:
a group may prefer to move to some other state first,
so as to enter the grand coalition later from a better bargaining position.
What follows makes that precise,
and identifies the condition under which one-step formation is an equilibrium after all.

\begin{condition}[No share-improving move]\label{cond:noshare}
Assume Condition~\ref{cond:nash}.
For every $x \in \X\setminus\G$ and every $m \in \M(x)$ other than the merge-all move,
with target $y$ and relevant set $R$, there is an $i \in R$ with $\psi_i(y)<\psi_i(x)$.
\end{condition}

In words: no move improves the relative share of \emph{all} the players whose approval it needs.
Since $\sum_i\psi_i\equiv0$, every move improves somebody's share;
the condition is that it cannot improve every mover's share at once.

\begin{theorem}[One-step formation]\label{thm:onestep}
Assume Conditions~\ref{cond:nash} and~\ref{cond:noshare}.
\begin{enumerate}[label=(\alph*)]
\item For every $\delta\in(0,1)$ and either termination rule,
the merge-all move is profitable and undominated at every non-grand state,
and no move is profitable at any grand state; hence $p^\dagger$ satisfies \ref{E1} and \ref{E2}.
\item Put $\theta := \min\{\psi_i(x)-\psi_i(y)\}>0$,
the minimum being over the triples $(x,m,i)$ of Condition~\ref{cond:noshare},
and $R_\pi := \max_i\max_{y,z}|\pi_i(y)-\pi_i(z)|$.
For every $\delta > R_\pi/(R_\pi+\theta)$, every moving move other than merge-all is dominated by a stay,
and no stay is profitable, so $\Cset(x)=\Fset(x)=\{\text{merge-all}\}$ and $p^\dagger$ satisfies \ref{E1}--\ref{E3}.
\end{enumerate}
In case (b) the grand coalition forms in one step from every state.
\end{theorem}

\begin{proof}
(a) \emph{Merge-all is profitable and beats staying.}
Profitability is immediate, $p^\dagger$ being deterministic at $x$: $\evp_i(x)=\ev_i(g_x)$ for every $i$.
For the comparison with staying, which is what undominatedness needs, for every $i$
\[
  \ev_i(g_x)-\ev_i(x) = \varphi_i(x)-\big[(1-\delta)\pi_i(x)+\delta\varphi_i(x)\big] = (1-\delta)\big(\varphi_i(x)-\pi_i(x)\big) > 0
\]
by Lemma~\ref{lem:pareto}.

\emph{The key inequality.} Let $m\ne m_0$ be a move at $x$, with target $y$ and relevant set $R$,
and let $i \in R$ satisfy $\psi_i(y)<\psi_i(x)$.
By Lemma~\ref{lem:entry} the only move at $x$ whose target is grand is $m_0$,
so $y \in \X\setminus\G$ and $\varphi_i(y)$ is defined.
Then $\varphi_i(y)=\psi_i(y)+w_iV^\ast<\psi_i(x)+w_iV^\ast=\varphi_i(x)$,
and $\pi_i(y)=\psi_i(y)+w_iV(y)<\psi_i(x)+w_iV^\ast=\varphi_i(x)$ since $V(y)\le V^\ast$.
Hence
\begin{equation}\label{eq:onestep-key}
   \ev_i(y) = (1-\delta)\pi_i(y)+\delta\varphi_i(y) \;<\; \varphi_i(x) \;=\; \ev_i(g_x).
\end{equation}

\emph{Merge-all is undominated.} Immediate from \eqref{eq:onestep-key}:
a dominating move would need $\ev_j(\tau(m))>\ev_j(g_x)$ for \emph{all} $j$ in its relevant set.

\emph{No profitable move at a grand state.} Let $h=g_x$ and let $m$ be a move at $h$;
by Lemma~\ref{lem:entry} it is a termination, with bottom node $K$.
If $K=N$ its target is $x$, and $\ev_i(x)<\varphi_i(x)=\ev_i(h)$ for every $i$, so every stay at $h$ dominates $m$.
If $K\subsetneq N$,
Lemma~\ref{lem:shadow} gives a termination $\tilde m$ at $x$ with the same target $y$ and $R(\tilde m)\subseteq R(m)$;
applying \eqref{eq:onestep-key} to $\tilde m$ yields an $i \in R(\tilde m)\subseteq R(m)$ with $\ev_i(y)<\varphi_i(x)=\ev_i(h)$,
so again $s^i_h$ dominates $m$.
Thus $\Cset(h)$ contains no moving move;
it consists of the stays, which at an absorbing state are profitable with equality and, by the same
argument, undominated, so \ref{E3} asks only for $p_{h\to h}>0$ and \ref{E2} is satisfied at $h$;
at non-grand states \ref{E2} holds because merge-all is realised.

(b) With $i$ as above, $\varphi_i(y)-\varphi_i(x)=\psi_i(y)-\psi_i(x)\le-\theta$, so
\[
  \ev_i(y)-\ev_i(x) = (1-\delta)\big(\pi_i(y)-\pi_i(x)\big)+\delta\big(\varphi_i(y)-\varphi_i(x)\big) \le (1-\delta)R_\pi-\delta\theta,
\]
which is negative exactly when $\delta>R_\pi/(R_\pi+\theta)$.
So $s^i_x$ dominates $m$;
and since $\ev_j(g_x)>\ev_j(x)$ for every $j$ by the first step, $\evp_j(x)=\ev_j(g_x)$ exceeds
$\ev_j(x)$, so no stay at $x$ is profitable.
Hence $\Cset(x)=\{m_0\}$, and $m_0$ is then every player's favourite.
\end{proof}

\begin{remark}\label{rem:onestep-fails}
Condition~\ref{cond:noshare} is restrictive, and both of the running examples violate it.
In the payoff structure of Proposition~\ref{prop:counterexample},
the move $\emptyset\to\{AB\}$ has relevant set $\{A,B\}$ and raises the relative share of \emph{both} movers,
from $\psi_A=\psi_B=-60$ to $+60$; correspondingly $\varphi_A$ rises from $4$ to $124$,
so that move dominates merge-all at $\emptyset$,
and $p^\dagger$ is not an equilibrium there for any $\delta$.
In the three-player example of Example~\ref{ex:sym} the condition fails at $\{AB\}$ rather than at $\emptyset$:
terminating the pair returns $A$ and $B$ from $\psi_A=-0.5$ to $\psi_A=0$,
so termination improves both movers' shares.
This is the same phenomenon that \citet{HeitzigKornek2018} report as the finding that a global market emerges ``but probably not in one move''.
Under $p^\dagger$ the dominating deviation does reach a grand state,
one period later and on better terms for its movers,
so in that comparison nobody dislikes the grand coalition as such.
That reading must not be carried over to equilibria other than $p^\dagger$.
Failure of Condition~\ref{cond:noshare} says only that merge-all is dominated relative to the continuation $p^\dagger$ specifies;
in another equilibrium the same deviation may lead into a continuation that never reaches a grand state at all,
and then blocking has nothing to do with waiting for better terms.
Remark~\ref{rem:why-not-arrival} works this out in the counterexample,
where the blockers prefer perpetual cycling to the grand state in front of them.
\end{remark}

\section{Arrival: when it fails}\label{sec:cycles}

By Corollary~\ref{cor:dichotomy} the only way the grand coalition can fail to form is a closed communicating class $C\subseteq\X\setminus\G$.
We record what such a class must satisfy, exclude the two shortest candidates, and then construct one.

\begin{proposition}[Necessary conditions]\label{prop:cycle}
Let $p$ be an equilibrium, $\delta\in(0,1)$,
and $C\subseteq\X\setminus\G$ a closed communicating class with stationary distribution $\mu$;
let $\bar a_i,\bar V$ be as in Lemma~\ref{lem:occupation}, and, where Condition~\ref{cond:nash} is invoked,
$\bar\psi_i:=\sum_{x\in C}\mu(x)\psi_i(x)$.
Then
\begin{enumerate}[label=(\roman*)]
\item $\bar V\le V^\ast-\gamma$, and under Condition~\ref{cond:nash}, $\bar a_i=\bar\psi_i+w_i\bar V$ for every $i$;
\item under Condition~\ref{cond:nash}, for every $i$,
\[
\sum_{x\in C}\mu(x)\big(\varphi_i(x)-\ev_i(x)\big)=w_i\big(V^\ast-\bar V\big)\ \ge\ w_i\gamma\ >\ 0,
\]
i.e.\ \emph{every} player strictly prefers immediate formation of the grand coalition to the continuation,
on $\mu$-average over the class, by exactly her share of the average efficiency loss;
\item for every $x\in C$ there are $i \in N$ and $z\in C$ with $\ev_i(z)\ge\ev_i(g_x)$;
if moreover every grand state is absorbing,
then for every $x \in C$ there is an $i$ with $\pi_i(g_x)\le\max_{w\in C}\pi_i(w)$,
so that \eqref{eq:prefer} must fail,
and under Condition~\ref{cond:nash} there is for every $x \in C$ an $i$ with $\max_{w\in C}\psi_i(w)-\psi_i(x)\ge w_i\gamma$,
so that $\max_i S_i\ge w_{\min}\gamma$.
\end{enumerate}
\end{proposition}

\begin{proof}
(i) is Assumption~\ref{ass:eff} and $\pi_i=\psi_i+w_iV$ averaged against $\mu$.

(ii) By Lemma~\ref{lem:psi}, $\sum_x\mu(x)\varphi_i(x)=\bar\psi_i+w_iV^\ast$; by Lemma~\ref{lem:occupation},
$\sum_x\mu(x)\ev_i(x)=\bar a_i=\bar\psi_i+w_i\bar V$.
Subtract and use (i).

(iii) The first claim is Lemma~\ref{lem:objection}.
For the second, absorption of $g_x$ gives $\ev_i(g_x)=\pi_i(g_x)$,
while $\ev_i(z)\le\max_{w\in C}\pi_i(w)$ as in the proof of Theorem~\ref{thm:arrival},
so $\pi_i(g_x)\le\max_{w\in C}\pi_i(w)\le\max_{y\in\X\setminus\G}\pi_i(y)$, contradicting \eqref{eq:prefer}.
For the third,
under Condition~\ref{cond:nash} the same chain reads $\psi_i(x)+w_iV^\ast\le\max_{w\in C}\psi_i(w)+w_i(V^\ast-\gamma)$.
\end{proof}

Part (ii) is the exact form of the intuition that cycling is collectively wasteful;
it shows the obstruction is purely distributional, since no player can object \emph{on average}:
objections must be state-specific, with the roles rotating.
Part (iii) is the converse of Theorem~\ref{thm:arrival}(b).

\subsection{The two shortest candidate cycles}

By Lemma~\ref{lem:reversal} a class cannot traverse a pair of states in both directions.
The two shortest cycles respecting this are the following.
Recall that $\pi$ is \emph{symmetric} if $\pi_{\rho(i)}(\rho(x))=\pi_i(x)$ for every permutation $\rho$ of $N$,
where $\rho$ acts on states elementwise.

\begin{proposition}[Rotating pairs]\label{prop:rotating}
Let $n=4$, $N=\{A,B,C,D\}$, and consider
\[
  x_1:=\emptyset,\quad x_2:=\{\{A,B\}\},\quad x_3:=\{\{A,B\},\{C,D\}\},\quad x_4:=\{\{C,D\}\}.
\]
Suppose an equilibrium has $C=\{x_1,x_2,x_3,x_4\}$ as a closed communicating class,
traversed deterministically in the order $x_1\to x_2\to x_3\to x_4\to x_1$.
Put
\[
  D_i \;:=\; \pi_i(x_2)+\pi_i(x_4)-\pi_i(x_1)-\pi_i(x_3) \qquad (i \in N),
\]
the extent to which the two mergers are \emph{substitutes} for $i$ in static payoffs.
Then, for every $\delta \in (0,1)$,
\[
   D_A>0,\quad D_B>0,\quad D_C<0,\quad D_D<0 .
\]
In particular the cycle is impossible whenever $D_i$ has the same sign for all four players,
and impossible for symmetric $\pi$.
\end{proposition}

\begin{proof}
The four moves are: merge $\{A,B\}$ (relevant set $\{A,B\}$); merge $\{C,D\}$ (relevant set $\{C,D\}$);
terminate $\{A,B\}$ (relevant set a non-empty subset of $\{A,B\}$ under either termination rule,
since $\{A,B\}$ is top-level in $x_3$); terminate $\{C,D\}$ (likewise inside $\{C,D\}$).
So at steps $1$ and $3$ at least one of $A,B$ is relevant, and at steps $2$ and $4$ at least one of $C,D$ is;
let $i$ be relevant at steps $1$ and $3$.
By Lemma~\ref{lem:cycle} with $k=4$, writing $\Sigma_\delta := 1+\delta+\delta^2$,
\begin{align*}
 \text{step } 1:&\quad \pi_i(x_2)+\delta\pi_i(x_3)+\delta^{2}\pi_i(x_4) \;>\; \Sigma_\delta\,\pi_i(x_1),\\
 \text{step } 3:&\quad \pi_i(x_4)+\delta\pi_i(x_1)+\delta^{2}\pi_i(x_2) \;>\; \Sigma_\delta\,\pi_i(x_3).
\end{align*}
Adding the two and collecting terms,
the coefficient of $\pi_i(x_2)$ and of $\pi_i(x_4)$ on the left is $1+\delta^2$,
that of $\pi_i(x_1)$ and of $\pi_i(x_3)$ on the left is $\delta$, and $\Sigma_\delta-\delta = 1+\delta^2$, so
\[
   (1+\delta^{2})\big(\pi_i(x_2)+\pi_i(x_4)\big) \;>\; (1+\delta^{2})\big(\pi_i(x_1)+\pi_i(x_3)\big),
\]
i.e.\ $D_i>0$; the discount factor has cancelled.
The same computation at steps $2$ and $4$ gives $D_j<0$ for a player $j$ relevant there.
At each of the four states the two members of the pair involved play the same role---both singletons at $x_1$,
both members of the pair at $x_2$ and $x_3$, both outsiders at $x_4$ for $A,B$,
and symmetrically for $C,D$---but that is not needed:
it suffices that the relevant sets at steps $1,3$ lie in $\{A,B\}$ and those at steps $2,4$ in $\{C,D\}$,
so that $D_i>0$ for some $i \in \{A,B\}$ and $D_j<0$ for some $j \in \{C,D\}$,
and the merge moves require all of $\{A,B\}$ respectively $\{C,D\}$,
giving the four strict inequalities as stated.
Finally,
if $\pi$ is symmetric then the permutation $\rho := (AC)(BD)$ fixes $x_1$ and $x_3$ and exchanges $x_2$ and $x_4$,
so $\pi_C(x_2)=\pi_A(x_4)$ and $\pi_C(x_4)=\pi_A(x_2)$ and $\pi_C(x_k)=\pi_A(x_k)$ for $k \in \{1,3\}$,
whence $D_C=D_A$, contradicting $D_A>0>D_C$.
\end{proof}

\begin{remark}\label{rem:separable}
$D_i$ measures a genuine interaction between the two disjoint mergers,
and it vanishes identically whenever payoffs are additively separable across disjoint top-level coalitions,
i.e.\ whenever the effect on $i$ of $\{A,B\}$ merging does not depend on whether $\{C,D\}$ has merged.
The public-good model of Example~\ref{ex:negative} with \emph{linear} benefits is of this kind:
each block's optimal contributions depend only on that block, total provision is additive across blocks,
and one checks that $D_i=0$ for every $i$ and every parameter choice.
So in that model the cycle is excluded outright, for every $\delta$, both termination rules,
and all asymmetries.
The previous revision proved this only for symmetric payoffs.
\end{remark}

\begin{proposition}[Build--merge--strip]\label{prop:bms}
For $n=3$ the cycle $\emptyset\to\{\{A,B\}\}\to\{\{A,B\},N\}\to\{N\}\to\emptyset$ does not exist,
because its third step---deleting $\{A,B\}$ while retaining $N$---is not a move (Definition~\ref{def:moves}).
\end{proposition}

\begin{remark}
It is instructive to see what would exclude that cycle if internal terminations \emph{were} allowed,
since it isolates the role of the chain rule.
By Lemma~\ref{lem:psi}, $\psi$ takes only two values on that cycle,
$\alpha:=\psi(\emptyset)=\psi(\{N\})$ and $\beta:=\psi(\{\{A,B\}\})=\psi(\{\{A,B\},N\})$.
Writing $V_0<V_1<V^\ast$ for the three pie sizes and applying Lemma~\ref{lem:cycle} exactly as in Proposition~\ref{prop:rotating},
the condition for the player who dissolves $\{N\}$ and the conditions for the unanimous merge into $\{\{A,B\},N\}$ are incompatible unless $V_1<V_0$,
which contradicts monotonicity of $V$ under coarsening of the coalition structure.
So the cycle is doubly excluded.
\end{remark}

\subsection{What a counterexample must look like}

Combining Theorem~\ref{thm:n3},
Proposition~\ref{prop:cycle} and Propositions~\ref{prop:rotating}--\ref{prop:bms},
a closed communicating class $C\subseteq\X\setminus\G$ requires all of:
\begin{itemize}
\item $n\ge4$;
\item for the rotating-pairs cycle,
a sign reversal of the interaction term $D_i$ between the two pairs (Proposition~\ref{prop:rotating}),
which in particular rules out symmetric payoffs and any payoff structure that is additively separable across disjoint coalitions;
\item failure of \eqref{eq:prefer}, i.e.\ some player must prefer some non-grand state to some grand state;
under Condition~\ref{cond:nash} this means stakes exceeding the efficiency gain,
$\max_i S_i\ge w_{\min}\gamma$ (Proposition~\ref{prop:cycle}(iii));
\item at every $x\in C$ an objector in the sense of Lemma~\ref{lem:objection}, whose identity must
rotate, since by Proposition~\ref{prop:cycle}(ii) no player objects on $\mu$-average.
\end{itemize}
For a fixed candidate cycle and mover pattern these are finitely many linear inequalities in the payoffs (Lemma~\ref{lem:cycle}) together with the blocking conditions,
so the question is decidable by enumeration for small $n$;
Section~\ref{sec:counterexample} carries that out and finds that the conditions are satisfiable.

\subsection{A counterexample}\label{sec:counterexample}

Over the whole class of payoff structures allowed by Assumptions~\ref{ass:eff} and~\ref{ass:share},
the question can be decided rather than sampled.
Fix a candidate cycle and suppose it is traversed deterministically.
Then, by Lemma~\ref{lem:cycle}, the evaluations along it are explicit linear functions of the static payoffs;
by Theorem~\ref{thm:persist-unan} the evaluations at grand states are the static payoffs there;
and the requirements that each move be profitable and undominated and that every other move be blocked are linear inequalities in the static payoffs,
up to a disjunction for each blocked move.
Since the static payoffs are themselves linear in the free parameters---the raw payoffs $r_i(P)$ of each player in each of the $15$ coalition structures,
from which Assumption~\ref{ass:share} determines $\pi$---the whole question is a mixed-integer linear feasibility problem.
Solving it for the cycles of length at most five gives a counterexample.

\begin{theorem}[No deterministic cycles]\label{thm:nodetcycle}
Under either termination rule, if $p$ satisfies \ref{E1} and \ref{E3} then no closed communicating
class of $p$ is a deterministic cycle $x_1\to x_2\to\cdots\to x_r\to x_1$ of $r\ge2$ distinct states.
\end{theorem}

\begin{proof}
Suppose it were, so $p_{x_k\to x_{k+1}}=1$ for each $k$, indices modulo $r$, and
$\evp_j(x_k)=\ev_j(x_{k+1})$ for every $j$.
Fix a player $j$.
The increments $\ev_j(x_{k+1})-\ev_j(x_k)$ sum to zero around the cycle, so some $k$ has
\begin{equation}\label{eq:cycledrop}
\ev_j(x_{k+1})\;\le\;\ev_j(x_k);
\end{equation}
take for instance $k$ maximising $\ev_j$ along the cycle.
By \ref{E1} the realised move lies in $\Cset(x_k)$, so that set is non-empty;
let $m^\ast$ be $\succeq^{x_k}_j$-maximal in it.
Then $m^\ast\in\Fset(x_k)$, so \ref{E3} gives $p_{x_k\to\tau(m^\ast)}>0$ and hence
$\tau(m^\ast)=x_{k+1}$.

Consider the stay $s:=s^j_{x_k}$.
It is profitable, since $\ev_j(x_k)\ge\ev_j(x_{k+1})=\evp_j(x_k)$ by \eqref{eq:cycledrop}.
If $s$ is undominated then $s\in\Cset(x_k)$, so maximality gives $x_{k+1}=\tau(m^\ast)\succeq^{x_k}_j x_k$;
as $x_{k+1}\ne x_k$ and ties are broken towards $x_k$, this forces $\ev_j(x_{k+1})>\ev_j(x_k)$,
contradicting \eqref{eq:cycledrop}.
If instead $s$ is dominated, the dominator has a non-empty relevant set contained in $\{j\}$, hence
equal to $\{j\}$, and its target $z$ satisfies $\ev_j(z)>\ev_j(x_k)$.
It is therefore profitable, $\ev_j(z)>\ev_j(x_k)\ge\evp_j(x_k)$, so Lemma~\ref{lem:domination}
supplies $m''\in\Cset(x_k)$ with $R(m'')\subseteq\{j\}$ and
$\ev_j(\tau(m''))\ge\ev_j(z)>\ev_j(x_k)\ge\ev_j(x_{k+1})$.
Maximality of $m^\ast$ then gives $\ev_j(x_{k+1})\ge\ev_j(\tau(m''))$, a contradiction.
\end{proof}

The argument uses no property of the payoffs and no property of the termination rule.
Under the unanimous rule the second case cannot arise, since every moving move has at least two
relevant players and a stay is never dominated;
under the unilateral rule it can, and is disposed of by the domination chain.
What the theorem does not exclude is a cycle traversed at random, where $\evp_j(x_k)$ is an average
over several successors and may exceed $\ev_j(x_k)$ at every state of the class.
So a cycle cannot run on a fixed schedule: somebody always reaches a state they would rather not
leave, and \ref{E3} gives them the floor.

\begin{remark}[Status of Proposition~\ref{prop:counterexample} under Definition~\ref{def:rationality}]\label{rem:cyclestatus}
The construction below predates Definition~\ref{def:rationality} and does not survive it, under
either termination rule.
Its cycle is deterministic, so Theorem~\ref{thm:nodetcycle} applies directly;
computing the evaluations of the displayed cycle confirms the mechanism, the stays of $C$ and $D$
at the disagreement state and of $A$ and $B$ at $A,B,[CD]$ being profitable, undominated and their
owners' favourites, so that \ref{E3} forces a self-loop the cycle does not have.
Whether arrival can fail at all is therefore open again, under both rules:
it would now require a cycle traversed at random, and the mixed-integer program of
Subsection~\ref{sec:numerics} would have to search over processes rather than over payoffs with the
process fixed.
Subsection~\ref{sec:logit} reports the corresponding numerical picture: with noisy play the cycle is
traversed at intermediate temperatures and disappears as the noise is removed.
The rest of this subsection is left as it stands, as a record of what the earlier definition allowed.
\end{remark}

\begin{proposition}[The grand coalition need not form, under the earlier notion of profitability]\label{prop:counterexample}
\emph{This proposition, and everything in the rest of this subsection, is stated for the notion of
profitability used before Definition~\ref{def:rationality}---moves compared with $\ev_i(x)$ rather
than with $\evp_i(x)$, and no stay moves.
By Theorem~\ref{thm:nodetcycle} the process below is not an equilibrium in the sense of
Definition~\ref{def:eq}, under either termination rule; see Remark~\ref{rem:cyclestatus}.}

Let $n=4$ with players $A,B,C,D$, and let Condition~\ref{cond:nash} hold with equal weights $w_i=1/4$.
Let the partition function be $v(K;P)=0$ for every coalition structure $P$ and every block $K \in P$, except
\begin{align*}
  v\big(ABCD;\ \{ABCD\}\big) &= 256, & &\\
  v\big(AB;\ AB|C|D\big) &= 240, & v\big(CD;\ AB|CD\big) &= 240,\\
  v\big(C;\ A|B|C|D\big) &= 112, & v\big(D;\ A|B|C|D\big) &= 128 .
\end{align*}
Only the joint payoff of each block matters,
since Assumption~\ref{ass:share} determines how a block divides it.
Then Assumption~\ref{ass:eff} holds with $V^\ast = 256$ and $\gamma = 16$,
and under \emph{either} termination rule, for $\delta = 1/2$,
there is an equilibrium whose unique closed communicating class is
\[
  C \;=\; \big\{\ \emptyset,\ \ \{AB\},\ \ \{AB,CD\},\ \ \{CD\}\ \big\} \;\subseteq\; \X\setminus\G,
\]
traversed deterministically in that order.
Started at the fully non-cooperative state, the process never forms the grand coalition.
\end{proposition}

Whether the failure of arrival is a property of this payoff structure or only of one equilibrium at
it, we can report evidence but not a proof.
Solving for equilibria from $40$ random initial belief systems, at $\delta=1/2$ and under each
termination rule, every exact equilibrium found contained the cycle and none reached a grand state;
and $p^\dagger$ is not an equilibrium here, since the move $\emptyset\to\{AB\}$ dominates the
merge-all at $\emptyset$.
We have not been able to settle the question, and the asymmetry is instructive: exhibiting one
cycling equilibrium requires constructing a single self-confirming belief system, whereas ruling out
an arriving one requires excluding all of them.
The verification is a finite computation, and we report it in full.
The static payoffs and the evaluations on the cycle, and the static payoffs at the four grand states above it,
are
\[
\begin{array}{lcccc}
\text{state } x & \pi(x) & \ev(x) & g_x & \pi(g_x)\\[2pt]
x_1=A,B,C,D & (0,0,112,128) & (32,32,75.2,84.8) & [ABCD] & (4,4,116,132)\\
x_2=[AB],C,D & (120,120,0,0) & (64,64,38.4,41.6) & [[AB]CD] & (124,124,4,4)\\
x_3=[AB],[CD] & (0,0,120,120) & (8,8,76.8,83.2) & [[AB][CD]] & (4,4,124,124)\\
x_4=A,B,[CD] & (0,0,-8,8) & (16,16,33.6,46.4) & [AB[CD]] & (64,64,56,72)
\end{array}
\]
where the evaluations follow from Lemma~\ref{lem:cycle} with $k=4$ and $\delta=1/2$,
so that $\ev_i(x_j) = \tfrac{8}{15}\sum_{s=0}^{3}2^{-s}\pi_i(x_{j+s})$,
and those at the grand states equal the static payoffs there by Theorem~\ref{thm:persist-unan}.
Two entries are worth checking by hand.
At $x_2$ the agreement $\{A,B\}$ is top-level with reference state $x_1$, where $A$ and $B$ receive $0$ each,
so its surplus is $240-0=240$ and Assumption~\ref{ass:share} gives them $120$ each.
At $x_4$ the agreement $\{C,D\}$ has reference state $x_1$, where $C$ and $D$ receive $112$ and $128$,
and its surplus is $0-(112+128)=-240$, so they receive $112-120=-8$ and $128-120=8$:
coordinating while $A$ and $B$ are not coordinated destroys their rent.

Each of the four moves is profitable for the players it needs, and in three cases only far-sightedly:
\begin{itemize}
\item $x_1 \to x_2$, $A$ and $B$ sign: $\ev_A$ rises from $32$ to $64$;
\item $x_2 \to x_3$, $C$ and $D$ sign: $\ev_C$ rises from $38.4$ to $76.8$;
\item $x_3 \to x_4$, $A$ and $B$ terminate their own agreement: $\ev_A$ rises from $8$ to $16$,
although $\pi_A$ is $0$ in both states,
because from $x_4$ the process returns to $x_1$ and they can sign again;
\item $x_4 \to x_1$, $C$ and $D$ terminate: $\ev_C$ rises from $33.6$ to $75.2$.
\end{itemize}
The merge-all move is blocked at each of the four states,
and all three mechanisms of Lemma~\ref{lem:objection} occur: at $x_1$ it is unprofitable for $A$,
since $\ev_A(g_{x_1}) = 4 < 32 = \ev_A(x_1)$, and likewise for $B$;
at $x_2$ it is unprofitable for $C$ and $D$ ($4 < 38.4$); at $x_3$ for $A$ and $B$ ($4 < 8$);
and at $x_4$ \emph{no} player objects,
but the merge-all move is dominated by the termination that is realised there,
whose relevant set $\{C,D\}$ is contained in $N$ and whose target $x_1$ both $C$ and $D$ strictly prefer to $g_{x_4}$:
$\ev_C(x_1) = 75.2 > 56 = \pi_C(g_{x_4})$ and $\ev_D(x_1)=84.8>72$.
Every other move at every state of $C$ is unprofitable,
and the transition matrix produced by the rule of Definition~\ref{def:eq} from these evaluations reproduces the cycle exactly,
so this is an equilibrium and not an approximation to one.
The same payoffs give the same closed class for $\delta \in \{0.3,0.4,0.45,0.5,0.55,0.6,0.7,0.75\}$;
for $\delta = 0.8$ and above the same payoffs yield a different equilibrium,
in which grand states are reached.

\emph{Unilateral termination.} The same partition function works there too, for the same discount factors,
and the cycle and the evaluations along it are unchanged,
since $C$ is closed and the same transitions are realised.
What changes is everything off the cycle, in an instructive way.
Each termination is now available to each single member separately,
so at $x_3$ the moves ``$A$ terminates $\{A,B\}$'' and ``$B$ terminates $\{A,B\}$'' are both realised,
with the same target, and similarly at $x_4$.
More importantly, the four grand states above the cycle are no longer absorbing:
Theorem~\ref{thm:persist-unil} does not reach them, since they are never entered,
and in this equilibrium each is dissolved again at once,
\[
\begin{array}{lccl}
\text{grand state} & \pi & \ev & \text{dissolved by, and to}\\[2pt]
[ABCD] & (4,4,116,132) & (18,18,95.6,108.4) & A \text{ or } B,\ \to x_1\\{}
[[AB]CD] & (124,124,4,4) & (94,94,21.2,22.8) & C \text{ or } D,\ \to x_2\\{}
[[AB][CD]] & (4,4,124,124) & (10,10,78.8,85.2) & A \text{ or } B,\ \to x_4\\{}
[AB[CD]] & (64,64,56,72) & (48,48,65.6,78.4) & C \text{ or } D,\ \to x_1
\end{array}
\]
each evaluation being $(1-\delta)\pi+\delta\,\ev(\text{target})$ by Lemma~\ref{lem:onestep}.
In the last two rows the realised move is not the dissolution of the overarching agreement but a \emph{deeper} termination---of $\{A,B\}$ and of $\{C,D\}$ respectively---which by the chain rule destroys the overarching agreement as well.
This is exactly the configuration that Proposition~\ref{prop:arrival-unil} isolates as the one it cannot exclude,
and that Remark~\ref{rem:open-unil} describes as the residual scenario.
The blocking is also sharper than under unanimity.
At $x_1$ and $x_2$ the merge is still unprofitable, for $A$ and $B$ respectively for $C$ and $D$.
At $x_3$ and $x_4$, however,
it is profitable for \emph{every} player---at $x_3$ the grand state is worth $(10,10,78.8,85.2)$ against $(8,8,76.8,83.2)$ in the cycle,
and at $x_4$ it is worth $(48,48,65.6,78.4)$ against $(16,16,33.6,46.4)$---and it is still not taken,
because a single player prefers to terminate instead:
$\ev_A(x_4)=16>10=\ev_A(g_{x_3})$ and $\ev_C(x_1)=75.2>65.6=\ev_C(g_{x_4})$, so the termination,
whose relevant set is a singleton contained in $N$, dominates the merge.
A Pareto-improving move that every player wants is thus blocked by one player's better outside option,
which is precisely why Definition~\ref{def:rationality}(b) is part of the model.

\emph{Far-sighted sharing.} Still under the earlier notion of profitability, it is natural to ask whether the far-sighted sharing rule of Section~\ref{sec:pivot} repairs the example,
since under that rule the merge-all move is profitable in evaluations whenever it raises total long-term welfare (Proposition~\ref{prop:farsighted}).
It does not.
Keeping the same partition function,
solving for the transfers that make the rule hold and for the process simultaneously,
one obtains for unilateral termination and $\delta=2/5$ an exact equilibrium with the same closed class,
\[
\begin{array}{lccc}
\text{state} & \pi & \ev & \ev(g_x)-\ev(x)\\[2pt]
x_1 & (0,0,112,128) & (29.56,\,29.56,\,80.91,\,90.51) & (6.36,6.36,6.36,6.36)\\
x_2 & (120,120,0,0) & (73.89,\,73.89,\,34.29,\,34.29) & (9.91,9.91,9.91,9.91)\\
x_3 & (0,0,123.2,116.8) & (4.73,\,4.73,\,85.71,\,85.71) & (2.40,2.40,2.40,2.40)\\
x_4 & (0,0,-4.8,4.8) & (11.82,\,11.82,\,29.49,\,39.09) & (38.40,38.40,38.40,38.40)
\end{array}
\]
The last column is the point.
By construction the four players share the long-term gain from merging equally,
and that gain is strictly positive at every state of the cycle,
so the merge-all move is Pareto-improving \emph{in evaluations}---the strongest form the Pareto property can take---at every state.
It is nevertheless never realised,
because at each state the pair that moves strictly prefers the cycle successor to the grand state,
and their relevant set is contained in $N$, so the merge is dominated.
The far-sighted rule therefore removes the gap between static and long-term payoffs that hypothesis \eqref{eq:prefer} exists to bridge,
and arrival still fails.
What defeats it is domination alone.

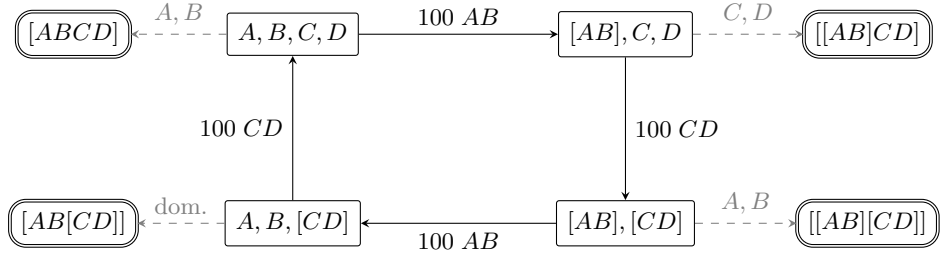
\begin{figure}[htbp]
\centering
\begin{tikzpicture}[>=stealth,
   st/.style={draw,rectangle,rounded corners=1pt,inner sep=4pt,font=\small},
   gr/.style={draw,rectangle,rounded corners=6pt,double,inner sep=4pt,font=\small}]
\node[st] (a) at (0,1.5) {$A,B,C,D$};
\node[st] (b) at (4.4,1.5) {$[AB],C,D$};
\node[st] (c) at (4.4,-1.0) {$[AB],[CD]$};
\node[st] (d) at (0,-1.0) {$A,B,[CD]$};
\node[gr] (ga) at (-2.9,1.5) {$[ABCD]$};
\node[gr] (gb) at (7.6,1.5) {$[[AB]CD]$};
\node[gr] (gc) at (7.6,-1.0) {$[[AB][CD]]$};
\node[gr] (gd) at (-2.9,-1.0) {$[AB[CD]]$};
\draw[->] (a) -- node[above,font=\scriptsize]{$100\;AB$} (b);
\draw[->] (b) -- node[right,font=\scriptsize]{$100\;CD$} (c);
\draw[->] (c) -- node[below,font=\scriptsize]{$100\;AB$} (d);
\draw[->] (d) -- node[left,font=\scriptsize]{$100\;CD$} (a);
\draw[->,dashed,gray] (a) -- node[above,font=\scriptsize]{$A,B$} (ga);
\draw[->,dashed,gray] (b) -- node[above,font=\scriptsize]{$C,D$} (gb);
\draw[->,dashed,gray] (c) -- node[above,font=\scriptsize]{$A,B$} (gc);
\draw[->,dashed,gray] (d) -- node[above,font=\scriptsize]{dom.} (gd);
\end{tikzpicture}
\caption{The equilibrium of Proposition~\ref{prop:counterexample} at $\delta=1/2$.
Solid arrows carry probability one; their labels give the players whose agreement the move needs.
Dashed arrows are the merge-all moves, which are \emph{not} realised: at the first three states some
player, named on the arrow, is made worse off by merging; at the fourth no player objects, but the
merge is dominated by the termination that is realised. Doubled boxes are grand states, absorbing by
Theorem~\ref{thm:persist-unan}.}\label{fig:cycle}
\end{figure}

The example is consistent with everything proved above, and tight against two of the results.
Condition \eqref{eq:prefer} fails, as Theorem~\ref{thm:arrival} requires:
$\pi_A$ is $120$ at $x_2$ but only $4$ at the grand state $[ABCD]$.
The sign condition of Proposition~\ref{prop:rotating} holds with room to spare: $D_A=D_B=120>0>-240=D_C=D_D$,
so the two pairs do indeed disagree about whether the two agreements are substitutes.
And Proposition~\ref{prop:cycle}(ii) holds as an identity: with $\bar V = (240+240+240+0)/4 = 180$,
each player's $\mu$-average of $\varphi_i - \ev_i$ is $\tfrac14(256-180) = 19$,
so every player would prefer immediate formation of the grand coalition \emph{on average over the cycle},
and yet at each single state somebody blocks it.
That is the whole difficulty of the arrival problem in one example.

The economics is a stylised rent-shifting story rather than a market model, and deliberately not a
climate one: as Example~\ref{ex:running} warned, neither calibrated specification of
Section~\ref{sec:numerics} produces this behaviour, and the reader should not carry the pathology
back to the running example.
Think of $\{A,B\}$ and $\{C,D\}$ as two sides of a market, with the rent accruing to $C$ and $D$ by default.
$A$ and $B$ can seize it by coordinating, but only while $C$ and $D$ have not;
$C$ and $D$ can claw it back by coordinating in turn; and a coalition of $C$ and $D$ alone,
facing uncoordinated opponents, is worthless and indeed destroys the rent.
The grand coalition is efficient,
but its surplus is divided by the sharing rule against whatever reference state obtains when it is signed,
so $A$ and $B$ would enter it from a position of no bargaining power and refuse.
We do not claim that this structure arises from a standard market model;
the two specifications of Section~\ref{sec:numerics} do not produce it.

\begin{remark}[Why blocking does not imply eventual arrival]\label{rem:why-not-arrival}
One might argue as follows.
A group blocks the merge only because it expects to enter a grand state later on better terms;
so a grand state does form later, and since grand states are absorbing,
the process reaches one with positive probability.
The argument has a gap, and the example locates it.

What a blocking player compares is $\varphi_i(x)$ with her continuation value $\ev_i(x)$,
and that continuation need not contain a grand state at all.
At $x_1$ player $A$ blocks not because she awaits a better grand state but because she prefers the cycle itself:
the rotation pays her $\ev_A=32$, against $\varphi_A(x_1)=4$ for ever.
Her outside option is the merry-go-round, not a later merger.

It is nevertheless true, and worth stating, that every player prefers \emph{some} grand state to the cycle.
Here $A$ and $B$ would each take $[[AB]CD]$, worth $124$ to them against at most $64$ anywhere in the cycle,
and $C$ and $D$ would take $[D[ABC]]$ and $[C[ABD]]$,
worth $138.7$ and $149.3$ against at most $76.8$ and $84.8$.
What stops this from becoming an argument for arrival is Lemma~\ref{lem:entry}:
the only move into a grand state $h$ is the merge-all at its parent $t(h)$,
and it needs the approval of all of $N$.
Each grand state has exactly one gateway,
and at that gateway the consenting set contains precisely the players whom that grand state disfavours.
$A$ and $B$ can enter $[[AB]CD]$ only from $x_2$,
and there $C$ and $D$ would fall to $4$ from an evaluation near $40$, so they veto.

Who wants the merge at each state of the cycle, and who blocks it:
\[
\begin{array}{lcccc}
\text{state} & g_x & \pi(g_x) & \ev(x) & \text{blocked by}\\[2pt]
x_1=A,B,C,D & [ABCD] & (4,4,116,132) & (32,32,75.2,84.8) & A,B\\
x_2=[AB],C,D & [[AB]CD] & (124,124,4,4) & (64,64,38.4,41.6) & C,D\\
x_3=[AB],[CD] & [[AB][CD]] & (4,4,124,124) & (8,8,76.8,83.2) & A,B\\
x_4=A,B,[CD] & [AB[CD]] & (64,64,56,72) & (16,16,33.6,46.4) & \text{nobody}
\end{array}
\]
At the first three states one pair gains from merging and the other vetoes,
the roles alternating exactly as Proposition~\ref{prop:cycle}(ii) requires of any cycle.
At $x_4$ nobody vetoes: the merge is profitable for all four players.
It fails there for the second reason, domination---$C$ and $D$ prefer to terminate and return to $x_1$,
worth $75.2$ and $84.8$ to them against $56$ and $72$ in $g_{x_4}$.

Finally, no \emph{single} grand state beats the cycle for all four players at once.
Total payoff is $V^\ast=256$ at every grand state against an average of $180$ along the cycle,
so the cycle is inefficient;
but Lemma~\ref{lem:psi} allows that surplus to be divided only in the ways indexed by the parent states,
since $\psi$ is inherited from the parent, and each available division is rejected at its own gateway.
\end{remark}

\begin{remark}[What is still open]\label{rem:still-open}
The counterexample lives at moderate discount factors.
We have not found one for $\delta$ close to $1$,
and the feasibility problem above is a decision procedure only once a candidate cycle and a deterministic traversal are fixed,
so the search over $\delta \to 1$, over longer cycles,
over non-deterministic classes and over unilateral termination is not exhaustive.
Whether arrival can fail as $\delta \to 1$ therefore remains open, and it is the case that matters most,
for the reasons given in Section~\ref{sec:delta}.
\end{remark}

\subsection{How special is the counterexample?}\label{sec:numerics}

Proposition~\ref{prop:counterexample} shows that arrival can fail;
it says nothing about how contrived the payoff structure has to be.
Before finding it we had searched for a cycling equilibrium numerically,
in two economic specifications with $n=4$ players, and found none.
That search is worth reporting, because it is evidence in the other direction: the counterexample is
not what one meets by writing down a plausible model.
In both, the process,
the moves and the sharing rule are exactly those of Section~\ref{sec:payoffs} and Definition~\ref{def:moves};
only the map from a coalition structure to payoffs differs.
Equilibria were computed with the probability rule of \citet{HeitzigKornek2018},
under which a move's probability is the total weight of the players favouring it,
and which satisfies \ref{E1}--\ref{E3};
the fixed point between beliefs and evaluations was found by damped iteration and then verified exactly,
by checking that the rule applied to the resulting evaluations returns the same process.
The state space has $52$ states, $26$ of them grand.

\emph{The public-good specification.} Player $i$ contributes $q_i \ge 0$ at cost $c_iq_i^2/2$ and benefits linearly from total provision:
$\pi_i = b_i Q - c_i q_i^2/2$ with $Q=\sum_j q_j$.
A top-level coalition $K$ maximises $\sum_{i \in K}\pi_i$, giving $q_i = B_K/c_i$ for $i \in K$,
where $B_K := \sum_{j\in K}b_j$.
Total payoff is maximised by the grand partition, so Assumption~\ref{ass:eff} holds,
and every coalition's surplus is positive, so Condition~\ref{cond:gains} holds as well.
This is the specification of Example~\ref{ex:negative} with heterogeneous parameters and without the quadratic benefit term.

\emph{The Cournot specification.} Firm $i$ produces $q_i$ at cost $c_iq_i^2/2$ and faces inverse demand $A-Q$;
a top-level coalition maximises its members' joint profit, taking the other coalitions' outputs as given.
Writing $\gamma_K := \sum_{i\in K}1/c_i$ and $\Gamma := \sum_{K \in P}\gamma_K/(1+\gamma_K)$,
the unique equilibrium of the game between coalitions has
\[
   A-Q = \frac{A}{1+\Gamma}, \qquad q_i = \frac{A}{(1+\gamma_K)(1+\Gamma)\,c_i} \quad (i \in K).
\]
Total payoff is industry profit, which the grand coalition maximises, so Assumption~\ref{ass:eff} holds;
mergers here have genuine cross-effects,
since a merger of one pair raises the price and thereby the profit of the other pair,
so $D_i \ne 0$ in the sense of Proposition~\ref{prop:rotating}.
Condition~\ref{cond:gains} may fail, which does not matter outside Theorem~\ref{thm:smalldelta}.

\emph{The shape of a cycle, and the worst case.}
The same mixed-integer formulation answers two further questions.

\begin{proposition}[Only one shape of cycle, $n=4$]\label{prop:shapes}
Let $n=4$ and let termination be unanimous.
Among the simple cycles of length at most five in the move graph restricted to $\X\setminus\G$, the
twelve of length three and the twelve of length five cannot be closed communicating classes of an
equilibrium for any payoff structure, and the only ones that can are the six rotating-pairs cycles of
length four---the three ways of pairing the players, each in two orientations, all equivalent under
relabelling.
\end{proposition}

The proposition is established by deciding the feasibility of the system of
Section~\ref{sec:counterexample} for each of the thirty candidate cycles; infeasibility there is a
proof, since the relaxation leaves evaluations off the cycle free.
Maximising $V^\ast-\bar V$ over the payoff structures sustaining such a cycle, with $\gamma$ pinned to
$1$ by requiring some non-grand coalition structure to attain it, gives the unboundedness reported in
Proposition~\ref{prop:price-unbounded}: the optimum sits on the boundary of the payoff box and scales
linearly with it, so no bound in units of $\gamma$ exists.

\emph{Result of the search.} Over several hundred parameter draws in each specification,
with $\delta \in \{0.5,0.7,0.8,0.9,0.95,0.97,0.99\}$,
both termination rules and several random initial beliefs per case,
every equilibrium found had all its recurrent classes equal to single grand states.
This is not for want of room: condition \eqref{eq:prefer} failed in \emph{every} draw,
typically by a wide margin, so Theorem~\ref{thm:arrival} never applied.
In the public-good specification the failure is explained, for the rotating-pairs cycle,
by Remark~\ref{rem:separable}: there $D_i \equiv 0$ identically.
Two economic specifications are of course a small sample, and they are not representative of what is
\emph{possible}: dropping the requirement that payoffs come from such a model produces
Proposition~\ref{prop:counterexample}.
But they suggest that cycling is not what one should expect of a model written down for its
economics rather than for its dynamics.

\emph{An illustration.} Figure~\ref{fig:cournot} shows a computed equilibrium of the Cournot specification with $A=10$,
$c=(0.6,\,0.9,\,1.5,\,2.5)$, weights $w=(0.1,\,0.2,\,0.3,\,0.4)$, unanimous termination and $\delta=0.4$.
Here $V^\ast = 22.12$ and $\gamma = 1.60$, while $S_i/(w_i\gamma)$ ranges from $3.5$ to $20.6$,
so the sufficient condition of Theorem~\ref{thm:arrival} fails by a wide margin;
the grand coalition nevertheless forms within two steps with probability one.
The mechanism is the one emphasised by \citet{HeitzigKornek2018}.
The two efficient firms $A$ and $B$ favour immediate cartelisation.
The two inefficient firms each prefer to stay out for one period:
$D$'s static payoff is $2.28$ when nobody cooperates and $4.31$ in the flat cartel $[ABCD]$,
but $4.33$ if the other three cartelise without her,
and $4.97$ in the cartel $[D[ABC]]$ that she joins afterwards,
because the sharing rule then measures her contribution against a much better reference state.
Free-riding therefore delays the grand coalition and shifts the surplus towards the free-rider,
but does not prevent full cooperation.

\begin{figure}[htbp]
\centering
\begin{tikzpicture}[>=stealth,node distance=1cm,
   st/.style={draw,rectangle,rounded corners=1pt,inner sep=4pt,font=\small},
   gr/.style={draw,rectangle,rounded corners=6pt,double,inner sep=4pt,font=\small},
   every edge/.style={draw,->}]
\node[st] (x0) {$A,B,C,D$};
\node[st] (x1) at (4.2,1.6) {$D,[ABC]$};
\node[st] (x2) at (4.2,0) {$C,[ABD]$};
\node[gr] (g0) at (4.2,-1.6) {$[ABCD]$};
\node[gr] (g1) at (8.6,1.6) {$[D[ABC]]$};
\node[gr] (g2) at (8.6,0) {$[C[ABD]]$};
\draw[->] (x0) -- node[above,sloped,font=\scriptsize]{$40\;D$} (x1);
\draw[->] (x0) -- node[above,font=\scriptsize]{$30\;C$} (x2);
\draw[->] (x0) -- node[below,sloped,font=\scriptsize]{$30\;AB$} (g0);
\draw[->] (x1) -- node[above,font=\scriptsize]{$100\;ABCD$} (g1);
\draw[->] (x2) -- node[above,font=\scriptsize]{$100\;ABCD$} (g2);
\draw[->] (g0) to [out=-30,in=-150,looseness=8] (g0);
\draw[->] (g1) to [out=-30,in=-150,looseness=8] (g1);
\draw[->] (g2) to [out=-30,in=-150,looseness=8] (g2);
\end{tikzpicture}
\caption{A computed equilibrium of the Cournot specification with $A=10$, $c=(0.6,0.9,1.5,2.5)$,
$w=(0.1,0.2,0.3,0.4)$, unanimous termination and $\delta=0.4$.
Arrow labels give the transition probability in percent and the players who favour the move.
Doubled boxes are grand states, which are absorbing by Theorem~\ref{thm:persist-unan};
all other states of the $52$-state space are reached with probability zero.
Notation as in \citet{HeitzigKornek2018}: $[ABC]$ is an agreement signed by $A$, $B$ and $C$ simultaneously,
and $[D[ABC]]$ one signed by $D$ with the already existing coalition $[ABC]$.}\label{fig:cournot}
\end{figure}
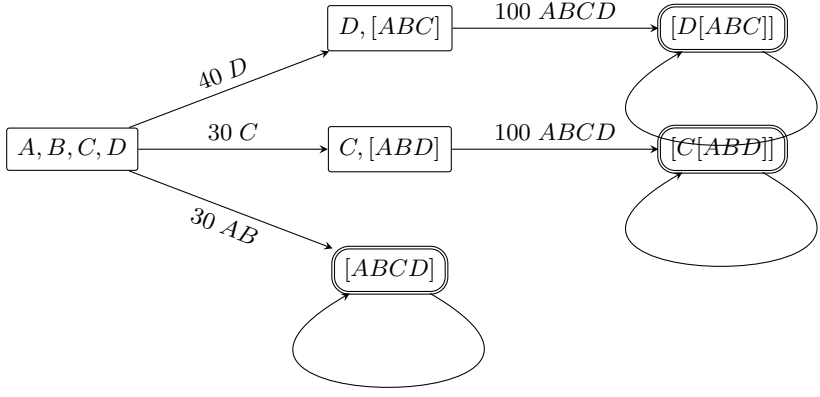

\subsection{Logit play of the protocol}\label{sec:logit}

Theorem~\ref{thm:implement} concerns the limit $\varepsilon\to0$ of best responses.
Playing $P_\varepsilon$ with logistic rather than best responses at a finite inverse temperature $\beta$ turns the axioms into rates, and shows how much of the structure survives noise.

Within a period there is no discounting and, on the path, no delay, so every option contributes $(1-\delta)\pi_i(x)$ plus $\delta$ times a next-period value, the common term cancels from every comparison, and only $\tilde\beta:=\beta\delta$ matters.
At the final vote $j$ accepts with probability $\sigma\bigl(\tilde\beta(\ev_j(\tau(m))-\evp_j(x))\bigr)$, at an amendment vote $k\in R(m')$ votes for $m'$ against $m$ with probability $\sigma\bigl(\tilde\beta(\ev_k(\tau(m'))-\ev_k(\tau(m)))\bigr)$, and a proposer $i$ tables $m$ with probability proportional to $e^{\tilde\beta\ev_i(\tau(m))}$, where $\sigma(u)=(1+e^{-u})^{-1}$.
The two logit arguments are exactly the two comparisons of Definition~\ref{def:rationality}: profitability at the final vote, domination at the amendment votes.
Substitutes are offered in a uniformly random order and the first that passes replaces the proposal;
we take the successful one to be distributed proportionally to its own passing probability, which is the only approximation.
Rejection restarts the round at vanishing cost, so the kernel is the renormalised weight of agreements, and the whole system is closed by \eqref{eq:bellman} and \eqref{eq:table}.
The map is continuous in $p$ on a compact set, so Brouwer gives existence at every finite $\beta$, without the difficulty of Appendix~\ref{app:existence}, which is an artefact of the strict inequalities at $\beta=\infty$.

Two identities are worth recording.
By the definition of $V$, $\sum_y p_{x\to y}\bigl(\ev_j(y)-\evp_j(x)\bigr)=0$ for every $j$ and $x$: realised margins average to zero, which is why profitability had to be a weak inequality.
And no state is absorbing at finite $\beta$, since every stay carries positive weight.

\emph{Numerics.}
The solver is \texttt{code/amend.py}.
For the Cournot specification of Figure~\ref{fig:cournot}, unanimous termination, $\delta=0.4$:
\[
\begin{array}{lcccccc}
\beta & 0.5 & 1 & 2 & 5 & 10 & 20\\[2pt]
1-\mu(\G) & 4.8\cdot10^{-1} & 1.3\cdot10^{-1} & 3.8\cdot10^{-3} & 4.7\cdot10^{-7} & 2.4\cdot10^{-12} & 1.9\cdot10^{-11}\\
\varepsilon_{\mathrm{exit}} & 1.6\cdot10^{-3} & 7.9\cdot10^{-4} & 2.7\cdot10^{-4} & 1.5\cdot10^{-5} & 1.4\cdot10^{-7} & 5.9\cdot10^{-11}\\
T & 854 & 431 & 170 & 32.1 & 7.0 & 2.6
\end{array}
\]
with $T$ the longest expected time to reach $\G$ and $\varepsilon_{\mathrm{exit}}$ the largest probability of leaving it.
Dissolution is far rarer here than under a protocol without the amendment stage, because a dissolution is undercut by a member's stay before it ever reaches a vote.
Figures~\ref{fig:logitgraph} and~\ref{fig:logitflow} show the process at $\beta=0.5$, low enough for terminations to be seen.

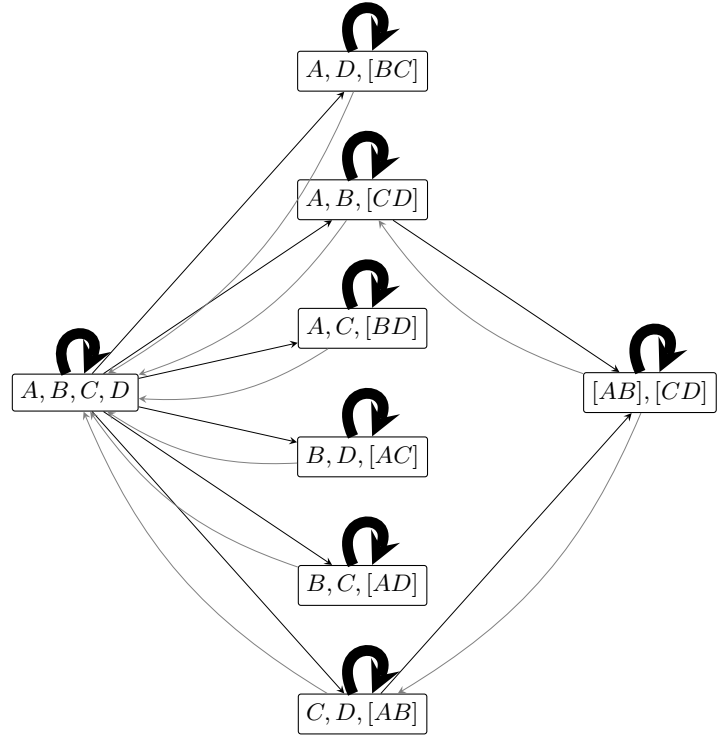
\begin{figure}[htbp]
\centering
\begin{tikzpicture}[>=stealth,
   st/.style={draw,rectangle,rounded corners=1pt,inner sep=3pt,font=\small},
   gr/.style={draw,rectangle,rounded corners=6pt,double,inner sep=3pt,font=\small}]
\node[st] (n0) at (0.00,0.00) {$A,B,C,D$};
\node[st] (n1) at (3.80,-4.25) {$C,D,[AB]$};
\node[st] (n3) at (3.80,-2.55) {$B,C,[AD]$};
\node[st] (n2) at (3.80,-0.85) {$B,D,[AC]$};
\node[st] (n5) at (3.80,0.85) {$A,C,[BD]$};
\node[st] (n6) at (3.80,2.55) {$A,B,[CD]$};
\node[st] (n4) at (3.80,4.25) {$A,D,[BC]$};
\node[st] (n12) at (7.60,0.00) {$[AB],[CD]$};
\draw[->,line width=3.97pt] (n0) to[out=115,in=65,looseness=8] (n0);
\draw[->,line width=0.29pt] (n0) -- (n1);
\draw[->,line width=0.29pt] (n0) -- (n3);
\draw[->,line width=0.29pt] (n0) -- (n2);
\draw[->,line width=0.29pt] (n0) -- (n5);
\draw[->,line width=0.29pt] (n0) -- (n6);
\draw[->,line width=0.29pt] (n0) -- (n4);
\draw[->,line width=0.33pt,gray] (n1) to[bend left=18] (n0);
\draw[->,line width=4.08pt] (n1) to[out=115,in=65,looseness=8] (n1);
\draw[->,line width=0.31pt] (n1) -- (n12);
\draw[->,line width=0.33pt,gray] (n3) to[bend left=18] (n0);
\draw[->,line width=4.08pt] (n3) to[out=115,in=65,looseness=8] (n3);
\draw[->,line width=0.33pt,gray] (n2) to[bend left=18] (n0);
\draw[->,line width=4.08pt] (n2) to[out=115,in=65,looseness=8] (n2);
\draw[->,line width=0.33pt,gray] (n5) to[bend left=18] (n0);
\draw[->,line width=4.08pt] (n5) to[out=115,in=65,looseness=8] (n5);
\draw[->,line width=0.33pt,gray] (n6) to[bend left=18] (n0);
\draw[->,line width=4.08pt] (n6) to[out=115,in=65,looseness=8] (n6);
\draw[->,line width=0.32pt] (n6) -- (n12);
\draw[->,line width=0.33pt,gray] (n4) to[bend left=18] (n0);
\draw[->,line width=4.07pt] (n4) to[out=115,in=65,looseness=8] (n4);
\draw[->,line width=0.33pt,gray] (n12) to[bend left=18] (n1);
\draw[->,line width=0.33pt,gray] (n12) to[bend left=18] (n6);
\draw[->,line width=4.09pt] (n12) to[out=115,in=65,looseness=8] (n12);
\end{tikzpicture}
\caption{Logit play of $P_\varepsilon$, Cournot specification, unanimous termination, $\delta=0.4$, $\beta=0.5$:
the states carrying the bulk of the expected occupation from full disagreement.
Line widths are proportional to the transition probability, $0.25\,\mathrm{pt}+4\,\mathrm{pt}\cdot p_{x\to y}$, self-loops included, arrows below $p=0.002$ omitted.
Doubled boxes are grand states.
Self-loops dominate the picture, as they must: at finite $\beta$ most rounds end in a stay, and at this temperature the expected time to reach a grand state is $854$ periods.}\label{fig:logitgraph}
\end{figure}

\begin{figure}[htbp]
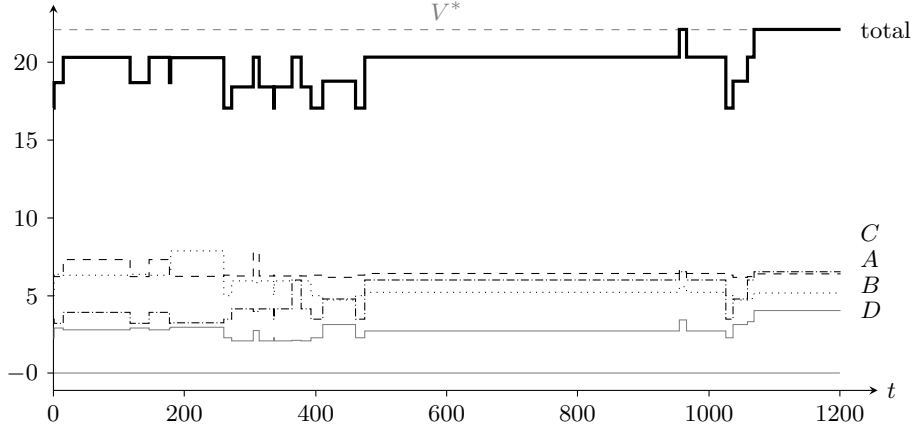

\centering

\caption{One realisation of the same process, $1200$ periods from full disagreement:
the payoff flows $\pi_i(X_t)$ of the four players and their total, as step functions.
Pairs form and dissolve repeatedly, then the nested structures $B,[D[AC]]$ and $C,[ABD]$;
a first grand state $[C[ABD]]$ is held for eleven periods and terminated again, and the run ends absorbed in $[C[D[AB]]]$.
Total flow reaches $V^\ast=22.12$, marked dashed, only in the grand states.}\label{fig:logitflow}
\end{figure}

\begin{proposition}[Escape from the grand states is exponentially rare]\label{prop:metastable}
Assume unanimous termination and Assumption~\ref{ass:eff}, let $h\in\G$ and let $m$ dissolve the grand coalition.
Then, with $\Lambda:=\max_{h'\in\G}(V^\ast-W(h'))$ and $D:=V^\ast-\min_z W(z)$,
\[
\prod_{j\in R(m)}\sigma\bigl(\tilde\beta\Delta_j(m)\bigr)\le\exp\Bigl(-\tfrac{\beta}{n}\bigl(\delta(1-\delta)\gamma-\Lambda\bigr)\Bigr),
\qquad
\Lambda\le\frac{\delta D}{1-\delta}\,\varepsilon_{\mathrm{exit}},
\]
where $\Delta_j(m):=\ev_j(\tau(m))-\evp_j(x)$.
\end{proposition}

\begin{proof}
Put $L:=V^\ast-W$, so $L\ge0$ by Lemma~\ref{lem:welfare} and $L(z)\ge(1-\delta)\gamma$ for $z\notin\G$ by Assumption~\ref{ass:eff} and \eqref{eq:bellman};
summing \eqref{eq:bellman} at a grand $h$, where $V(h)=V^\ast$, gives $L(h)=\delta\sum_z p_{h\to z}L(z)$.
Under unanimous termination $R(m)=N$ by Lemma~\ref{lem:entry}, so with $y:=\tau(m)$,
$\sum_j\Delta_j(m)=W(y)-\sum_j \evp_j(h)=L(h)/\delta-L(y)\le\Lambda/\delta-(1-\delta)\gamma$,
whence $\min_j\Delta_j(m)\le(\Lambda/\delta-(1-\delta)\gamma)/n$;
since $\sigma\le1$ and $\sigma(u)\le e^u$, the product is at most $\exp(\tilde\beta\min_j\Delta_j(m))$, which is the first claim as $\tilde\beta=\beta\delta$.
For the second, $L(h)\le\delta((1-\varepsilon_{\mathrm{exit}})\Lambda+\varepsilon_{\mathrm{exit}}D)$ for every grand $h$.
\end{proof}

The two halves bootstrap, so beyond a threshold in $\beta$ depending only on $\delta,\gamma,D$ and $n$, dissolution is accepted with probability at most $\exp(-\beta\delta(1-\delta)\gamma/2n)$:
Theorem~\ref{thm:persist-unan} survives at finite $\beta$ as a rate.
The exponent vanishes as $\delta\to1$, patient players finding one period outside the grand coalition nearly free.
Counting excursions gives $\mu(\G)\ge(1+\varepsilon_{\mathrm{exit}}T)^{-1}$, which the table above bears out.

\emph{The cycle is a finite-temperature phenomenon.}
Under unilateral termination Proposition~\ref{prop:metastable} does not apply, dissolutions having singleton relevant sets, and the cycle of Proposition~\ref{prop:counterexample} reappears as noisy oscillation.
On its payoff specification at $\beta=20$ the four transitions of the cycle carry probabilities $0.32$, $1.00$, $0.82$ and $1.00$, and $\mu(\G)=0.82$:
the cycle is traversed but leaks into the grand states, from which a single player dissolves again.
Continuing the fixed point upward in $\beta$ shows what it is worth.
The cycle's stationary mass rises to $0.18$ at $\beta=20$, then falls: $0.14$ at $\beta=40$,
$0.034$ at $80$, $1.7\cdot10^{-3}$ at $160$ and below $10^{-4}$ at $320$, with $\mu(\G)\to1$;
the first transition of the cycle tends to $\tfrac12$, an exact indifference, and the third to zero.
So the oscillation is a product of the noise and not a shadow of an equilibrium cycle---as
Theorem~\ref{thm:nodetcycle} requires, since the deterministic limit it would have to approach does
not exist.
Figures~\ref{fig:cycgraph} and~\ref{fig:cycflow} show the regime at $\beta=20$, where it is most
pronounced.

\begin{figure}[htbp]
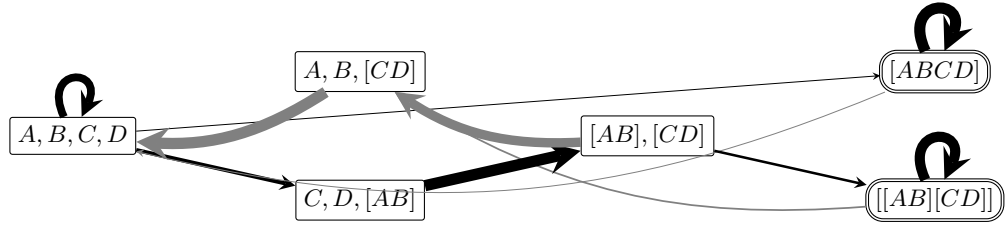

\centering

\caption{Logit play under unilateral termination on the payoff specification of Proposition~\ref{prop:counterexample}, $\beta=20$, the temperature at which the cycle is most visible; conventions as in Figure~\ref{fig:logitgraph}.
The cycle $\emptyset\to C,D,[AB]\to[AB],[CD]\to A,B,[CD]\to\emptyset$ is visible, together with the leakage into $[ABCD]$ and the return from it.}\label{fig:cycgraph}
\end{figure}

\begin{figure}[htbp]
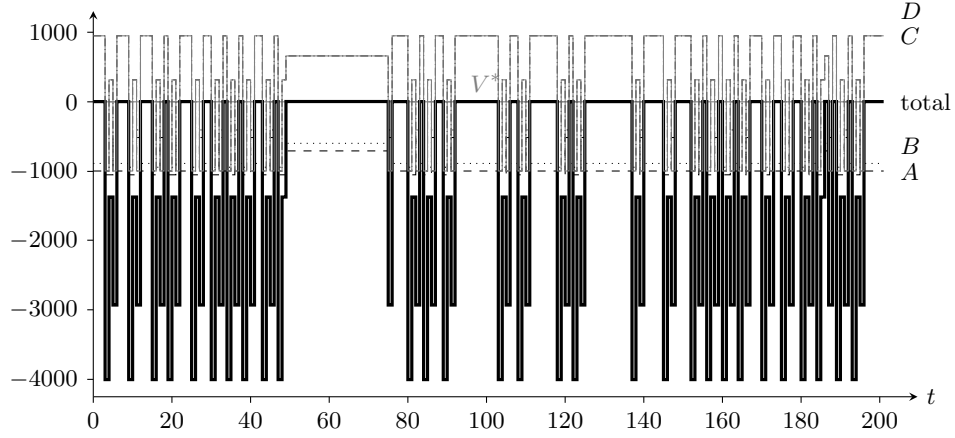

\centering
%
\caption{One realisation of the oscillation, $200$ periods.
Payoffs here are the integer program's, bounded by $1000$, and the states of the cycle are far from efficient, so the swings are large:
total flow is $V^\ast=0$ in the grand states and $-4000$ at $C,D,[AB]$.}\label{fig:cycflow}
\end{figure}

\section{From myopic to farsighted surplus sharing}\label{sec:pivot}

In the running example, this section asks whether a country evaluates a proposed linkage by this
year's abatement costs or by the whole path the agreement sets in motion---and shows that the answer
decides whether patience helps or hurts.
Everything so far has taken the sharing rule to divide the surplus of an agreement measured in
\emph{per-period} payoffs $\pi_i$.
That choice is not forced, and this section shows it is the one that decides how the model behaves in
the limit that matters.
Under the alternative studied here---dividing the surplus measured in discounted long-term
payoffs $\ell_i$---merging becomes Pareto-improving in evaluations rather than merely in static payoffs, the
limit $\delta\to1$ turns from the intractable case into the tractable one, and the question of
whether a cycle can survive becomes a question about the shape of the cycle rather than about the
numbers.

Assumption~\ref{ass:share} shares the \emph{static} surplus $\Delta(x,K)$,
i.e.\ the change in the per-period payoffs of the signatories.
One may ask what happens if the transfers are instead set so that the \emph{long-term} surplus is shared,
i.e.\ if the differences of evaluations rather than of static payoffs are divided in the given proportions.
This section records what changes.
What began as a sketch has turned out to carry the sharpest results in the paper: under this rule the limit $\delta\to1$, intractable elsewhere, becomes the tractable case.

\begin{definition}[Far-sighted sharing]\label{def:farsighted}
A pair $(\pi,p)$ satisfies \emph{far-sighted sharing} with shares $s_i(x,K)>0$, $\sum_{i\in K}s_i(x,K)=1$,
if the evaluations $\ev$ defined from $(\pi,p)$ by \eqref{eq:bellman} satisfy, for every state $x$,
every top-level node $K \in x\cap P(x)$ and every $i \in K$,
\[
   \ev_i(x) - \ev_i(x^{-K}) \;=\; s_i(x,K)\,\sum_{j\in K}\big(\ev_j(x)-\ev_j(x^{-K})\big).
\]
\end{definition}

The rule is implicit rather than explicit: $\ev$ is determined by $\pi$ and $p$,
so a condition on $\ev$ is a condition on the transfers that enter $\pi$,
and payoffs and process must be solved for jointly.
Existence therefore needs a further fixed-point argument on top of the one in Appendix~\ref{app:existence},
and we do not attempt it here.
Note also that the rule is a condition on $\ev$ only through differences between a state and its reference state,
so it does not by itself pin down the level of $\ev$.

What the rule buys is that the Pareto property of merge-all,
which under Assumption~\ref{ass:share} holds in static payoffs and therefore says nothing directly about profitability,
upgrades to a statement about evaluations.

\begin{proposition}\label{prop:farsighted}
Let $(\pi,p)$ satisfy far-sighted sharing and let $\delta \in (0,1)$.
Then for every $x \in \X\setminus\G$ and every $i \in N$,
\[
  \ev_i(g_x)-\ev_i(x) \;=\; s_i(g_x,N)\,\big(W(g_x) - W(x)\big),
\]
which is strictly positive if and only if $W(g_x)>W(x)$, and in particular whenever $g_x$ is absorbing,
since then $W(g_x)=V^\ast>W(x)$ by Lemma~\ref{lem:welfare}.
In that case the merge-all move at $x$ is profitable, and if in addition $p$ satisfies \ref{E2},
then no non-grand state is absorbing:
Theorem~\ref{thm:absorbing-grand} holds under either termination variant and without Criterion~\ref{cond:ties}.
\end{proposition}

\begin{proof}
Apply Definition~\ref{def:farsighted} at $g_x$ with $K=N$, whose reference state is $(g_x)^{-N}=x$:
summing over $j \in N$, the bracket on the right-hand side is $W(g_x)-W(x)$, which gives the display.
If $g_x$ is absorbing then $W(g_x)=V(g_x)=V^\ast$ by \eqref{eq:bellman},
while $W(x)<V^\ast$ by Lemma~\ref{lem:welfare}, since $X_0=x\notin\G$ with probability one.
For the last claim, let $x\notin\G$ be absorbing;
then $\evp_i(x)=\ev_i(x)$ by \eqref{eq:bellman}, so profitability of the merge-all move, whose relevant set is $N$, is exactly the strict inequality just established for all $i$,
and Lemma~\ref{lem:nonstalling} contradicts absorption.
\end{proof}

\subsection{Arrival for large $\delta$}\label{sec:fs-arrival}

Under this rule the limit $\delta \to 1$, which is the hard end everywhere else in this note,
becomes tractable, and for a reason worth isolating.
With the sharing rule of Assumption~\ref{ass:share} the gain from merging is $\varphi_i(x)-\ev_i(x)$,
which by Lemma~\ref{lem:psi} is $w_i\gamma$ plus a difference of relative shares $\psi$,
and those share differences do not vanish as $\delta \to 1$: they are the stakes $S_i$,
and they are exactly what Theorem~\ref{thm:arrival} has to assume away.
Under far-sighted sharing there is no such term.
The merge gain is proportional to $w$ by Proposition~\ref{prop:farsighted}, with no residual redistribution,
so it is bounded below by $w_{\min}\gamma$ as soon as the grand state is absorbing.
An objection must therefore come from inside the recurrent class---and evaluations become constant on a recurrent class as $\delta\to1$.
That is the whole argument.
We first make the rate explicit.

\begin{lemma}[Spread of evaluations on a recurrent class]\label{lem:spread}
Let $p$ be a process and $C$ a closed communicating class of $p$.
For $i \in N$ put $s_i(C) := \max_{x\in C}\pi_i(x)-\min_{x\in C}\pi_i(x)$,
and let $\bar T_C := \max\{\E_x[T_y] : x,y \in C,\ x \ne y\}$, where $T_y$ is the hitting time of $y$.
Then $|\ev_i(x)-\ev_i(y)| \le (1-\delta)\,\bar T_C\,s_i(C)$ for all $x,y \in C$.
If moreover every positive entry of $p$ on $C$ is at least $\eta>0$,
then $\bar T_C \le (|C|-1)\,\eta^{-(|C|-1)}$.
\end{lemma}

\begin{proof}
Fix $x \ne y$ in $C$ and write $T:=T_y$, $s := s_i(C)$ and $c := \min_{z \in C}\pi_i(z)$.
Since $C$ is closed, Lemma~\ref{lem:value} makes $\ev_i$ on $C$ a convex combination of the values $\pi_i(z)$,
$z \in C$, so both $\tilde\pi := \pi_i - c$ and $\tilde\ev := \ev_i - c$ take values in $[0,s]$ on $C$,
and differences of $\ev_i$ equal differences of $\tilde\ev$.
By the strong Markov property at $T$, which is finite almost surely because $C$ is finite and irreducible,
\[
  \tilde\ev(x) \;=\; (1-\delta)\,\E_x\Big[\sum_{t<T}\delta^t\tilde\pi(X_t)\Big] \;+\; \E_x[\delta^{T}]\,\tilde\ev(y),
\]
so that
\[
  \tilde\ev(x)-\tilde\ev(y) \;=\; (1-\delta)\,\E_x\Big[\sum_{t<T}\delta^t\tilde\pi(X_t)\Big] \;-\; \big(1-\E_x[\delta^{T}]\big)\,\tilde\ev(y).
\]
The first term lies in $[0,(1-\delta)\E_x[T]\,s]$, since $\sum_{t<T}\delta^t \le T$ and $0\le\tilde\pi\le s$;
the second lies in the same interval, since $1-\delta^{T}\le(1-\delta)T$ and $0\le\tilde\ev(y)\le s$.
The difference of two numbers in one interval $[0,B]$ has absolute value at most $B$, which is the claim.
For the last sentence, $C$ is irreducible with $|C|$ states,
so from any $x$ there is a path to $y$ of length at most $|C|-1$,
each of whose steps has probability at least $\eta$; hence $\Prob_x(T\le|C|-1)\ge\eta^{|C|-1}$,
and $T$ is stochastically dominated by $|C|-1$ times a geometric variable of success probability $\eta^{|C|-1}$.
\end{proof}

\begin{theorem}[Arrival for large $\delta$ under far-sighted sharing]\label{thm:fs-arrival}
Let $(\pi,p)$ satisfy far-sighted sharing with the fixed shares $s_i(x,K)=w_i/w(K)$ of Condition~\ref{cond:nash},
let $p$ satisfy \ref{E1}--\ref{E3}, and assume every grand state is absorbing.
Then every closed communicating class $C \subseteq \X\setminus\G$ satisfies
\[
   w_{\min}\,\gamma \;\le\; (1-\delta)\;\bar T_C\;\max_{i\in N} s_i(C).
\]
Consequently, if the right-hand side is smaller than $w_{\min}\gamma$, no such class exists,
and the process reaches a grand state almost surely and never leaves it.
\end{theorem}

\begin{proof}
Let $x \in C$.
Since $C$ is closed and contained in $\X\setminus\G$, the chain started at $x$ stays in $C$,
so $V(X_t)\le V^\ast-\gamma$ for every $t$ and hence $W(x)\le V^\ast-\gamma$ by \eqref{eq:series}.
Since $g_x$ is absorbing, $W(g_x)=V(g_x)=V^\ast$ by \eqref{eq:bellman}.
Proposition~\ref{prop:farsighted} therefore gives, for every $i \in N$,
\[
   \ev_i(g_x)-\ev_i(x) \;=\; w_i\big(W(g_x)-W(x)\big) \;\ge\; w_i\,\gamma \;\ge\; w_{\min}\,\gamma .
\]
By Lemma~\ref{lem:objection} there are $i \in N$ and $z \in C$ with $\ev_i(z)\ge\ev_i(g_x)$,
whence $\ev_i(z)-\ev_i(x)\ge w_{\min}\gamma$.
Lemma~\ref{lem:spread} bounds the left-hand side by $(1-\delta)\bar T_C s_i(C)$.
\end{proof}

\begin{corollary}\label{cor:fs-arrival}
Assume termination is unanimous, so that every grand state is absorbing by Theorem~\ref{thm:persist-unan},
and that every positive transition probability is at least some $\eta>0$ not depending on $\delta$.
Consider a family of equilibria indexed by $\delta$,
and suppose the static payoffs grow sublinearly in $1/(1-\delta)$, in the sense that
\[
   (1-\delta)\,\max_{i\in N}\;\max_{x,y\in\X}\big|\pi_i(x)-\pi_i(y)\big| \;\longrightarrow\; 0 \qquad (\delta\to1).
\]
Then for all $\delta$ close enough to $1$ the process reaches a grand state almost surely.
If moreover the static payoffs are bounded by $\Sigma$ uniformly in $\delta$, this holds as soon as
\[
   \delta \;>\; 1-\frac{w_{\min}\,\gamma\,\eta^{\,|\X|-1}}{(|\X|-1)\,\Sigma} .
\]
\end{corollary}

\begin{proof}
By Lemma~\ref{lem:spread}, $\bar T_C\le(|\X|-1)\eta^{-(|\X|-1)}$ for every class $C$,
so the right-hand side in Theorem~\ref{thm:fs-arrival} is at most $(1-\delta)(|\X|-1)\eta^{-(|\X|-1)}\max_i\max_{x,y}|\pi_i(x)-\pi_i(y)|$,
which tends to $0$ by hypothesis and is therefore eventually below $w_{\min}\gamma$.
\end{proof}

The hypothesis is sharp, in the following sense.

\begin{proposition}[Failure of arrival forces the transfers to grow]\label{prop:transfer-growth}
Under the hypotheses of Theorem~\ref{thm:fs-arrival},
if a closed communicating class $C\subseteq\X\setminus\G$ exists at the discount factor $\delta$, then
\[
   \max_{i\in N}\;\big(\max_{x\in C}\pi_i(x)-\min_{x\in C}\pi_i(x)\big) \;\ge\; \frac{w_{\min}\,\gamma}{(1-\delta)\,\bar T_C} .
\]
So arrival can fail at discount factors near $1$ only if the transfers grow at least like $1/(1-\delta)$.
\end{proposition}

\begin{proof}
Rearrange the inequality of Theorem~\ref{thm:fs-arrival}.
\end{proof}

The mechanism is visible in the Bellman equation.
Solving \eqref{eq:bellman} for $\pi$ rather than for $\ev$ gives the identity
\[
   \pi_i \;=\; \ev_i \;+\; \frac{\delta}{1-\delta}\,(I-p)\,\ev_i ,
\]
so the static payoffs stay bounded exactly when the evaluations become harmonic at rate $1-\delta$.
Under far-sighted sharing the evaluations are the primitive object---the sharing rule is a condition on $\ev$,
not on $\pi$---and there is no reason for them to be that nearly harmonic;
the transfers absorb the discrepancy, magnified by $\delta/(1-\delta)$.
\subsection{Arrival under either termination rule}\label{sec:fgains}

Theorem~\ref{thm:fs-arrival} needs the grand states to be absorbing, which under unilateral
termination is exactly what can fail off the equilibrium path.
There is a condition that removes the difficulty at a stroke, and it is the far-sighted analogue of
Condition~\ref{cond:gains}.

\begin{criterion}[Long-run coalitional gains]\label{cond:fgains}
For every state $x$ and every top-level node $K \in x\cap P(x)$,
\[
   \Delta^{\ev}(x,K) \;:=\; \sum_{j\in K}\big(\ev_j(x)-\ev_j(x^{-K})\big) \;\ge\; 0 .
\]
\end{criterion}
In words: the signatories of an agreement are, jointly and in discounted long-term payoffs, no worse
off for having signed it, given everything else in force.
It is the same requirement as Condition~\ref{cond:gains} with $\pi$ replaced by $\ev$, and by
\eqref{eq:bellman} the two differ by the continuation term
$\Delta^{\ev}(x,K)=(1-\delta)\Delta(x,K)+\delta\big(\E_x[L_K]-\E_{x^{-K}}[L_K]\big)$,
so neither implies the other.
Unlike the hypotheses of Theorem~\ref{thm:fs-arrival} it says nothing about grand states, nothing
about the discount factor, and nothing about which termination rule is in force.

\begin{theorem}[Arrival under long-run gains]\label{thm:fgains}
Assume far-sighted sharing with strictly positive shares and Criterion~\ref{cond:fgains}, and let $p$
satisfy \ref{E1} and \ref{E2}.
Then for every $\delta\in(0,1)$ and under either termination rule: no termination is profitable at any
state, every grand state is absorbing, and the process reaches a grand state after at most $n-1$
realised moves, almost surely.
\end{theorem}

\begin{proof}
Let $m$ be a termination at $x$ with bottom node $K_0$, whose ancestors in $x$ are
$K_0\subsetneq K_1\subsetneq\cdots\subsetneq K_r=S(x,K_0)$, and let $y:=\tau(m)$.
Put $x^{(0)}:=x$ and $x^{(t+1)}:=x^{(t)}\setminus\{K_{r-t}\}$ for $0\le t\le r$, so $x^{(r+1)}=y$.
At each step $K_{r-t}$ is a maximal element of $x^{(t)}$, hence top-level there, and
$x^{(t+1)}=(x^{(t)})^{-K_{r-t}}$.
Far-sighted sharing therefore applies at $x^{(t)}$ to the node $K_{r-t}$, and gives, for every
$i \in K_{r-t}$,
\[
   \ev_i(x^{(t)})-\ev_i(x^{(t+1)}) \;=\; s_i\big(x^{(t)},K_{r-t}\big)\,\Delta^{\ev}\big(x^{(t)},K_{r-t}\big) \;\ge\;0
\]
by Criterion~\ref{cond:fgains} and positivity of the shares.
Every $i \in K_0$ lies in every $K_j$, so summing over $t$ telescopes to $\ev_i(x)\ge\ev_i(y)$ for
every $i \in K_0$.
Under either termination rule $R(m)$ meets $K_0$, so for some $i\in R(m)$ we have $\ev_i(x)\ge\ev_i(y)$,
and $s^i_x$ dominates $m$: no termination lies in $\Cset(x)$.

Consequently only merges are realised.
At a grand state no merge is available (Lemma~\ref{lem:entry}), so $\Cset(h)$ contains no moving move,
and \ref{E1} leaves no target other than $h$: the state is absorbing.
A merge strictly decreases $|P(x)| \in\{1,\dots,n\}$, so after at most $n-1$ realised moves the
process is at a state $x$ with no realised move, which by \ref{E2} has $\Cset(x)\cap\M(x)=\emptyset$.
If $x \notin \G$ then $g_x$ is absorbing, so $W(g_x)=V(g_x)=V^\ast$ by \eqref{eq:bellman}, while
$W(x)=V(x)\le V^\ast-\gamma$ because $x$ is absorbing and non-grand;
Proposition~\ref{prop:farsighted} then gives
$\ev_i(g_x)-\ev_i(x)=s_i(g_x,N)\,(V^\ast-V(x))>0$ for every $i$.
Since $x$ is absorbing, $\evp_i(x)=\ev_i(x)$, so the merge-all move is profitable;
Lemma~\ref{lem:domination} then supplies an undominated profitable move whose target $i$ ranks at
least as high, and it is not a stay, since that would give $\ev_i(x)\ge\ev_i(g_x)$.
So $\Cset(x)\cap\M(x)\ne\emptyset$, a contradiction.
Hence $x \in \G$.
\end{proof}

The condition is not merely sufficient; it is what a cycling equilibrium must violate.

\begin{corollary}\label{cor:fgains-necessary}
Under far-sighted sharing, if an equilibrium has a closed communicating class $C\subseteq\X\setminus\G$
then $\Delta^{\ev}(x^{(t)},K_{r-t})<0$ for some realised termination at some $x \in C$ and some node
$K_{r-t}$ on its chain.
\end{corollary}

\begin{proof}
By Lemma~\ref{lem:entry} a class in $\X\setminus\G$ that is closed and communicating cannot consist of
merges alone, since merges strictly decrease $|P(x)|$; so some realised move on $C$ is a termination,
and by \ref{E1} it is undominated, so no terminator's stay dominates it and every terminator strictly gains.
The displayed chain in the proof of Theorem~\ref{thm:fgains} then cannot have all its terms
non-negative.
\end{proof}

\begin{example}\label{ex:fgains-fails}
In the far-sighted equilibrium of Section~\ref{sec:counterexample}, at $\delta=2/5$, the two
terminations realised on the cycle are exactly where Criterion~\ref{cond:fgains} fails:
$\Delta^{\ev}(x_4,\{C,D\})=-102.86$ and $\Delta^{\ev}(x_3,\{A,B\})=-14.19$.
The second is instructive: there the \emph{static} surplus $\Delta(x_3,\{A,B\})$ is exactly $0$, so
the pair loses nothing per period by having signed, and yet in long-term payoffs it loses, because
having signed changes where the process goes next.
Criterion~\ref{cond:fgains} is therefore a genuine strengthening of Condition~\ref{cond:gains} and not
a restatement of it.
\end{example}

\subsection{The condition in primitives, and existence}\label{sec:pdagger}

Criterion~\ref{cond:fgains} is stated in evaluations, so it is not something one can read off a
partition function.
On one process it becomes primitive, and that process is the natural candidate: $p^\dagger$, which
merges everything at once from wherever it is.
It also disposes of the existence question left open in Section~\ref{sec:pivot}, since on
$p^\dagger$ the transfers can be written down rather than solved for.

\begin{proposition}[Far-sighted sharing on $p^\dagger$]\label{prop:pdagger}
Let $p^\dagger$ be the process of Section~\ref{sec:onestep}, and impose far-sighted sharing with the
fixed shares of Condition~\ref{cond:nash}.
Then the system has a unique solution, and it is the \emph{myopic} weighted-Nash payoff structure:
$\pi_i(x)-\pi_i(x^{-K}) = \frac{w_i}{w(K)}\Delta(x,K)$ for every top-level $K$, and in particular
$\pi_i(g_x)=\pi_i(x)+w_i(V^\ast-V(x))$.
Moreover, for every $x\in\X\setminus\G$ and every top-level $K\in x$,
\[
   \Delta^{\ev}(x,K) \;=\; \Delta(x,K) \;-\; \delta\,w(K)\big(V(x)-V(x^{-K})\big),
\]
so that on $p^\dagger$ Criterion~\ref{cond:fgains} is the primitive requirement
\begin{equation}\label{eq:fgains-primitive}
   \Delta(x,K) \;\ge\; \delta\,w(K)\,\big(V(x)-V(x^{-K})\big)
   \qquad\text{for every } x \text{ and every top-level } K.
\end{equation}
\end{proposition}

\begin{proof}
Under $p^\dagger$ the unique successor of $x\notin\G$ is $g_x$, which is absorbing, so
$\ev_i(g_x)=\pi_i(g_x)$ and $\ev_i(x)=(1-\delta)\pi_i(x)+\delta\pi_i(g_x)$.
Far-sighted sharing at $g_x$ for $K=N$, whose reference state is $x$, reads
$\ev_i(g_x)-\ev_i(x)=w_i\sum_j(\ev_j(g_x)-\ev_j(x))$.
The left side is $(1-\delta)(\pi_i(g_x)-\pi_i(x))$ and the sum on the right is
$(1-\delta)(V^\ast-V(x))$, so $\pi_i(g_x)=\pi_i(x)+w_i(V^\ast-V(x))$.
Now let $K$ be top-level in $x$; then $x^{-K}\notin\G$, and substituting the last display twice,
\[
   \ev_i(x)-\ev_i(x^{-K}) = \big(\pi_i(x)-\pi_i(x^{-K})\big) - \delta\,w_i\big(V(x)-V(x^{-K})\big).
\]
Summing over $i \in K$ gives the stated formula for $\Delta^{\ev}(x,K)$, and the far-sighted
condition $\ev_i(x)-\ev_i(x^{-K}) = \frac{w_i}{w(K)}\Delta^{\ev}(x,K)$ then reduces, the two terms in
$\delta w_i(V(x)-V(x^{-K}))$ cancelling, to
$\pi_i(x)-\pi_i(x^{-K}) = \frac{w_i}{w(K)}\Delta(x,K)$, which determines $\pi$ uniquely by induction
on the number of nodes.
\end{proof}

Two things are worth drawing out.

First, \eqref{eq:fgains-primitive} is a statement about externalities.
$\Delta(x,K)$ is what the signatories of $K$ gain, and $V(x)-V(x^{-K})$ is what everybody gains;
the condition asks that the signatories capture at least the fraction $\delta\,w(K)$ of what they
create.
It therefore fails when an agreement confers large benefits on outsiders and its signatories carry
much bargaining weight, and holds comfortably when agreements are close to privately appropriable.
At $K=N$ it is automatic, since there $\Delta(g_x,N)=V^\ast-V(x)$ and $w(N)=1$, so it reads
$V^\ast-V(x)\ge\delta(V^\ast-V(x))$.

Second, the coincidence in Proposition~\ref{prop:pdagger} is not an accident of $p^\dagger$ alone but
of any process that reaches $\G$ in one step: over a single period there is nothing for the two
sharing rules to disagree about.
It has a useful consequence for existence.

\begin{corollary}[Constructive existence]\label{cor:constructive}
Suppose that in a process only merges are ever realised.
Then the pair $(\pi,\ev)$ solving the far-sighted sharing system is determined by backward induction
on $|P(x)|$, with no fixed-point argument: the grand states are absorbing, so $\ev=\pi$ there and the
sharing condition at $N$ determines $\pi$; and a state with $|P(x)|=k$ has all its successors at
$|P|<k$.
In particular this applies whenever Criterion~\ref{cond:fgains} holds, since by
Theorem~\ref{thm:fgains} no termination is then realised.
\end{corollary}

This is the one place where the implicitness of the far-sighted rule, flagged in
Remark~\ref{rem:fs-bounds}, does no damage: the circularity that forces a fixed point in general is
broken by acyclicity of the realised moves.
Whether $p^\dagger$ itself is an \emph{equilibrium} is a separate question, and by
Proposition~\ref{prop:pdagger} it is exactly the question already answered by
Theorem~\ref{thm:onestep}, since the payoffs are the same: it is, if no move improves the relative
share of all the players whose approval it needs.

\subsection{A cooperative-game reading}\label{sec:coop}

The construction of Proposition~\ref{prop:pdagger} invites a comparison with cooperative solution
concepts, and the comparison is worth making, both because it identifies what the sharing rule is and
because it exposes a tension in the previous subsection.

\emph{What the sharing rule is.}
A state is a laminar hierarchy, which is precisely a \emph{level structure} in the sense of
\citet{Winter1989}, of which a coalition structure \citep{AumannDreze1974} and an
\emph{a priori} union structure \citep{Owen1977} are the one-level cases.
The rule of Assumption~\ref{ass:share} builds payoffs by bargaining once at each level against the
state in which that level's agreement is absent, and with the weights of Condition~\ref{cond:nash}
each such bargain is a weighted Shapley value of a unanimity game \citep{Kalai1977,KalaiSamet1987}.
Whether the resulting map from states to payoff vectors coincides with the level-structure value of
\citet{Winter1989}, or merely resembles it, we have not determined; the two are built from the same
ingredients but not by the same recursion, and settling it would be a worthwhile exercise.

\emph{Undomination of the merge is a blocking condition.}
Fix a far-sighted rule with arbitrary strictly positive shares $s_i(x,K)$ and consider $p^\dagger$ at
the fully non-cooperative state $x_0$.
Every state reachable from $x_0$ leads to a grand state, so each route determines an imputation of
$V^\ast$: going directly gives
$y^{\,\emptyset}_i := \pi_i(x_0)+s_i(g_{x_0},N)\big(V^\ast-V(x_0)\big)$,
and forming $K$ first and merging afterwards gives
$y^{K}_i := \pi_i(x_0)+s_i(\{K\},K)\,\Delta_K+s_i(g_{\{K\}},N)\big(V^\ast-V(\{K\})\big)$
for $i \in K$, where $\Delta_K := \Delta(\{K\},K)$.
Merge-all is undominated at $x_0$ exactly when no $K$ blocks $y^{\,\emptyset}$ in favour of $y^{K}$,
that is, when for every $K$ some $i \in K$ has $y^K_i \le y^{\,\emptyset}_i$.
This is a core-like requirement, but with a blocking notion \emph{stronger} than the classical one:
a deviating coalition does not consume its own worth for ever, it forms and then enters the grand
coalition on the better terms its prior formation has bought it.
The set of imputations surviving such blocking is therefore contained in the classical core, and the
question of which far-sighted rules make $p^\dagger$ an equilibrium is the question of which shares
place $y^{\,\emptyset}$ in that smaller set.
Whether it is ever non-empty when the classical core is, and whether a canonical selection such as
the nucleolus \citep{Schmeidler1969} or an asymmetric Shapley value lands in it, we leave open; the
question seems to us the natural cooperative counterpart of Theorem~\ref{thm:onestep}.

\emph{A tension.}
Under Condition~\ref{cond:nash} the two conditions can be computed and compared, and they pull
against each other.
Writing $\Delta V_K := V(\{K\})-V(x_0)$, the calculation of Proposition~\ref{prop:pdagger} gives
$\pi_i(\{K\})=\pi_i(x_0)+\frac{w_i}{w(K)}\Delta_K$ and
$\pi_i(g_y)=\pi_i(y)+w_i(V^\ast-V(y))$, whence
\[
   \ev_i(\{K\})-\ev_i(g_{x_0}) \;=\; w_i\left[\frac{\Delta_K}{w(K)}
   -(1-\delta)\big(V^\ast-V(x_0)\big)-\delta\,\Delta V_K\right],
\]
whose sign is the same for every $i \in K$.
So merge-all is undominated at $x_0$ if and only if
\[
   \Delta_K \;\le\; w(K)\Big[(1-\delta)\big(V^\ast-V(x_0)\big)+\delta\,\Delta V_K\Big]
   \qquad\text{for every } K,
\]
while \eqref{eq:fgains-primitive} at the state $\{K\}$ requires
$\Delta_K \ge \delta\,w(K)\,\Delta V_K$.
Both hold only inside an interval of width $w(K)(1-\delta)\big(V^\ast-V(x_0)\big)$, which closes as
$\delta\to1$.

\begin{remark}\label{rem:pdagger-tension}
This limits what Proposition~\ref{prop:pdagger} can be used for, and we would rather say so than let
the reader discover it.
The primitive form \eqref{eq:fgains-primitive} of Criterion~\ref{cond:fgains} is computed
\emph{on} $p^\dagger$, and the display above shows that when it holds strictly at large $\delta$,
$p^\dagger$ is dominated at $x_0$ and hence is not an equilibrium.
The two are not contradictory---Theorem~\ref{thm:fgains} is a statement about an equilibrium's own
evaluations, and its conclusion is arrival within $n-1$ merges, not arrival in one---but they say
that under long-run gains the grand coalition is generically approached through several mergers
rather than one, which is the phenomenon \citet{HeitzigKornek2018} report as a global market emerging
``but probably not in one move''.
What \eqref{eq:fgains-primitive} therefore is, honestly, is a benchmark computation rather than a
test one can apply to an arbitrary equilibrium; making Criterion~\ref{cond:fgains} primitive on a
process that is itself an equilibrium remains open.
\end{remark}

\subsection{Reading the growth off the equations}\label{sec:growth}

The question of how fast the transfers grow can in fact be settled from the system itself,
by eliminating $\pi$.
Substituting $\pi = \ev + \frac{\delta}{1-\delta}(I-p)\ev$ into the constraint that each top-level block receives its partition-function value,
$\sum_{i\in K}\pi_i(x) = v(K;P(x))$, and writing $L_K := \sum_{i\in K}\ev_i$, that constraint becomes
\begin{equation}\label{eq:blockbellman}
   L_K(x) \;-\; \delta\,(p\,L_K)(x) \;=\; (1-\delta)\,v\big(K;P(x)\big)
   \qquad\text{whenever } K \in P(x).
\end{equation}
Together with the far-sighted sharing conditions, which do not involve $\delta$ at all,
this is a square linear system
\[
   A(\delta)\,\ev \;=\; (1-\delta)\,b
\]
in the $|\X|\cdot n$ unknowns $\ev_i(x)$: one equation \eqref{eq:blockbellman} per state and top-level block,
and $|K|-1$ sharing equations per block with $|K|\ge2$, which is $n$ equations per state.
Three things follow.

First, $A(1)$ is singular:
every vector $\ev_i(x)=c_i$ constant in $x$ satisfies both families of equations with zero right-hand side.
So $A(\delta)^{-1}$ has a pole at $\delta=1$,
and the factor $(1-\delta)$ on the right cancels one order of it.
This is why the evaluations stay bounded although the transfers do not:
the evaluations are the object the system is written in, and the pole is spent on the right-hand side.

Second, the criterion for the transfers is exact.
Writing $\ev_0 := \lim_{\delta\to1}\ev(\delta)$,
which exists because $\ev(\delta)$ is a rational function of $\delta$ and is bounded,
\[
   \pi(\delta) \text{ stays bounded as } \delta\to1
   \qquad\Longleftrightarrow\qquad
   (I-p)\,\ev_0 = 0 ,
\]
that is, exactly when the limiting evaluation of \emph{each individual player} is $p$-harmonic.
Setting $\delta=1$ in \eqref{eq:blockbellman} shows only that the block \emph{sums} $L_K$ are harmonic,
at those states where $K$ is a block;
the transfers blow up precisely to the extent that individual evaluations fail to inherit that.

Third---and this is what makes the criterion usable---individual harmonicity can be read off the combinatorics of the class.

\begin{lemma}[Two sources of individual harmonicity]\label{lem:harmonic-sources}
Let $C$ be a closed communicating class on which $p$ is deterministic,
write $\succ(x)$ for the successor of $x$, and let $\ev_0$ be as above.
Fix $x \in C$ and $i \in N$, and suppose either
\begin{enumerate}[label=(\roman*)]
\item[$(\alpha)$] $\{i\}$ is a top-level block of $x$, that is, $i$ has signed nothing at $x$; or
\item[$(\beta)$] the move realised at $x$ terminates the top-level block $K$ of $x$ containing $i$.
\end{enumerate}
Then $\ev_{0,i}(x)=\ev_{0,i}(\succ(x))$.
\end{lemma}

\begin{proof}
$(\alpha)$ Apply \eqref{eq:blockbellman} with $K=\{i\}$ and let $\delta\to1$:
$\ev_{0,i}(x)=(p\,\ev_{0,i})(x)=\ev_{0,i}(\succ(x))$.

$(\beta)$ Since $K$ is top-level it has no ancestors,
so the chain rule makes the target of the termination exactly $x^{-K}$; hence $\succ(x)=x^{-K}$.
Applying \eqref{eq:blockbellman} to $K$ at $x$ and letting $\delta\to1$ gives $L_K(x)=L_K(\succ(x))=L_K(x^{-K})$,
so the long-term surplus of $K$ at $x$ vanishes in the limit.
The far-sighted sharing condition at $x$ for $K$ then reads $\ev_{0,i}(x)-\ev_{0,i}(x^{-K}) = \frac{w_i}{w(K)}\big(L_K(x)-L_K(x^{-K})\big)=0$.
\end{proof}

\begin{theorem}[A combinatorial criterion]\label{thm:combinatorial}
Let $C$ be a closed communicating class on which $p$ is deterministic,
and for $i \in N$ let $H_i := \{x \in C : (\alpha) \text{ or } (\beta) \text{ holds for } i \text{ at } x\}$.
If for every $i$ the edges $\{x \to \succ(x) : x \in H_i\}$ connect $C$ --- for which it suffices that $|H_i|\ge|C|-1$ --- then $\ev_{0,i}$ is constant on $C$ for every $i$,
and consequently
\[
   \max_{i\in N}\Big(\max_{x\in C}\ev_i(x)-\min_{x\in C}\ev_i(x)\Big) \;\longrightarrow\; 0 \qquad (\delta\to1).
\]
If in addition every grand state is absorbing,
then by Theorem~\ref{thm:fs-arrival} $C$ cannot be a closed class for $\delta$ close enough to $1$.
\end{theorem}

\begin{proof}
By Lemma~\ref{lem:harmonic-sources},
$\ev_{0,i}$ takes equal values at the two ends of every edge $x\to\succ(x)$ with $x \in H_i$;
if those edges connect $C$ then $\ev_{0,i}$ is constant on $C$.
The displayed limit is then $\max_i\operatorname{spread}_C(\ev_{0,i})=0$.
The last sentence is Theorem~\ref{thm:fs-arrival},
whose left-hand side $w_{\min}\gamma$ is a positive constant.
\end{proof}

\begin{example}\label{ex:cycle-dies}
The cycle of Proposition~\ref{prop:counterexample} satisfies the criterion, and comfortably.
Its four states, with their top-level blocks and the move realised at each, are
\[
\begin{array}{lccl}
x_1=A,B,C,D & A|B|C|D & \text{merge } AB & (\alpha)\ \text{for } A,B,C,D\\
x_2=[AB],C,D & AB|C|D & \text{merge } CD & (\alpha)\ \text{for } C,D\\
x_3=[AB],[CD] & AB|CD & \text{terminate } AB & (\beta)\ \text{for } A,B\\
x_4=A,B,[CD] & A|B|CD & \text{terminate } CD & (\alpha)\ \text{for } A,B,\ (\beta)\ \text{for } C,D
\end{array}
\]
so $H_A=H_B=\{x_1,x_3,x_4\}$ and $H_C=H_D=\{x_1,x_2,x_4\}$:
every player fails at exactly one of the four states, and $|H_i|=3=|C|-1$.
Theorem~\ref{thm:combinatorial} therefore predicts that this cycle disappears as $\delta\to1$,
and the computations agree: the spread of $\ev$ over the class is $82.4$, $56.0$,
$31.1$ and $25.4$ at $\delta=0.3,\,0.5,\,0.7,\,0.75$, close to $112(1-\delta)$ throughout and heading for $0$,
while Theorem~\ref{thm:fs-arrival} needs it to stay above $w_{\min}\gamma=4$.
\end{example}

\begin{remark}
The criterion is a statement about the \emph{shape} of a candidate cycle, not about payoffs,
and it is demanding to violate:
$(\alpha)$ fails for $i$ only at states where $i$ has already signed something,
and $(\beta)$ then requires the move there to dissolve $i$'s own top-level agreement.
A cycle escaping the criterion must therefore keep some player inside an agreement that she does not dissolve,
at two or more of its states.
We have not determined whether such a cycle can be an equilibrium;
that is now the precise form of the open question for far-sighted sharing.
\end{remark}

This is what the computations show.
For the family of equilibria of Proposition~\ref{prop:counterexample} under far-sighted sharing and unilateral termination,
the evaluations stay bounded, with $\max_{x,i}|\ev_i(x)|$ between $80$ and $205$ throughout,
while the largest transfer grows from $64$ at $\delta=1/2$ to about $10^3$ at $\delta=0.999$ and $3\times10^4$ at $\delta=0.9999$ --- of the order of $1/(1-\delta)$,
as Proposition~\ref{prop:transfer-growth} requires of any family of cycling equilibria.

\begin{remark}\label{rem:fs-bounds}
Neither hypothesis of the corollary is free.
The bound $\eta$ is harmless under the probability rule of \citet{HeitzigKornek2018}:
there a realised move's probability is the total bargaining weight of the players favouring it,
so if a player facing several tied favourites splits her weight equally among them,
every positive transition probability is at least $w_{\min}/|\X|$, independently of $\delta$.
The bound $\Sigma$ is a genuine restriction,
because under far-sighted sharing the static payoffs are \emph{endogenous}:
the transfers solve an implicit system involving $\ev$, hence $\delta$.
In fact they are \emph{not} bounded,
and Proposition~\ref{prop:transfer-growth} below shows that this is forced rather than accidental.
What can be salvaged is that the hypothesis needed is not boundedness but a growth rate, and a generous one.
\end{remark}

\begin{remark}\label{rem:fs-vs-counterexample}
Theorem~\ref{thm:fs-arrival} does not contradict the far-sighted computation reported after Proposition~\ref{prop:counterexample},
and the numbers there show why.
That equilibrium uses \emph{unilateral} termination,
and two of the four grand states above the cycle are not absorbing, so the hypothesis fails;
indeed the merge gains at those two states, $38.40$ and $2.40$ at $\delta=2/5$,
are not both bounded below by $w_{\min}\gamma=4$, the second falling short,
whereas at the two states whose grand state is absorbing the gains are $6.36$ and $9.91$,
above $4$ as the proof requires.
Moreover the cycle lives at moderate $\delta$:
in that family of equilibria the spread of $\ev$ over the class is $82.4$, $69.2$,
$43.2$ and $37.1$ at $\delta=0.3,\,0.4,\,0.6,\,0.65$,
shrinking like $(1-\delta)$ as Lemma~\ref{lem:spread} predicts and staying far above the value $4$ that Theorem~\ref{thm:fs-arrival} would need it to fall below.
\end{remark}

Two further remarks.
The persistence theorems are unaffected, since neither uses any sharing rule:
Theorem~\ref{thm:persist-unan} uses only Lemma~\ref{lem:welfare},
and Theorem~\ref{thm:persist-unil} only the shadow lemma and the definition of domination.
The arrival problem, on the other hand, is not resolved by the far-sighted rule:
profitability of merge-all is necessary but not sufficient for it to be realised,
since it may still be dominated or fail to be anyone's favourite,
which is precisely the obstruction analysed in Lemma~\ref{lem:objection} and Proposition~\ref{prop:arrival-unil}.
This is not merely a gap in the argument.
The counterexample of Proposition~\ref{prop:counterexample},
computed under this rule rather than under Assumption~\ref{ass:share}, still cycles for ever,
and it does so with the merge-all move profitable for every player at every state of the cycle:
the rule removes the gap between static and long-term payoffs that hypothesis \eqref{eq:prefer} exists to bridge,
and domination alone suffices to prevent arrival.

\section{Discussion}\label{sec:discussion}

This section returns to the three things promised in Section~\ref{sec:intro}: which hypotheses are
doing the work, why the limit $\delta\to1$ behaves the way it does, and how much of the model should
be believed.

\paragraph{Why $\delta\to1$ is the demanding limit.}
Persistence holds at every $\delta$ and restricts it in no way,
while the arrival results either bound $\delta$ from above (Theorem~\ref{thm:smalldelta},
Proposition~\ref{prop:flat}) or bound the distributional stakes uniformly in $\delta$
(Theorem~\ref{thm:arrival}).
That the limit promised in Section~\ref{sec:delta} is the demanding one deserves a word of
explanation, since the opposite is often true in dynamic models, and it is not an artefact of the
proofs.
The aggregate cost of one period spent outside the grand coalition is $(1-\delta)(V^\ast - V(x))$,
which vanishes as $\delta \to 1$; holding out for a better division becomes cheap, not expensive.
Moreover the relative shares $\psi_i$ introduced in Section~\ref{sec:decomp} sum to zero,
so at every grand state some player has a strictly better grand state available in principle.
Consequently no argument that merely compares payoff \emph{levels} can establish absorption as $\delta \to 1$;
what does establish it is the move structure,
which is why Theorems~\ref{thm:persist-unan}--\ref{thm:absorbing-grand} hold for every $\delta$,
whereas every quantitative sufficient condition involving $\delta$ necessarily degrades as $\delta \to 1$.

\paragraph{Which hypothesis carries which theorem.}
Assumption~\ref{ass:share} gives Lemma~\ref{lem:pareto}, without which nothing works;
only its weak form is needed, that every signatory receives a strictly positive share.
Assumption~\ref{ass:eff} gives $\gamma>0$ and the welfare bound.
The chain rule gives Lemmas~\ref{lem:entry} and~\ref{lem:shadow}, hence both persistence theorems.
Unanimous termination gives the stronger persistence statement (Theorem~\ref{thm:persist-unan}) and removes the residual case of Remark~\ref{rem:open-unil};
unilateral termination gives a slightly weaker statement plus Criterion~\ref{cond:ties}.
Condition~\ref{cond:gains} is used only for small $\delta$,
and Condition~\ref{cond:nash} only for the $\psi$-decomposition and what quotes it;
see Table~\ref{tab:map}.

\paragraph{Where the gap lies.}
It is worth restating which configuration is unresolved, because it is not the one a reader would
guess from the pattern of results.
Persistence holds everywhere.
Arrival holds for small $\delta$, for three players, under condition~\eqref{eq:prefer}, and---under
far-sighted sharing and unanimous termination---for $\delta$ near one.
It fails, under either termination rule, at moderate $\delta$.
Under far-sighted sharing the unilateral case is no longer open: Theorem~\ref{thm:fgains} settles it
for every discount factor under Criterion~\ref{cond:fgains}, and
Corollary~\ref{cor:fgains-necessary} shows that a cycling equilibrium must violate that condition at
a state where it terminates something.
What is open is arrival under unilateral termination and the sharing rule of
Assumption~\ref{ass:share} as $\delta\to1$; and by the argument of Section~\ref{sec:realism} that is
the configuration the motivating application most plausibly satisfies.
A negative answer there would say that fast, freely terminable, bilaterally negotiated agreements
need not converge on full cooperation even when everyone is patient; a positive one would say that
patience suffices after all.
We do not know which.

\paragraph{How realistic is the model?}\label{sec:realism}

Now that the results are in hand we can say which of the modelling choices deserve an argument
rather than a stipulation.
Two of them do, and they interact.

\emph{Termination should be unilateral.}
The model's two variants differ in who must approve a dissolution: all the signatories of the
top-level agreement that is destroyed, or a single member of the agreement being terminated.
The second is the realistic one.
Agreements between sovereign parties almost always carry a withdrawal clause that any single party
may exercise---the Paris Agreement has one, as do most treaties and most commercial contracts---while
formation invariably requires everyone's signature.
The asymmetry is not incidental: it is what distinguishes an agreement from an institution, and it is
exactly the feature (R) that yields the Reversal Lemma.
Requiring unanimity to dissolve is, in effect, assuming an external enforcement mechanism that keeps
a party in an agreement it wants to leave.
Where such a mechanism exists---a supranational court, a collateralised contract---the unanimous
variant is the right one; where it does not, the unilateral variant is.

This matters more, not less, for automated negotiation.
When the parties delegate the search for agreements to software agents, the enforcement question does
not go away, but the practical ability of a single agent to stop executing an agreement is if
anything greater, and the natural default is that any signatory can walk.
An agent can always cease to comply; making it unable to do so requires machinery outside the
negotiation itself.

\emph{And then the discount factor is close to one.}
As Section~\ref{sec:delta} argues, $\delta$ aggregates time preference, period length and trust that
the process continues, and $\delta\to1$ is what one obtains by letting the period length go to zero
at fixed impatience.
Automated negotiation is exactly the regime in which the period length collapses: the interval
between opportunities to change the state is measured in seconds rather than in years of diplomacy.

Putting the two together is uncomfortable, and we prefer to say so plainly rather than bury it.
\emph{Unilateral termination together with $\delta$ near one, under the sharing rule of
Assumption~\ref{ass:share}, is precisely the case this paper leaves open}---see
Table~\ref{tab:map} and Remark~\ref{rem:open-unil}.
Persistence is settled there (Theorem~\ref{thm:persist-unil}), and the counterexample of
Proposition~\ref{prop:counterexample} shows that arrival can fail under that termination rule at
moderate $\delta$; what is not known is whether it can fail as $\delta\to1$.
The one setting in which we can answer the question in the limit, Section~\ref{sec:pivot}, requires
the sharing rule to be far-sighted and, for the cleanest statement, termination to be unanimous.
So the gap in our results is not a corner case chosen for convenience; it is the configuration that
the motivating application most plausibly satisfies.
Remark~\ref{rem:continuous-time} sharpens this: in the continuous-time reading, faster negotiation
raises $\delta$ towards one without changing anything else, so the open case is not merely plausible
but is what one approaches by making the agents quicker.

\emph{Concurrency is not the objection it appears to be.}
That exactly one move occurs per period looks, at first, like a serious restriction once negotiations
are fast and parallel, since simultaneous independent agreements among disjoint groups are what one
would expect of automated negotiation at scale.
But real time is continuous, and negotiations that run concurrently do not \emph{conclude}
concurrently: generically one concludes first, and once that is observed the others are being
conducted against a changed status quo and must be re-evaluated, because what a party gains from an
agreement depends on the state in which it is signed---that is precisely the content of
Assumption~\ref{ass:share}.
Consecutive moves are therefore not an assumption that negotiations are serial.
They are a description of the sequence of \emph{conclusions}, which is all the model needs.

This can be made precise, and doing so returns two of the model's stipulations as consequences.

\begin{remark}[A continuous-time reading]\label{rem:continuous-time}
Suppose each feasible move $m$ at the current state concludes at the first arrival of an independent
Poisson clock of rate $\lambda_m>0$, and that payoffs are discounted continuously at rate $\rho>0$.
Then almost surely no two moves conclude at the same instant.
The embedded jump chain is a discrete-time process on $\X$ in which $m$ is realised with probability
$\lambda_m/\Lambda$, where $\Lambda := \sum_{m}\lambda_m$; the waiting time $\tau$ between
consecutive conclusions is exponential with mean $1/\Lambda$; and the discount factor per realised
move is
\[
   \delta \;=\; \E\big[e^{-\rho\tau}\big] \;=\; \frac{\Lambda}{\Lambda+\rho}.
\]
Memorylessness of the exponential clocks is exactly the ``restart from the new status quo'' step: a
negotiation interrupted by another's conclusion carries no residual progress and begins afresh at the
new state.
Two features of the model then follow rather than being assumed.
Transition probabilities proportional to the total bargaining weight of the players favouring a move
are what one obtains by setting $\lambda_m$ proportional to that weight.
And $\delta\to1$ as $\Lambda\to\infty$ at fixed $\rho$: faster negotiation, in the literal sense of
more conclusions per unit of real time, is exactly the patient limit of
Section~\ref{sec:delta}.
For $\delta$ to be a constant rather than a function of the state, $\Lambda$ must not depend on the
state, which amounts to a fixed negotiating capacity; we assume it.
\end{remark}

What the argument does require is that a negotiation \emph{in progress} be unobservable, or at least
uninformative: the model's states carry no record of pending moves, so a player's evaluation cannot
condition on them.
In inter-governmental negotiation that is doubtful, since ongoing talks are visible and are
themselves instruments of bargaining.
Among automated agents concluding in seconds it is more plausible, which is a rare instance of the
fast setting being kinder to the model than the slow one.

\emph{What is clearly unrealistic.}
Two assumptions are made for tractability and should be read as such.
Bargaining weights are exogenous and fixed, where in reality bargaining power is itself an outcome---
indeed the free-riding effect of Section~\ref{sec:decomp} is precisely a manoeuvre to improve one's
position, so the model already contains the phenomenon it holds fixed.
And beliefs are common and self-confirming, which is a strong coordination requirement.
We regard the second as the more serious of the two, and more serious in the automated setting than
in the diplomatic one: a common, correct model of how the process will unfold is a demanding thing to
attribute to a population of independently designed agents, and it is an empirical question rather
than a modelling convenience.

\paragraph{Bounded transfers.}\label{sec:bounded}
Section~\ref{sec:tu} flagged that transfers are unbounded.
Dropping that is not a refinement but a different model, and it is worth saying precisely what it
costs and what it buys, because the two go in opposite directions.

It costs feature (F).
With a cap on side payments the payoff primitive can no longer be a partition function: one must
specify raw payoffs $r_i(P)$ player by player, together with a feasible set of transfers, and a
coalition then bargains over a truncated utility possibility set rather than over a simplex.
If some player's raw payoff falls steeply when everything merges, and the cap prevents the others
from compensating her, the merge-all move is simply not Pareto-improving, and
Lemma~\ref{lem:pareto} fails.
With it fail Theorem~\ref{thm:absorbing-grand} and its corollary: there can now be absorbing states
that are not grand, so the dichotomy of Corollary~\ref{cor:dichotomy} becomes a trichotomy---settle
at a grand state, settle at a non-grand state, or cycle---and the second branch is a form of failure
this paper does not otherwise contain.
Step~4 of Theorem~\ref{thm:smalldelta} goes the same way, and so does everything resting on
Proposition~\ref{prop:farsighted}, since the far-sighted rule presupposes that its prescribed shares
can be paid.

What survives is more than one might expect.
Both persistence theorems do: Theorem~\ref{thm:persist-unan} uses only the welfare bound and the fact
that a grand agreement is destroyed by any termination, and Theorem~\ref{thm:persist-unil} uses only
the shadow lemma and the definition of domination.
Neither mentions transfers.
Lemmas~\ref{lem:reversal}, \ref{lem:shadow}, \ref{lem:entry} and~\ref{lem:objection} are likewise
structural.
And Theorem~\ref{thm:arrival} survives, because its hypothesis \eqref{eq:prefer} is a condition on
static payoffs that one can simply impose on the bounded-transfer payoffs directly.

What it buys is the point the counterexample makes vivid.
Merge-all fails there not because anybody dislikes the grand coalition but because it is
\emph{dominated}: at $x_1$ the pair $\{A,B\}$ prefers to form on its own first, which under the
sharing rule pays them $120$ each against $4$ in the grand state reached directly.
That $120$ is a transfer-heavy outcome; it is the pair capturing the whole of a surplus of $240$
measured against a reference state in which they had nothing.
Cap the transfers and that deviation may cease to be profitable for both of them, in which case it no
longer dominates, and the merge goes through.
So bounded transfers can destroy the profitability of the merge and destroy the profitability of the
deviations that block it, and which effect dominates is not something we can settle here.
The natural conjecture---that a tight cap favours arrival, because blocking coalitions rely on large
redistributions more than the grand coalition does---is plausible, and the counterexample is the
place to test it: the grand state $[ABCD]$ pays $(4,4,116,132)$ against raw payoffs summing to $256$,
while $\{A,B\}$ forming alone pays $(120,120,0,0)$ against a partition worth $240$ to that pair
alone.
Which of the two needs the larger transfer depends on the raw payoffs, which the partition function
does not record.

\paragraph{Coalitions as coherent players.}
Underneath every result in this paper is a modelling decision that deserves separate scrutiny: that
once a coalition has signed, it acts as a single rational player maximising its joint payoff within
the domain the agreement covers.
The alternative is to keep every action individual and to model the agreement as at most a
correlating device, with compliance sustained by punishment.
The choice is not innocent, and this paper is in a position to say something about what it buys and
what it costs.

\emph{Tractability} is the clearest gain, and it is a gain of kind rather than of degree.
The reduction turns the object of study into a finite state space with a payoff attached to each
state, and an equilibrium into a fixed point on that space.
The non-cooperative alternative is a repeated game whose strategies are functions of the entire
history; the analysis then concerns which payoff vectors can be supported, not which structures
arise, and the two questions of this paper---does the grand coalition survive, does it form---would
have no direct counterpart.
The gain is real but bounded: even after the reduction the state space grows faster than
exponentially, from $52$ states at $n=4$ to $41{,}106$ at $n=6$ \citep{HeitzigKornek2018}, which is
why the computations here stop at four players.

\emph{Equilibrium multiplicity} is reduced but not removed, and the residual matters.
Repeated-game treatments of the same interaction typically admit large sets of equilibrium payoffs
once players are patient, so they predict little; the coalitional reduction replaces that
indeterminacy with a specific bargaining rule and a specific move structure, and buys sharp
predictions in exchange.
But it does not buy uniqueness.
\citet{HeitzigKornek2018} report several equilibrium processes at intermediate discount factors, and
our own results show that what survives can be qualitatively decisive rather than a matter of detail:
at the payoff structure of Proposition~\ref{prop:counterexample} we found only cycling equilibria and
could not determine whether an arriving one exists, so the question ``does the grand coalition form''
may have no answer that is a function of the payoffs alone.
That is a limitation the reduction does not remove, and it inherits the coordination problem in a
smaller but still awkward form.

\emph{Efficiency} is decomposed cleanly, which is a genuine analytical benefit.
Within a coalition, joint maximisation is assumed, so all residual inefficiency lies between
top-level coalitions and is measured by a single number, the gap $\gamma$ of
Assumption~\ref{ass:eff}; the arrival question is exactly whether that gap is eventually closed.
In a fully non-cooperative treatment within-coalition and between-coalition inefficiency would be
entangled and no such decomposition would be available.
The price is that the grand coalition becomes the efficient benchmark by construction, so the
substantive content shifts entirely to whether it forms---which is why a paper of this kind can be
written at all, and also why its efficiency conclusions should not be read as welfare conclusions
(see the caveat below).

\emph{Enforcement} is the assumption's weakest point, because it is outsourced rather than modelled.
The model represents exit---indeed the distinction between unanimous and unilateral termination is
one of its two axes---but it does not represent \emph{under-compliance}: a member who cannot leave an
agreement can still fail to honour it, and nothing in the state space records that.
A non-cooperative treatment would have to derive compliance from punishment strategies, which is
where much of the difficulty of such models lies.
The reduction is therefore appropriate exactly where enforcement is credible, and the artificial-agent
setting discussed next is of interest partly because it is such a case.

\emph{Stability of the fiction} is the subtler point.
A coalition treated as a coherent player has a joint payoff but no internal politics: the sharing rule
fixes the division once, at signing, against the state that then obtained, and it is never revisited.
Our results show that this is precisely where the action is.
The relative shares $\psi$ are constant-sum (Lemma~\ref{lem:psi}), so all conflict is distributional;
the free-riding manoeuvre of Section~\ref{sec:onestep} is a play for a better division; and the chain
rule, together with the exclusion of internal renegotiation in
Remark~\ref{rem:variants}, is exactly the modelling choice that keeps a signed division frozen.
Permitting a standing coalition to renegotiate its internal split would require a richer state space
and would, we suspect, change the arrival results materially.
Of the assumptions behind the coherent-player reduction, this is the one we would least like to
defend and the first we would relax.

\paragraph{Coalition formation between artificial agents.}
If the negotiators are software agents, several of the modelling choices above become design choices,
and the results read as advice.

\emph{Divide long-term surplus, not per-period surplus.}
The sharing rule is the pivot (Section~\ref{sec:pivot}), and the far-sighted variant is the better
one on every count we can measure: it makes merging Pareto-improving in evaluations rather than only
in static payoffs (Proposition~\ref{prop:farsighted}), it turns $\delta\to1$ from the intractable
case into the tractable one (Theorem~\ref{thm:fs-arrival}), and under long-run gains it delivers the
grand coalition within $n-1$ moves under either termination rule (Theorem~\ref{thm:fgains}).
An agent architecture that computes the surplus of a proposed agreement as a difference of expected
discounted continuation values, rather than of immediate payoffs, is therefore not merely more
sophisticated but more likely to reach efficiency.
Two caveats.
The rule is implicit---payoffs and the process must be solved for jointly---so it is computationally
heavier, and its transfers grow like $1/(1-\delta)$ (Proposition~\ref{prop:transfer-growth}), which
interacts badly with the bounded-transfer discussion above.
And it is not sufficient: the counterexample survives it under unilateral termination.

\emph{Make agreements unanimously terminable if you can.}
Every result is at least as strong under unanimous termination, and two are strictly stronger:
all grand states are absorbing rather than only the reached ones
(Theorem~\ref{thm:persist-unan} against Theorem~\ref{thm:persist-unil}), and the $\delta\to1$
argument of Theorem~\ref{thm:fs-arrival} needs it.
The residual scenario that Proposition~\ref{prop:arrival-unil} cannot exclude is exactly a unilateral
exit taken off the equilibrium path.

Is that realistic between agents?
More so than between states, which is the reverse of what Section~\ref{sec:realism} concluded about
the parties themselves.
The reason unilateral termination is realistic among sovereigns is that no external mechanism can
keep a party in an agreement it wishes to leave.
Among software agents such mechanisms exist and are cheap: collateral held in escrow, staked bonds
forfeited on exit, agreements executed by a shared protocol rather than by each party severally.
These do not make exit impossible, but they make it require the others' cooperation, which is what
the unanimous variant models.
So this is a case where the institutional design our results favour is available precisely where the
negotiation is fastest---and, by Remark~\ref{rem:continuous-time}, fast negotiation is what pushes
$\delta$ towards one, where the design matters most.

Two warnings against reading this as an unqualified recommendation.
Our model has no private information, no mistakes and no change in the environment, so an agreement
that is hard to leave is never a trap; irrevocability is costless here in a way it certainly is not
in practice, and a mechanism that makes exit require unanimity also makes entry a decision that
cannot be revisited.
And the model's efficiency criterion is the total payoff of the players in it.
Agents negotiating on behalf of principals may form a grand coalition that is efficient for
themselves and a cartel from anyone else's point of view; nothing in this paper distinguishes the
two, and the Cournot specification of Section~\ref{sec:numerics} is a case where the grand coalition
is exactly the outcome one would not want.

\paragraph{Open problems.}
\begin{enumerate}[label=(\arabic*)]
\item Prove Theorem~\ref{thm:arrival} under unilateral termination without hypothesis (a), i.e.\
remove the gap of Remark~\ref{rem:open-unil}.
\item Decide whether arrival can fail as $\delta \to 1$ under the sharing rule of Assumption~\ref{ass:share}.
Proposition~\ref{prop:counterexample} settles the question for $\delta \le 3/4$,
but its payoff structure yields grand absorption for $\delta \ge 0.8$,
and no counterexample is known for large $\delta$.
This is the main open problem.
Theorem~\ref{thm:fs-arrival} answers the corresponding question for far-sighted sharing in the negative,
so an affirmative answer here would have to exploit the term that distinguishes the two rules,
namely the persistence of the relative shares $\psi$ in the merge gain.
\item Show that the transfers implementing far-sighted sharing stay bounded as $\delta\to1$,
or exhibit a payoff structure where they do not,
so as to remove or confirm the hypothesis $\Sigma$ of Corollary~\ref{cor:fs-arrival}.
\item Extend Proposition~\ref{prop:rotating} to non-deterministic classes and to approximate
symmetry.
\item Analyse the individual-withdrawal variant of Remark~\ref{rem:variants}(ii), where
Lemma~\ref{lem:pareto} does not apply to the deviation.
\item Check condition \eqref{eq:prefer} numerically in a calibrated setting; if it holds, a positive
result there becomes a corollary of Theorem~\ref{thm:arrival} rather than a numerical observation.
\item Obtain \ref{E3} from a fixed-point argument, or show that it cannot be so obtained
(Appendix~\ref{app:existence}).
\end{enumerate}

\backmatter

\bmhead{Acknowledgements}
The analysis in this paper was carried out with substantial assistance from a large language model
(Claude, Anthropic), which was used to generate proof strategies, candidate results, counterexamples
and numerical code, and to draft the exposition.
All statements and proofs were subsequently checked by the author, who takes full responsibility for
the correctness of the contents.
The numerical results of Sections~\ref{sec:numerics} and~\ref{sec:counterexample} were verified
independently of the analytical arguments, as described there; the code is available.

\section*{Declarations}

\textbf{Conflict of interest.} The author declares no competing interests.

\textbf{Data and code availability.} All data reported are generated by the code accompanying this
paper; no external data were used.

\textbf{Use of AI.} See the Acknowledgements. No AI system is listed as an author.

\begin{appendices}
\section{Two subsidiary results}\label{app:subsidiary}

Neither result below is used anywhere in the paper.
The first records that the sharing rule, imposed at the top level only, is consistent with itself all
the way down the hierarchy; the second is a sufficient condition, at small discount factors, for the
one grand state in which everybody signs simultaneously.

\begin{lemma}[The sharing rule propagates to all nodes]\label{lem:sharing-all}
Under Condition~\ref{cond:nash}, for every state $x$, every node $K \in x$ and every $i \in K$,
\[
\pi_i(x) = \pi_i(x^{-K}) + \frac{w_i}{w(K)}\,\Delta(x,K).
\]
\end{lemma}

\begin{proof}
For a node $K$ of a state $y$ let $d(y,K) := |\{K' \in y : K \subsetneq K'\}|$ be its depth,
so $d(y,K)=0$ iff $K$ is top-level.
We prove by induction on $d$ the statement $P(d)$:
the display holds for every state $y$ and every node $K \in y$ with $d(y,K)=d$.

$P(0)$ is Condition~\ref{cond:nash}.
Let $d\ge1$ and assume $P(d-1)$.
Let $K \in x$ with $d(x,K)=d$ and let $S \in x$ be the parent of $K$,
i.e.\ the smallest node of $x$ strictly containing $K$; then $d(x,S)=d-1$.
Note that $S$ is a node of $x^{-K}$ of the same depth $d-1$, that $K$ is a node of $x^{-S}$ of depth $d-1$,
and that $(x^{-K})^{-S}=(x^{-S})^{-K}=:x^{-S,-K}$.
Applying $P(d-1)$ to $S$ in the states $x$ and $x^{-K}$, and to $K$ in the state $x^{-S}$, we get,
for all $i \in K \subseteq S$:
\begin{align*}
 \pi_i(x) &= \pi_i(x^{-S}) + \frac{w_i}{w(S)}\Delta(x,S),\\
 \pi_i(x^{-K}) &= \pi_i(x^{-S,-K}) + \frac{w_i}{w(S)}\Delta(x^{-K},S),\\
 \pi_i(x^{-S}) &= \pi_i(x^{-S,-K}) + \frac{w_i}{w(K)}\Delta(x^{-S},K).
\end{align*}
Subtracting the second from the first and inserting the third,
\begin{equation}\label{eq:share-step}
 \pi_i(x)-\pi_i(x^{-K}) = \frac{w_i}{w(K)}\Delta(x^{-S},K)
   + \frac{w_i}{w(S)}\Big(\Delta(x,S)-\Delta(x^{-K},S)\Big) \qquad\text{for all } i \in K.
\end{equation}
Summing \eqref{eq:share-step} over $i \in K$ gives, by the definition of $\Delta(x,K)$,
\[
\Delta(x,K) = \Delta(x^{-S},K) + \frac{w(K)}{w(S)}\Big(\Delta(x,S)-\Delta(x^{-K},S)\Big),
\]
and substituting this back into \eqref{eq:share-step} yields $\pi_i(x)-\pi_i(x^{-K}) = \frac{w_i}{w(K)}\Delta(x,K)$,
which is $P(d)$.
\end{proof}

\begin{remark}
Lemma~\ref{lem:sharing-all} is a statement about the \emph{recursive} character of the specification:
payoffs at $x$ are obtained from payoffs at the coarser state $x^{-K}$,
whose payoffs are in turn obtained by the same rule,
and the induction shows that the resulting numbers are consistent with the proportional formula at every level.
It is used nowhere below---every application is to a top-level node---but it means that no harm is done by quoting the formula in the general form.
It does require Condition~\ref{cond:nash}:
the proof compares the shares of $S$ in the two states $x$ and $x^{-K}$,
and state-dependent shares need not agree there.
\end{remark}

\subsection*{A threshold for the flat grand state}

Let $x_0:=\emptyset$ be the fully non-cooperative state and $h_0 := g_{x_0} = \{N\}$ the \emph{flat} grand state,
in which all players sign one agreement simultaneously.
Write
\begin{align*}
G_i &:= \pi_i(h_0)-\pi_i(x_0) &&\text{(what $i$ gains by merging at all)},\\ B_i &:= \max_{h\in\G}\pi_i(h)-\pi_i(h_0) &&\text{(what $i$ could gain by merging later)} ,
\end{align*}
both non-negative, the first by Lemma~\ref{lem:pareto} and the second by definition.
Under Assumption~\ref{ass:share},
$G_i = w_i\Delta_0$ with $\Delta_0 := V^\ast-V(x_0)$ and $B_i = \psi_i^{\max}-\psi_i(x_0)$ with $\psi_i^{\max}:=\max_{u\in\X\setminus\G}\psi_i(u)$.

\begin{proposition}\label{prop:flat}
Let $p$ satisfy \ref{E1}.
If $\delta B_i\le(1-\delta)G_i$ for every $i \in N$, then $h_0$ is absorbing.
Under Condition~\ref{cond:nash} such a $\delta$ exists only if $\delta \le \Delta_0/(\Delta_0+\sum_i\psi_i^{\max})$.
\end{proposition}

\begin{proof}
$h_0$ has $N$ as its only node,
so by Lemma~\ref{lem:entry} every moving move at $h_0$ terminates $N$ and has target $x_0$.
If one is realised, with $q:=p_{h_0\to x_0}>0$, then by Lemma~\ref{lem:onestep}(a) applied at $h_0$,
profitability for its terminator $i$, $\ev_i(x_0)>\ev_i(h_0)$, is equivalent to $\ev_i(x_0)>\pi_i(h_0)$.
By Lemma~\ref{lem:value},
splitting off the first period (spent at $x_0$) and bounding all later periods by $\max_{v\in\X}\pi_i(v)=\max_{h\in\G}\pi_i(h)$ (the maximum of $\pi_i$ is attained on $\G$,
by Lemma~\ref{lem:pareto}),
\[
\ev_i(x_0)\ \le\ (1-\delta)\pi_i(x_0)+\delta\max_{h\in\G}\pi_i(h) \ =\ \pi_i(h_0) - (1-\delta)G_i + \delta B_i,
\]
which is $\le\pi_i(h_0)$ exactly when $\delta B_i\le(1-\delta)G_i$ --- a contradiction.
Under Condition~\ref{cond:nash} we may substitute $G_i=w_i\Delta_0$ and $B_i=\psi_i^{\max}-\psi_i(x_0)$;
summing over $i$ and using $\sum_i\psi_i(x_0)=0$ and $\sum_iw_i=1$ gives $\delta\sum_i\psi_i^{\max}\le(1-\delta)\Delta_0$.
\end{proof}

\begin{remark}
In Example~\ref{ex:sym}, $G_i=2$ and $B_i=1$,
so Proposition~\ref{prop:flat} gives absorption of $\{N\}$ for all $\delta\le2/3$;
equivalently $\psi_i^{\max}=1$, $\psi_i(x_0)=0$, $w_i=1/3$, $\Delta_0=6$.
The computations reported by \citet{HeitzigKornek2018} show that $\{N\}$ is absorbing there also at $\delta=0.9$,
so the bound is sufficient and far from necessary.
\end{remark}

\section{Existence of equilibrium}\label{app:existence}

Nothing in this note requires existence:
every statement is of the form ``every equilibrium has property $P$''.
Still, one wants to know that the class is not empty.
Since Definition~\ref{def:rationality} measures profitability against $\evp_i(x)$ by a weak inequality, that clause is already closed, which was not so when profitability was measured against $\ev_i(x)$.
What remains strict is domination---and with it, through the stay moves, the requirement that every relevant player gain strictly over the status quo, which is what the old profitability clause used to carry.
So the correspondence taking evaluations to admissible processes is still not upper hemicontinuous, and the standard remedy applies: weaken the strict inequalities on the ``forbidding'' side and keep them on the ``requiring'' side, so that all conditions become closed.
We do that here.
The construction below is stated for the old primitives and has not been restated for Definition~\ref{def:rationality};
what changes is that $\M(x)$ becomes $\M^{+}(x)$, that weak profitability is replaced by profitability itself, which needs no weakening, and that the conditions now depend on $p$ as well as on $\ev$, through $V$, which is continuous in $p$ and so costs nothing.
The strictness that blocks existence of exact equilibria is unchanged in kind, having moved from the profitability clause to the stay-domination clause.

\begin{definition}[Weak equilibrium]\label{def:weakeq}
Let $\ev$ be given.
Call $m \in \M(x)$ \emph{weakly profitable} if $\ev_i(\tau(m)) \ge \ev_i(x)$ for all $i \in R(m)$,
and \emph{weakly dominated} if there is $m'\in\M(x)$ with $R(m')\subseteq R(m)$ and $\ev_i(\tau(m'))\ge\ev_i(\tau(m))$ for all $i \in R(m')$;
and let $A(x,\ev)$ be the set of weakly profitable moves that are not (strictly) dominated in the sense of Definition~\ref{def:rationality}(b).
A process $p$ is a \emph{weak equilibrium} if, with $\ev$ from \eqref{eq:bellman},
for every $x$ the support of $p_{x\to\cdot}$ is contained in
\[
T(x,\ev) \;:=\; \{\tau(m) : m \in A(x,\ev)\} \;\cup\;
   \begin{cases} \{x\} & \text{if every $m\in\M(x)$ is weakly unprofitable or weakly dominated},\\
                 \emptyset & \text{otherwise.}\end{cases}
\]
\end{definition}

\begin{lemma}\label{lem:Tnonempty}
$T(x,\ev)\ne\emptyset$ for every $x$ and $\ev$,
and the correspondence $\ev \mapsto T(x,\ev)$ has a closed graph.
\end{lemma}

\begin{proof}
\emph{Non-emptiness.} Let $W$ be the set of weakly profitable moves at $x$.
If $W=\emptyset$ then every move is weakly unprofitable and $x \in T(x,\ev)$.
If $W\ne\emptyset$, note that strict domination maps $W$ into itself:
if $m \in W$ is dominated by $m'$ then for $i \in R(m')\subseteq R(m)$ we get $\ev_i(\tau(m'))>\ev_i(\tau(m))\ge\ev_i(x)$,
so $m'\in W$.
Strict domination is acyclic on the finite set $\M(x)$ by the argument in the proof of Lemma~\ref{lem:domination},
so $W$ has a maximal element $m$, which is then undominated and lies in $A(x,\ev)$.

\emph{Closed graph.} Let $\ev^k\to\ev$ and $y^k \in T(x,\ev^k)$ with $y^k\to y$;
as $\X$ is finite we may assume $y^k=y$ for all $k$.
If $y=\tau(m)$ with $m \in A(x,\ev^k)$ for infinitely many $k$,
then weak profitability of $m$ passes to the limit (weak inequalities are closed) and $m$ is not strictly dominated at $\ev$ (if it were,
it would be strictly dominated at $\ev^k$ for large $k$, since strict domination is an open condition),
so $m \in A(x,\ev)$.
Otherwise $y=x$ and, for infinitely many $k$, every move is weakly unprofitable or weakly dominated at $\ev^k$;
both properties are defined by weak inequalities and are therefore preserved in the limit,
so the same holds at $\ev$ and $x \in T(x,\ev)$.
\end{proof}

\begin{theorem}[Existence]\label{thm:existence}
For every $\delta\in(0,1)$ a weak equilibrium exists.
If, at such a $p$,
no ties occur---i.e.\ $\ev_i(y)\ne\ev_i(z)$ whenever $i \in N$ and $y\ne z$---then $p$ satisfies \ref{E1} and \ref{E2} of Definition~\ref{def:eq}.
\end{theorem}

\begin{proof}
Let $\Pi := \prod_{x\in\X}\Delta(\X)$ be the (compact,
convex) set of stochastic matrices on $\X$ and $\Psi:\Pi\to\R^{N\times\X}$,
$\Psi(p) := \ev$ the map defined by \eqref{eq:bellman};
$\Psi$ is continuous because $\ev=(1-\delta)(I-\delta p)^{-1}\pi$ and matrix inversion is continuous on the open set where $I - \delta p$ is invertible,
which contains $\Pi$.
Define $\Phi(\ev) := \{p \in \Pi : \operatorname{supp}p_{x\to\cdot}\subseteq T(x,\ev) \text{ for all } x\}$.
By Lemma~\ref{lem:Tnonempty}, $\Phi(\ev)$ is non-empty; it is a product of faces of simplices,
hence convex and compact; and it has a closed graph,
because supports can only shrink under limits and $T$ has a closed graph.
Hence $\Phi\circ\Psi:\Pi\rightrightarrows\Pi$ is a non-empty, convex,
compact-valued correspondence with closed graph on a compact convex set,
and Kakutani's theorem \citep{Kakutani1941} provides $p \in \Phi(\Psi(p))$,
which is a weak equilibrium by definition.

Now suppose no ties occur at $p$.
Weak profitability then coincides with profitability and weak domination with domination,
so $A(x,\ev)=\Cset(x)$ and the alternative case in the definition of $T(x,\ev)$ says exactly $\Cset(x)=\emptyset$.
Thus the support condition gives \ref{E1},
and it gives \ref{E2} because $x \in T(x,\ev)$ is possible only when $\Cset(x)=\emptyset$.
\end{proof}

\begin{remark}[\ref{E3} is not delivered by this argument]\label{rem:E3-existence}
Axiom \ref{E3} requires that certain moves \emph{be} realised,
i.e.\ that the support of $p_{x\to\cdot}$ \emph{contain} a given set.
That is an inclusion in the direction opposite to the one preserved by limits:
along a converging sequence supports can only shrink,
so the set of processes satisfying \ref{E3} relative to a varying $\ev$ need not have a closed graph,
and Kakutani does not apply to it directly.
The model of \citet{HeitzigKornek2018} obtains \ref{E3} by construction,
by assigning to each favourite move a probability proportional to the bargaining weight of the players favouring it;
establishing existence for that explicit rule requires an argument we do not reproduce here.
Since Table~\ref{tab:map} shows that \ref{E3} is needed only for the results of Section~\ref{sec:arrival} that rest on Lemma~\ref{lem:objection},
Theorem~\ref{thm:existence} already guarantees that the persistence theory of Section~\ref{sec:persistence} is not vacuous.
\end{remark}

\section{Abel and Ces\`aro limits of evaluations}\label{app:abel}

This appendix proves the claim of Remark~\ref{rem:limit},
which is used nowhere in the note but explains why the exact identities of Section~\ref{sec:averages} are the right tools.

\begin{proposition}\label{prop:abel}
Let $p$ be a fixed stochastic matrix on the finite set $\X$ and let $p^\ast := \lim_{T\to\infty} \frac1T\sum_{t<T}p^t$ be its Ces\`aro limit,
which exists by the ergodic theorem for finite Markov chains.
Then, for $\ev^{(\delta)} := (1-\delta)\sum_{t\ge0}\delta^tp^t\pi_i$,
\[
   \lim_{\delta\uparrow1}\ev^{(\delta)} \;=\; p^\ast\pi_i .
\]
In particular, if $C$ is a closed communicating class with stationary distribution $\mu$,
then $\ev^{(\delta)}_i(x) \to \bar a_i = \sum_{y\in C}\mu(y)\pi_i(y)$ for every $x \in C$: in the limit,
and only in the limit, the evaluation becomes constant on $C$ and equal to the long-run average.
\end{proposition}

\begin{proof}
Write $S_T := \sum_{t<T}p^t = T\,A_T$ with $A_T \to p^\ast$.
Summation by parts gives, for every $M$,
\[
  \sum_{t=0}^{M}\delta^t p^t = \sum_{t=0}^{M}\delta^t(S_{t+1}-S_t)
  = \delta^{M}S_{M+1} + (1-\delta)\sum_{t=0}^{M-1}\delta^tS_{t+1},
\]
and since $\|S_{M+1}\|_\infty \le M+1$ while $\delta<1$,
letting $M\to\infty$ yields $\sum_{t\ge0}\delta^tp^t = (1-\delta)\sum_{T\ge1}\delta^{T-1}S_T$.
Hence
\[
  (1-\delta)\sum_{t\ge0}\delta^tp^t \;=\; \sum_{T\ge1}\lambda_T(\delta)\,A_T,
  \qquad \lambda_T(\delta) := (1-\delta)^2T\delta^{T-1},
\]
and $\sum_{T\ge1}\lambda_T(\delta) = (1-\delta)^2\sum_{T\ge1}T\delta^{T-1} = 1$,
so the right-hand side is a probability average of the $A_T$.
Fix $\varepsilon>0$ and $T_0$ with $\|A_T-p^\ast\|_\infty<\varepsilon$ for $T\ge T_0$.
Since $\sum_{T<T_0}\lambda_T(\delta)\to0$ as $\delta\uparrow1$ and $\|A_T\|_\infty\le 1$,
we get $\limsup_{\delta\uparrow1}\|(1-\delta)\sum_t\delta^tp^t - p^\ast\|_\infty \le \varepsilon$ for every $\varepsilon>0$.
Applying this to $\pi_i$ gives the first claim.
For the second, $p^\ast$ restricted to $C$ has all rows equal to $\mu$.
\end{proof}
\end{appendices}

\section*{Note on this revision}

This version changes the model's primitive notion of rationality and adds the bargaining game that
produces it, and the change is recent enough that its consequences have not been traced through every
proof.

Two things are new in Definition~\ref{def:rationality}.
Staying is a move: each player $i$ has $s^i_x$ with target $x$ and relevant set $\{i\}$, so the status
quo enters domination and the favourite comparison on the same footing as everything else, with ties
broken towards $x$.
And profitability compares $\ev_i(\tau(m))$ with $\evp_i(x)=\sum_y p_{x\to y}\ev_i(y)$ rather than with
$\ev_i(x)$, because a player who declines a move does not thereby spend a period at $x$;
the two benchmarks differ by one period of flow, as \eqref{eq:anchor} records, and coincide at
absorbing states.

Subsection~\ref{sec:stay} then proves implementation.
Theorem~\ref{thm:implement} shows that limit points of stationary subgame-perfect equilibria of the
protocol $P_\varepsilon$ as $\varepsilon\to0$ satisfy \ref{E1}--\ref{E3}, and
Proposition~\ref{prop:proposals} that agreement is reached in the first round, so that the delay is a
threat that is never paid on the equilibrium path.
The one selection the proof needs is that a proposer with several tied favourites randomises over
them, which Criterion~\ref{cond:ties} makes vacuous.
Both directions of the earlier draft's conjecture are not claimed: the converse, that every
equilibrium arises as such a limit, is untouched.

Re-checked against the new definitions, and standing:
Lemma~\ref{lem:domination}, whose chaining argument needs only that $\succ^{x}_i$ is transitive and
that the profitability benchmark is one number per player and state;
Lemma~\ref{lem:nonstalling}, where the strict gain of a relevant player now comes from the absence of
stay domination;
Theorem~\ref{thm:persist-unan}, whose Bellman computation is unchanged, the inequality
$\ev_i(y)>\ev_i(h)$ for all $i$ now following from undominatedness against the stays rather than from
profitability;
Theorem~\ref{thm:persist-unil}, which invokes no payoff facts at all;
Lemma~\ref{lem:objection}, whose first case now runs through $\evp_i(x)$ and yields a strict inequality
where it used to yield a weak one, so that Theorem~\ref{thm:arrival} and the rest of
Section~\ref{sec:arrival} resting on it stand;
and Proposition~\ref{prop:farsighted}, where the step from the merge gain to profitability goes
through the observation that $\evp_i(x)=\ev_i(x)$ at an absorbing $x$.

Also re-checked, and standing after repair: Theorem~\ref{thm:onestep}, whose exclusions now run
through stay domination and where merge-all is profitable automatically, the process being
deterministic there; and Theorem~\ref{thm:fgains} with Corollary~\ref{cor:fgains-necessary}, where the
same substitution applies and where absorption at a grand state now follows from \ref{E1} rather than
\ref{E2}.

Not standing: the construction behind Proposition~\ref{prop:counterexample}.
Theorem~\ref{thm:nodetcycle}, which is new, shows that under unanimous termination no equilibrium
traverses a deterministic cycle at all, so the cycle exhibited there is not one;
Remark~\ref{rem:cyclestatus} records what this leaves open.

Not yet re-checked, and to be treated as provisional until they are:
the remaining arguments of Sections~\ref{sec:cycles}--\ref{sec:pivot} that turn on a move being
unprofitable, since non-profitability is now $\ev_i(\tau(m))<\evp_i(x)$ for some relevant $i$ and
neither implies nor is implied by the former condition;
the numerical work of Subsection~\ref{sec:numerics}, which searched over payoffs with the process held
fixed and would now have to search over processes as well;
and Appendix~\ref{app:existence}, whose construction has not been restated, though its opening now
records what would change and where the obstacle to existence has moved.

\bibliographystyle{sn-mathphys-ay}

\bibliography{refs}

\end{document}